\documentclass[11pt,letterpaper]{article}
\usepackage{hyperref}
\hypersetup{urlcolor=blue, colorlinks=true, citecolor=green!50!black, linkcolor=blue}
\usepackage{geometry}
\usepackage[utf8]{inputenc}
\usepackage[american]{babel}
\usepackage{twoopt}
\usepackage{tikz}
\usetikzlibrary{fit}

\usepackage[normalem]{ulem}

\usepackage{graphicx}
\graphicspath{{images/}{../images/}}

\usepackage{amsmath, amssymb, cases}
\usepackage{theoremstyles}
\usepackage{cleveref}
\usepackage{enumitem}
\usepackage{mdframed}
\usepackage{bbm}
\usepackage{bm}

\usepackage[ruled]{algorithm}
\usepackage[noend]{algpseudocode}

\usepackage{microtype}

\usepackage{mdframed}
\usepackage{xcolor}

\usepackage{modletters}

\usepackage{mathtools}

\newcommand{\Interest}{Interesting\ }

\newcommand{\bsni}{\bigskip\noindent}

\newcommand\restr[2]{{
  \left.\kern-\nulldelimiterspace 
  #1 
  \vphantom{\big|} 
  \right|_{#2} 
  }}

\newcommand{\eps}{\varepsilon}

\makeatletter
\def\moverlay{\mathpalette\mov@rlay}
\def\mov@rlay#1#2{\leavevmode\vtop{%
   \baselineskip\z@skip \lineskiplimit-\maxdimen
   \ialign{\hfil$\m@th#1##$\hfil\cr#2\crcr}}}
\newcommand{\charfusion}[3][\mathord]{
    #1{\ifx#1\mathop\vphantom{#2}\fi
        \mathpalette\mov@rlay{#2\cr#3}
      }
    \ifx#1\mathop\expandafter\displaylimits\fi}
\makeatother

\usepackage{graphicx}

\newcommand{\poly}{\mathrm{poly}}

\newcommand{\floor}[1]{\lfloor #1 \rfloor}

\newcommand{\size}[1]{\mathrm{size}}

\newcommand{\set}[2][ ]{\{#2 \ifthenelse{\equal{#1}{ }}{ }{~|~#1}\}}

\newcommand{\comment}[1]{}

\newcommand{\seepage}[2][See]{
    \marginnote{
        \scriptsize {#1} p.~\pageref{#2}
    }
}

\newcommand{\reuse}[1]{
	\expandafter\stepcounter{#1_help}
    \expandafter\label{#1_app}
    \csname#1\endcsname*
}
\crefname{claim}{claim}{claims}

\newcommandtwoopt\Textbox[5][0cm][1cm]{%
\begin{tikzpicture}[remember picture,overlay]
  \coordinate (aux) at ([xshift=#1]#4);
  \node[inner ysep=3pt,yshift=0.6ex,draw=green!50!black,thick,
    fit=(#3) (aux),baseline] 
    (box) {};
  \node[text width=#2,anchor=north east,
    font=\sffamily\footnotesize,align=right] 
    at (box.north east) {#5};
\end{tikzpicture}%
}

\newtheorem*{theorem*}{Theorem}
\usepackage{caption}
\usepackage{subcaption}

\usepackage[utf8]{inputenc}
\usepackage{tikz}
\newcommand{\remove}[1]{}

\newcommand{\dw}[1]{#1^{\downarrow}}
\newcommand{\ShowComment}{false}
\ifdefined\ShowComment
\newcommand{\mtodo}[1]{\textcolor{blue}{[\textbf{Michal}: #1]}}
\else
\newcommand{\mtodo}[1]{}
\usepackage{todonotes}
\fi
\newcommand{\potintnum}{O(\log n)}
\newcommand{\info}{\mathsf{info}}
\usepackage{mathtools}

\usepackage{xcolor}
\usepackage{xspace}

\newcommand{\lca}[1]{\text{LCA}(#1)}
\ifdefined\ShowComment
\newcommand{\danupon}[1]{{\color{red} DANUPON: #1}}
\else
\newcommand{\danupon}[1]{}
\usepackage{todonotes}
\fi

\renewcommand{\paragraph}[1]{\medskip\noindent{\bf #1}\xspace}

\newcommand{\cov}{\mathsf{Cov}}
\newcommand{\extcov}{\mathsf{Cov}^\mathsf{extr}}
\newcommand{\ECov}{\mathsf{CovSet}}
\newcommand{\covs}{\mathsf{Cov}_{\textsf{samp}}}
\newcommand{\ECovs}{\mathsf{CovSet}_{\textsf{samp}}}
\newcommand{\cut}{\mathsf{Cut}}
\newcommand{\ecut}{\mathsf{CutSet}}
\newcommand{\dfrag}{D_{\sf frag}}
\newcommand{\nfrag}{N_{\sf frag}}
\newcommand{\sfrag}{S_{\sf frag}}
\newcommand{\interesting}{\mathsf{Int}}
\newcommand{\expt}[2]{\underset{#1}{\mathbb{E}}\left[#2\right]}

\ifdefined\ShowComment
\newcommand{\sagnik}[1]{{\bf \color{green!50!blue} [SAGNIK: #1]}}
\else
\newcommand{\sagnik}[1]{}
\usepackage{todonotes}
\fi

\ifdefined\ShowComment
\newcommand{\yuval}[1]{{\bf \color{red!50!black} YUVAL: #1}}
\else
\newcommand{\yuval}[1]{}
\fi

\newcommand{\high}{P_{H_1}}
\newcommand{\way}{P_{H_2}}

\newcommand{\intpot}[1]{\mathsf{Int}_{\textsf{pot}}(#1)}

\usepackage{authblk}

\renewcommand{\cT}{{\mathcal{T}}}

\title{Distributed Weighted Min-Cut in Nearly-Optimal Time}

\author[1]{Michal Dory\thanks{\texttt{smichald@cs.technion.ac.il}.}}
\author[2]{Yuval Efron\thanks{\texttt{efronyuv@gmail.com.} Most of this work was done while the author was affiliated with the Technion.}}
\author[3]{Sagnik Mukhopadhyay\thanks{\texttt{sagnik@kth.se}}}
\author[3]{Danupon Nanongkai\thanks{\texttt{danupon@kth.se}}}
\affil[1]{Technion, Israel}
\affil[2]{University of Toronto, Canada}
\affil[3]{KTH Royal Institute of Technology, Sweden}
\date{}
\begin{document}

\begin{titlepage}
	\maketitle
	\pagenumbering{roman}


\begin{abstract}
Minimum-weight cut (min-cut) is a basic measure of a network's connectivity strength. While the min-cut can be computed efficiently in the sequential setting [Karger STOC'96], there was no efficient way for a distributed network to compute its own min-cut without limiting the input structure or dropping the output quality: In the standard CONGEST model, existing algorithms with nearly-optimal time (e.g. [Ghaffari, Kuhn, DISC'13; Nanongkai, Su, DISC'14]) can guarantee a solution that is $(1+\epsilon)$-approximation at best while the exact $\tilde O(n^{0.8}D^{0.2} + n^{0.9})$-time algorithm [Ghaffari, Nowicki, Thorup, SODA'20] works only on {\em simple} networks (no weights and no parallel edges).\footnote{Throughout, $n$ and $D$ denote the network's number of vertices and hop-diameter, respectively.} For the weighted case, the best bound was $\tilde O(n)$ [Daga, Henzinger, Nanongkai, Saranurak, STOC'19]. 
	
In this paper, we provide an {\em exact} $\tilde O(\sqrt n + D)$-time algorithm for computing min-cut on {\em weighted} networks. Our result improves even the previous algorithm that works only on simple networks. Its time complexity matches the known lower bound up to polylogarithmic factors. 
At the heart of our algorithm are a clever routing trick and two structural lemmas regarding the structure of a minimum cut of a graph. These two structural lemmas considerably strengthen and generalize the framework of Mukhopadhyay-Nanongkai [STOC'20] and can be of independent interest.
\end{abstract}

	\newpage
	\setcounter{secnumdepth}{2}
	\setcounter{tocdepth}{2}
	\tableofcontents
\end{titlepage}

\newpage
\pagenumbering{arabic}

\section{Introduction}\label{sec:intro}

\paragraph{Min-cut.} Minimum cut (min-cut) is a 
basic mathematical concept that is of great importance from the network design perspective as it captures the connectivity of the network. 
%
Given a graph with $n$ vertices and $m$ (possibly weighted) edges, a cut is a set of edges removing which disconnects the graph, and the weight of the cut is the total weight of the edges participating in the cut.
%
%
%

In the sequential setting, a long line of work spanning over many decades since the 1950s \cite{EliasFS56,FordF87} was concluded by the STOC'95 $O(m\log^3 n)$-time randomized algorithm of Karger \cite{Karger00} (see \cite{mukhopadhyay2019weighted,GawrychowskiMW19} for recent improvements). 





\paragraph{Distributed min-cut.}
Efficient sequential algorithms, however, do not necessarily lead to an efficient way for a {\em distributed network} to compute the min-cut. The question of how a distributed network can compute its own min-cut has been actively studied in the CONGEST model of distributed networks (e.g. \cite{PritchardT11,Ghaffari-Kuhn,Nanongkai-Su,GhaffariH16,GhaffariN18,DagaHNS19,Parter19-smallcut,Ghaffari0T20}). 
%
%
In this model, a network is represented by an $n$-vertex unweighted graph of diameter $D$. Each edge $e$ is associated with weight $w(e)\in \set{1, 2, \ldots, \poly(n)}$ that does not affect the communication. In each communication round each vertex sends a message of $O(\log n)$ bits to each of its neighbors which arrives at the end of the round. The goal is to minimize the number of rounds to compute the value of the min-cut or to make every vertex realize which edges incident to it are in the min-cut (our and previous results can achieve both so we do not distinguish the two objectives in the discussion below).
Throughout, we use $\tO$ and $\tilde \Omega$ to hide $\poly\log(n)$ factors.

\paragraph{Previous works.}
%
For many graph problems in the CONGEST model such as minimum cut, minimum spanning tree, and single-source shortest paths, the ultimate goal is the $\tO(\sqrt n + D)$ round complexity. 
This is mainly because \cite{Das-Sharma} showed in STOC'11 an $\tilde{\Omega}(\sqrt n + D)$ lower bound for a number of fundamental graph problems, which holds even for $\poly(n)$-approximation algorithms (also see \cite{ElkinKNP14,KorKP13,Elkin06,PelegR00,Ghaffari-Kuhn}).
%
Since the work of \cite{Das-Sharma}, a lot of effort has been put on to match that lower bound by devising efficient algorithms for all of these problems, and for many of these problems, near-optimal upper bounds have been achieved (e.g. \cite{HenzingerKN-STOC16,BeckerKKL16,Nanongkai-STOC14,Nanongkai-Su,Ghaffari-Kuhn,GhaffariKKLP15}). For min-cut, the first algorithm towards this goal was by Ghaffari and Kuhn \cite{Ghaffari-Kuhn} which $(2+\epsilon)$-approximates the min-cut in  $\tO(\sqrt n + D)$ rounds. 
The approximation ratio was subsequently improved to $(1+\epsilon)$ by Nanongkai and Su \cite{Nanongkai-Su}. Obtaining efficient algorithms for the \emph{exact} case remained wide open.
%
%

Towards designing an efficient distributed algorithm for \textit{exact} min-cut, Daga, Henzinger, Nanongkai, and Saranurak~\cite{DagaHNS19} in STOC'19 gave the first algorithm that is {\em sublinear-time} ($\tilde O(n^{1-\epsilon}+D)$-time for some constant $\epsilon>0$). Their algorithm works on simple networks and guarantees $\tO(n^{1 - 1/353}D^{1/353} + n^{1 - 1/706})$ round complexity. This bound was recently improved in SODA'20 by Ghaffari, Nowicki and Thorup~\cite{Ghaffari0T20} to $\tO(n^{0.8}D^{0.2} + n^{0.9})$.\footnote{Prior to Daga-Henzinger-Nanongkai-Saranurak~\cite{DagaHNS19}, $O(D)$ bound was shown for finding min-cuts of values at most two by Prichard and Thurimella \cite{PritchardT11} 
%
and, later, $\tO(\sqrt n + D)$ bound was shown for finding min-cuts of values $O(\poly\log(n))$ by Nanongkai and Su \cite{Nanongkai-Su}. Parter \cite{Parter19-smallcut} recently improved the round complexity to $poly(D)$ when the min-cut has value $O(1)$, 
answering some open problems in \cite{DagaHNS19}. Additionally, distributed min-cut has been considered on fully-connected networks (congested clique) by Ghaffari and Nowicki \cite{GhaffariN18}.}
We emphasize that the algorithms of Daga-Henzinger-Nanongkai-Saranurak~\cite{DagaHNS19} and Ghaffari-Nowicki-Thorup~\cite{Ghaffari0T20} crucially exploit the fact that the network is {\em simple}, i.e. it is an unweighted graph without parallel edges. It is very unclear how to extend their techniques to work on even unweighted graphs with parallel edges. 
%
%
For exact min-cut on {\em weighted graphs}, the only known upper bound is an $\tO(n)$ one which follows from Daga-Henzinger-Nanongkai-Saranurak~(see Theorem~5.1, \cite{DagaHNS19}). To conclude, it was widely open whether exact min-cut on weighted graphs can be computed in sublinear time in $n$, and even in the simpler case of simple graphs, there was still a wide gap of at least $n^{0.4}$ between known upper and lower bounds. 

\setlength{\tabcolsep}{18pt}
\renewcommand{\arraystretch}{1.4}

\begin{table}[]
\centering
    \begin{tabular}{ |p{1.1cm}| c | c | c | }
 \hline
 \textbf{Authors} & \textbf{Variant} & \textbf{Approximation} & \textbf{Complexity}\\
 \hline
 \cite{Das-Sharma} & weighted & $n^c, c > 0$ & $\tilde{\Omega}(D + \sqrt n)$\\
 \cite{Ghaffari-Kuhn} & weighted & $(2+\epsilon)$ & $\tO(D + \sqrt n)$\\
  \cite{Nanongkai-Su} & weighted & $(1+\epsilon)$ & $\tO(D + \sqrt n)$\\
 \cite{DagaHNS19} & unweighted, simple & exact & $\tilde O(n^{1 - 1/353}D^{1/353}+ n^{1 - 1/706})$\\
 \cite{DagaHNS19} & weighted & exact & $\tilde O(n)$\\
 \cite{Ghaffari0T20} & unweighted, simple & exact & $\tilde O(n^{0.8}D^{0.2} + n^{0.9})$\\
 Here & weighted & exact & $\tilde O(D + \sqrt n)$\\
 \hline
\end{tabular}
\caption{\small Our results for distributed min-cut and comparison with other works. The $\tO(\cdot)$ notation hides polylogarithmic factors in $n$.} \label{tab:results-seq}
\end{table}

\paragraph{Our results.}
We present a randomized distributed algorithm that essentially resolves the distributed weighted exact min-cut problem (naturally, we also improve the upper bound of Ghaffari-Nowicki-Thorup~\cite{Ghaffari0T20} for simple graphs):



\begin{theorem}\label{thm:main}
In the CONGEST model, a min-cut of a (possible weighted) graph with $n$ vertices and with diameter $D$
can be found with high probability in $\tO(\sqrt n + D)$ rounds.\footnote{With high probability (w.h.p.) means with probability at least $1-1/n^c$ for an arbitrary constant $c$. ``Finding the min-cut'' refers to the standard definition where after the algorithm finishes every vertex knows the min-cut value and for every edge $\{u,v\}$ both $u$ and $v$ know whether edge $\{u,v\}$ is in the min-cut or not.}
\end{theorem}



At the heart of our algorithm is an algorithm for the {\em minimum 2-respecting cut problem}. In this problem, we are given a spanning tree $T$ of the graph $G$, and the goal is to find a minimum cut in $G$ which contains at most two edges from $T$ (such a cut is called a \textit{2-respecting} cut). The seminal work of Karger \cite{Karger00} showed that the min-cut problem can be reduced to solving the min 2-respecting cut problem, this reduction also holds in the CONGEST model (see e.g. \cite{DagaHNS19}). This approach led to efficient algorithms for min-cut in various settings \cite{Karger00,DagaHNS19,Nanongkai-Su,Thorup07,mukhopadhyay2019weighted,GawrychowskiMW19,geissmann2018parallel}. 
In the distributed setting,  Nanongkai and Su \cite{Nanongkai-Su} could solve in $\tO(\sqrt n + D)$ rounds the easier {\em minimum 1-respecting cut problem}, where the goal is to find the min-cut that contains one tree edge, leading to their $(1+\epsilon)$-approximation result. 
%
The minimum 2-respecting cut problem, however, turns out to be much more challenging to solve efficiently. Recently, Daga-Henzinger-Nanongkai-Saranurak~\cite{DagaHNS19} devised an algorithm to solve this in $\tO(n)$ rounds---this was one of the main ingredients for  Daga-Henzinger-Nanongkai-Saranurak~\cite{DagaHNS19} and Ghaffari-Nowicki-Thorup~\cite{Ghaffari0T20} for obtaining the aforementioned sublinear time exact algorithm on simple graphs. 
Our main technical contribution is an efficient distributed algorithm for the minimum $2$-respecting cut problem which is mentioned in the following theorem.

\begin{theorem}\label{thm:intro:2-respect}
In the CONGEST model, the $2$-respecting cut problem can be solved with high probability in $\tO(\sqrt n + D)$ rounds. 
\end{theorem}
Our result builds on the framework of Mukhopadhyay and Nanongkai \cite{mukhopadhyay2019weighted}. In \cite{mukhopadhyay2019weighted}, efficient minimum cut algorithms are presented in the sequential, streaming, and query settings. These algorithms follow the same framework for solving the $2$-respecting cut problem. Implementing this framework in the distributed setting is however much more challenging, due to the locality and edge congestion. 
The key insights that allow us to overcome these challenges consist of one trick and two structural lemmas: 
\begin{itemize}[noitemsep]
    \item {\em Interesting path counting lemma:} Given a spanning tree $T$    decomposed into paths $P_1, P_2, \ldots$ (using some tree decomposition techniques), this lemma identifies a small ($\text{poly}\log(n)$) number of paths that each $P_i$ is {\em interested in}; here, the notion of ``interested in'' is defined in such a way that the minimum 2-respecting cut corresponds to two edges in two paths $P$ and $P'$ that are interested in each other.
    %
   %
   This lemma allows us to reduce the 2-respecting cut problem into the same problem over the pairs of paths in $T$ that are interested in each other (where we are allowed to cut only edges in the given pair of paths). 
    It strengthens and helps simplify the framework of \cite{mukhopadhyay2019weighted}. 
    %
    It provides a novel observation on the structure of the minimum cut in a given weighted graph, and we hope that it has applications beyond the scope of this work.
    
    \item  {\em Short-paths routing trick:} 
    This is a basic building block used in many subroutines of our algorithm. 
    Given two paths $P$ and $P'$ of length $k$ in $T$, this trick implies that either there is some edge nearby these paths that they can use to communicate  $\tilde O(k)$ bits of information to each other, or they can simply broadcast a small message to the whole network to find their min 2-respecting cut. This trick allows us to solve the 2-respecting cut problem on the interesting pairs of short paths (in particular, when $k=\tilde O(\sqrt{n})$). It exploits some simple property of cuts that might be of independent interest. 

    \item {\em Path-partitioning lemma:} 
    This lemma together with a divide-and-conquer technique allows us to solve the 2-respecting cut problem on the interesting pairs of long paths.
    Roughly, the lemma states that to solve the problem above for path $P$ and $P'$, we can partition $P$ into sub-paths $P_1, P_2, \ldots$ and $P'$ into sub-paths $P'_1, P'_2, \ldots$ so that we only have to solve the same problem on each pair of paths $(P_i,P'_i)$.
    This lemma about the structure of the 2-respecting cut problem might be of independent interest. To prove it, we exploit the {\em monotonicity property} observed in \cite{mukhopadhyay2019weighted} in a new way. 
\end{itemize}

\noindent
We describe these key insights more in the next section.

\section{Overview of the algorithm}\label{sec:overview}


As discussed above, our goal is to solve the min 2-respecting cut problem, i.e., given a tree $T$ find two tree
edges $e,e'$ such that the cut obtained by removing them from the tree is minimal. We denote by $\cut(e,e')$
the cut value defined by $e,e'$.
It is easy to show that if we fix two edges $e,e'$, we can compute $\cut(e,e')$ in $O(D)$ rounds by a simple computation over a BFS tree, where we sum the costs of all edges that cross the cut.\footnote{We use the tool of lowest common ancestors (LCA) labels to identify which edges cross the cut.} However, to find the two edges that define the minimum 2-respecting cut, we may need to go over all possible $\Omega(n^2)$ pairs of tree edges. This requires $\Omega(n^2)$ time, which is clearly too expensive. To get a faster algorithm, our general approach is as follows:

\begin{enumerate}
\item Use \emph{strucutral properties of cuts} to bound the number of values $\cut(e,e')$ we should compute.\label{bound_pairs}
\item Use a \emph{routing trick} to efficiently route information in the graph, and avoid using global communication over a BFS tree for all the computations.\label{bound_global} 
\end{enumerate} 

Next, we elaborate on our approach. We start by discussing helpful notation. Then, we discuss a simplified version of our algorithm for a \emph{spider graph} that already allows us to present some of the main ingredients of the algorithm. Finally, we discuss additional ingredients required for extending the algorithm to a general graph. The main goal of this section is to give an overview of the algorithm which is, on one hand, informal so that the reader can understand the main techniques and tools that are used; and, on the other hand, detailed enough so that the reader can convince herself of the correctness of the algorithm before progressing to the dry technicalities of the later sections.

\subsection{Basic notation: Cover values}

For a tree edge $e$, we denote by $\cov(e)$ the value of the cut obtained by removing $e$ from the tree (see Figure \ref{cover_pic}), this is the sum of costs of all edges that cross the cut. The notation implies that the edges crossing the cut \emph{cover} $e$. Note that these are exactly the edges $\{u,v\}$ such that $e$ is in the unique tree path between $u$ and $v$.
For two tree edges $e,e'$, we denote by $\cov(e,e')$ the sum of costs of all edges that cover both $e$ and $e'$ (see Figure \ref{cover_pic}). Note that the edges that cross the cut defined by $e,e'$ are exactly all edges that cover exactly one of $e,e'$. This immediately gives the following claim (see Section \ref{sec:cover} for a formal proof).

\setlength{\intextsep}{2pt}
\begin{figure}[h]
\centering
\setlength{\abovecaptionskip}{-2pt}
\setlength{\belowcaptionskip}{6pt}
\includegraphics[scale=0.5]{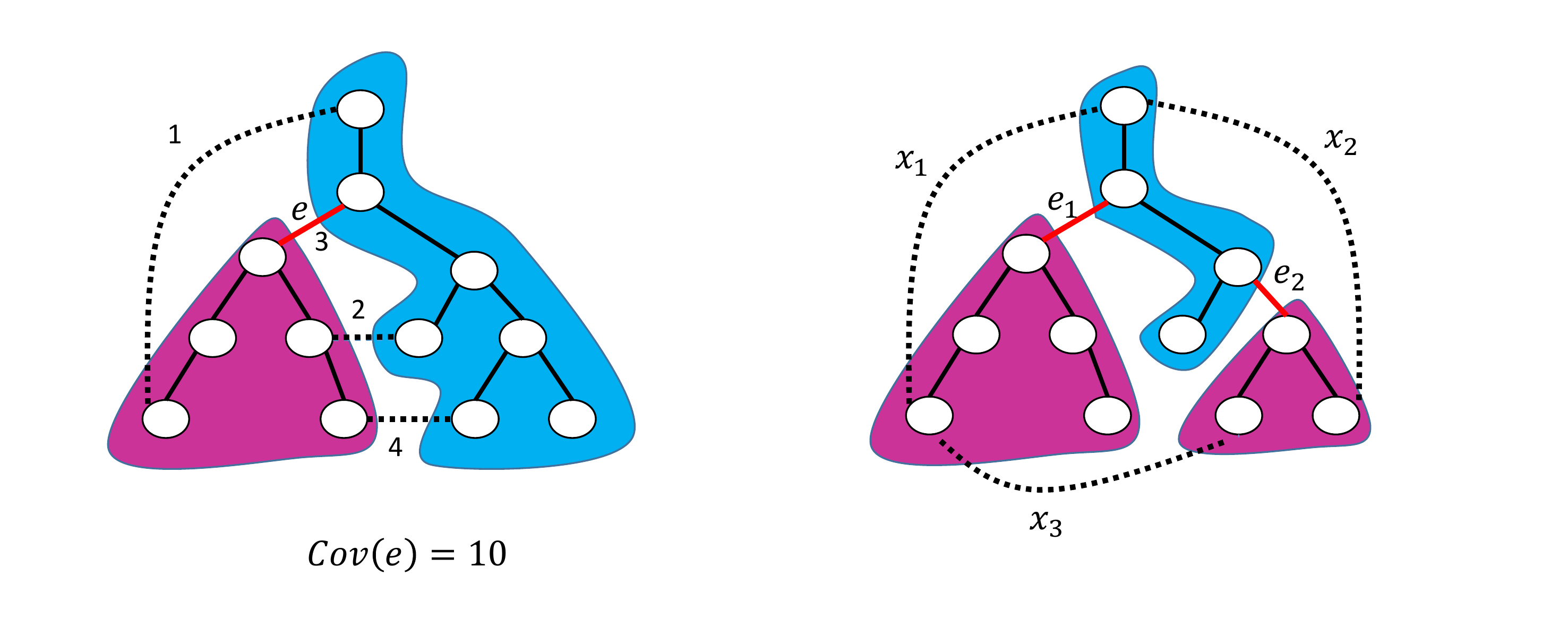}
 \caption{\small Illustration of cover values. Solid edges are tree edges, where dotted edges are non-tree edges. The purple-blue cut on the left is the 1-respecting cut defined by $e$. The purple-blue cut on the right is the 2-respecting cut defined by $e_1,e_2$. The edge $x_1$ is an example of an edge that covers $e_1$ but not $e_2$, the edge $x_2$ is an example of an edge that covers $e_2$ but not $e_1$ and the edge $x_3$ is an example of an edge that covers both $e_1,e_2$. The edges that cross the purple-blue cut are the edges that cover exactly one of $e_1,e_2$.
 }
\label{cover_pic}
\end{figure}

\begin{restatable}{claim}{claimCov}\label{claim_cover}
$\cut(e,e')=\cov(e)+\cov(e')-2 \cov(e,e').$
\end{restatable}


Computing the values $\cov(e)$ can be done in $\tilde{O}(D+\sqrt{n})$ time using standard techniques  such that each tree edge $e$ knows the value $\cov(e)$ (such computations are done, for example, in \cite{Nanongkai-Su}). Hence, the question of efficiently computing $\cut(e,e')$ boils down to the question of efficiently computing the cover value $\cov(e,e')$ due to Claim \ref{claim_cover}.
To this end, note that when we fix two edges $e,e'$, we can compute $\cov(e,e')$ in $O(D)$ rounds by summing the costs of edges that cover $e$ and $e'$ where the communication happens over a BFS tree. The main challenge is to bound the number of such computations needed, and to be able to \emph{parallelize} the computations by diving them into disjoint sets of local computations in order to avoid high congestion over the global BFS tree. Next, we discuss the ingredients allowing us to do so.   


\subsection{Simple example: Spider graph}

We start by discussing a simplified version of our algorithm where the tree $T$ is a spider with the following structure: $T$ has a root $r$ and attached to it are $k=\sqrt{n}$ paths of length $\sqrt{n}$. We refer to these paths as the \emph{legs} of the spider.

\subsubsection*{(I) Short-paths routing trick for comparing two paths}

The first observation is that if we fix two paths $P,P'$ of the spider and want to find the minimum 2-respecting cut with one edge in $P$ and one edge in $P'$, we can do so in just $O(\sqrt{n})$ time, although we have $n$ different pairs of edges.

\begin{claim}(Routing trick for spider graph) \label{spider_routing}
Fix two legs of the spider $P,P'$. Finding the values $\{e,e',\cut(e,e')\}$ for a pair of edges $e \in P, e' \in P'$ that minimize the cut value, takes $O(\sqrt{n})$ time. 
\end{claim}

\noindent
The proof can be divided into the following two cases.
\begin{enumerate}
\item \textbf{There is an edge $f$ between $P$ and $P'$.} The main idea here is to use the edge $f$ to route information between $P$ and $P'$, and then compute the cut values locally via aggregate computations inside $P$ and $P'$. This allows us to work in parallel in different paths.
In more detail, by Claim \ref{claim_cover}, for any pair of edges $e \in P, e' \in P'$, we have $\cut(e,e')=\cov(e)+\cov(e')-2 \cov(e,e').$ As discussed before, let us assume that the values $\{\cov(e)\}_{e \in P}$ are known to edges in $P$, and the values $\{\cov(e')\}_{e' \in P'}$ are known to edges in $P'$. As the paths have length $O(\sqrt{n})$, in $O(\sqrt{n})$ time we can route all these cover values to one of the paths, say $P$, leading to a $O(\sqrt n)$ congestion in the edge $f$. Now we only need to compute the values $\cov(e,e')$. To this end, note that if we fix two edges $e \in P, e' \in P'$, the edges that cover both of them are exactly all the edges that have one endpoint below $e$, and one endpoint below $e'$, all these edges connect $P$ and $P'$. Hence, we can run aggregate computation inside $P$ to sum the costs of all such edges. In this computation we fix $e' \in P'$, and compute for each edge $e \in P$, the cost $\cov(e,e')$. This requires $\tO(1)$ congestion in the edges of $P$ (the dilation, however, is $O(\sqrt n)$). To compute these values for all $e' \in P'$, we use pipelining, which results in $O(\sqrt{n})$ complexity ($O(\sqrt n)$ congestion and dilation) for computing all cut values.
\item \textbf{There is no edge between $P$ and $P'$.} Here we do not have a direct edge for communication, however it turns out that the structure of the minimum cut is actually much simpler in this case. The crucial observation here is that if there is no edge between $P$ and $P'$, then $\cov(e,e')=0$ for any $e \in P, e' \in P'$. Then, from Claim \ref{claim_cover}, the question of minimizing $\cut(e,e')$ boils down to finding two edges $e \in P, e' \in P'$ that minimize the cover values in $P$ and $P'$, respectively. We can compute these values locally in the paths and then broadcast them to the whole graph. Since we only need to broadcast $O(1)$ pieces of information per path, we can compare \emph{all} pairs of paths with no edge between them by broadcasting $O(\sqrt{n})$ information to the whole graph.  
\end{enumerate}


To conclude, we are in a win-win situation. We either have an edge between the paths, in which case we can use it for routing and compute the cut values by internal computations inside $P,P'$---this is helpful for running such computations in parallel in different paths. Or we do not have an edge, in which case we need some global communication but we can actually limit the amount of global communication significantly as the structure of the minimum cut becomes much simpler.
While we focus here on a spider graph, the same principles can be extended to work for a more general setting, where we have two tree paths $P,P'$ of size $O(\sqrt{n})$ we want to compare.

\subsubsection*{(II). Structural lemma for bounding interesting paths}

While we showed that comparing two legs of the spider can be done in $O(\sqrt{n})$ time and that this can be done in parallel for \emph{disjoint} pairs of paths, if we want to use it to compare all pairs of legs of the spider, it requires $\Omega(n)$ time. 
This follows, as comparing two legs $P$ and $P'$ is based on running $\Omega(\sqrt{n})$ aggregate computations in one of the paths, which leads to congestion $\Omega(\sqrt{n})$. If we need to compare the same leg to all $\sqrt{n}$ legs, the total congestion is $\Omega(n)$.
%
To overcome it, our main goal now is to bound the number of pairs of paths we need to compare.
For this, we define a notion of \emph{interesting paths}. We show that we only need to compare pairs of paths that are \emph{interested} in each other, and we prove a structural lemma that shows that each path is interested in $O(\log{n})$ paths.
This structural lemma and the short-paths routing trick together lead to a complexity of $\tilde{O}(\sqrt{n})$ for computing the min 2-respecting cut in a spider graph, as each leg of the spider only participates in $O(\log{n})$ computations that take $O(\sqrt{n})$ time. 
Next, we elaborate on the notion of interesting paths.

\paragraph{Interesting paths.} The notion of \emph{interesting paths} is an extension of \cite{mukhopadhyay2019weighted} where they show the following: for each \emph{edge} $e$, there is only a small number of ancestor to descendant paths where the edge $e'$ that minimizes $\cut(e,e')$ can be.  We say that an edge $e$ is \emph{interested} in an edge $e'$ if $$\cov(e,e') > \cov(e)/2.$$ The crucial observation is the following (see the beginning of Section \ref{sec:lemma_and_samp} for an explanation).

\begin{claim}(\cite{mukhopadhyay2019weighted})
If the pair $\{e,e'\}$ participates in the min 2-respecting cut, then $e$ and $e'$ are interested in each other.
\end{claim}

\cite{mukhopadhyay2019weighted} defines the notion of \emph{edges} being interested in each other. In this work, we generalize this notion to \emph{paths}. For simplicity of presentation, we focus here on the spider graph. We say that an edge $e$ is interested in a leg $P$ of the spider if $e$ is interested in at least one edge in $P$. As was shown in \cite{mukhopadhyay2019weighted}, $e$ can be only interested in one leg $P$ where $e \not \in P$. The reason is simple. If $e$ is interested in $P$, it follows that more than half of the edges that cover $e$ go towards $P$. This can only happen for one path.

\paragraph{Interesting path counting lemma.} While the above discussion implies that each \emph{edge} is interested in one path, this is not enough in our case. The reason is that each edge of a path $P$ may be interested in a different leg of the spider, in which case we may need to apply Claim \ref{spider_routing} on all pairs of paths to find the min 2-respecting cut, which is too expensive. To overcome it, we show a stronger argument. We say that a path $P$ is interested in a path $P'$ if there is an edge in $P$ interested in $P'$.

\begin{lemma}(Interesting path counting lemma for spider graph) \label{overview_structural_lemma}
Each leg $P$ of the spider is interested in $O(\log{n})$ legs.
\end{lemma}

The proof idea is as follows. If some edge $e \in P$ is interested in some leg $P_1$, it means
that the total weight of edges that cover $e$ and some edge $e' \in P_1$ (and, hence, the total weight of such edges that \textit{go towards} $P_1$) is
at least half of the weight of edges that cover $e$. To find out which legs of the spider $P$ is interested in, let us start from the leaf of $P$ and traverse towards to root: While doing so, we count the total weight of non-tree edges that cover the current edge in $P$ and end somewhere outside. The crucial observation is the following: Each time we reach some edge $e \in P$ that is interested in some \textit{new} leg $P_i$, we know that the total weight of such \textit{new} non-tree edges that cover $e$ and go towards $P_i$ has to be at least the total weight of non-tree edges that we have counted so far---otherwise, $e$ would not be interested in $P_i$. So, every time we encounter such an edge in $P$ while traversing from the leaf to the root, the total weight of the non-tree edges that we count doubles. Since the total weight of edges is polynomially bounded, such edges in $P$ can be found only a logarithmic number of times, which shows that $P$ can be only interested in $O(\log n)$ different legs.


\setlength{\intextsep}{2pt}
\begin{figure}[h]
\centering
\setlength{\abovecaptionskip}{-2pt}
\setlength{\belowcaptionskip}{6pt}
\includegraphics[scale=0.5]{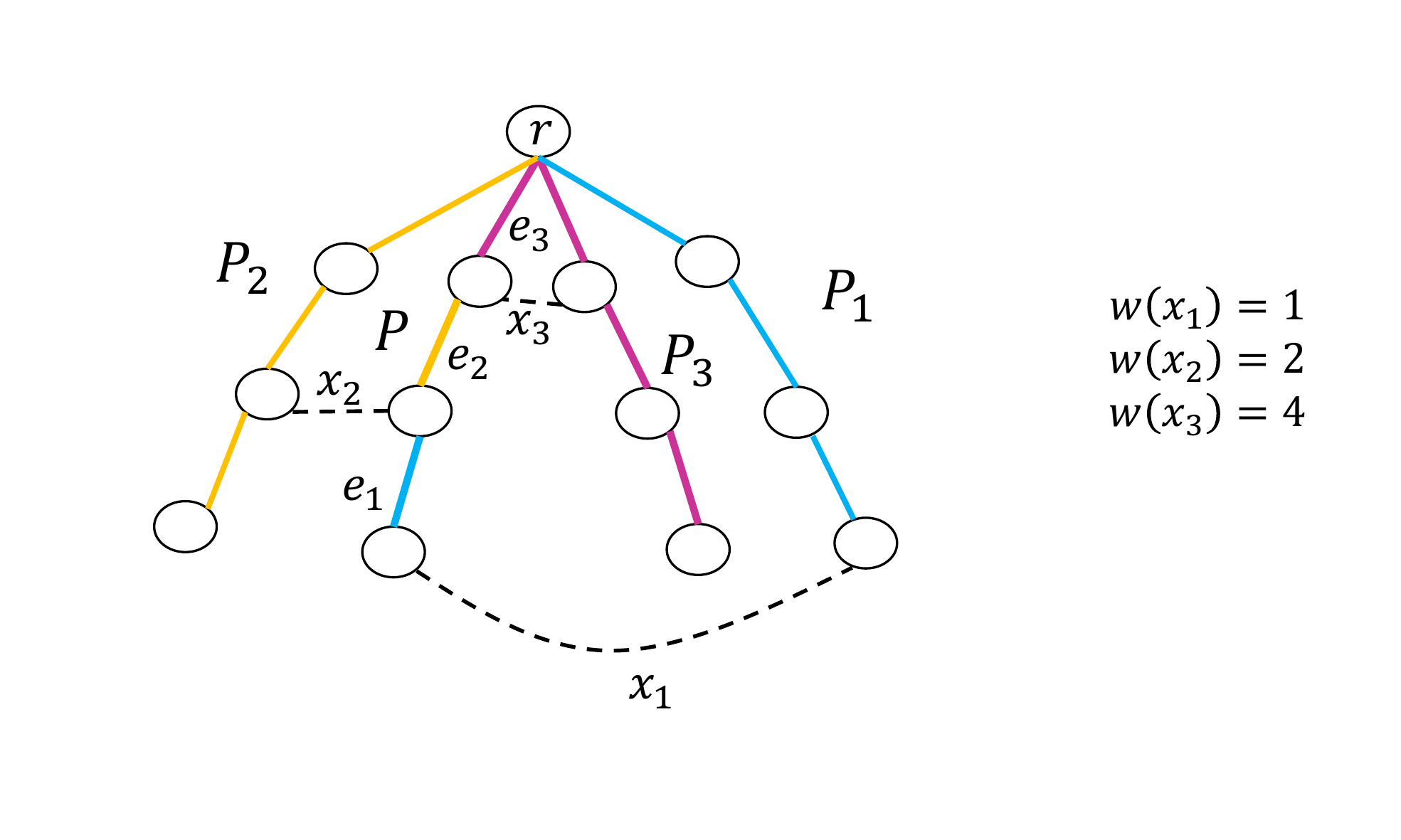}
 \caption{\small Illustration of the interesting path counting lemma. We have a path $P$, where the edge $e_i$ is interested in the path $P_i$. Assume that the weights of all tree edges are 0. Note that since $e_3$ is interested in $P_3$, we must have $w(x_3) > w(x_1)+w(x_2)$. Similarly, as $e_2$ is interested in $P_2$, we must have $w(x_2) > w(x_1).$}
\label{structural_lemma_pic}
\end{figure}


We remark that the proof of the interesting path counting lemma that we just mentioned crucially relies on the simple structure of the spider graph---in particular on the fact that the legs of the spider are edge-disjoint.\footnote{Even though edge-disjointness suffices to argue for the spider graph, we need a more restricted structure for general spanning tree. See Section \ref{overview:general} and \ref{ssec:-interesting_lemma} for more details.} In a general graph, however, we cannot guarantee such structure among the paths, and, hence, we can no longer show that each path is only interested in a small number of paths. To deal with it, we restrict our attention to paths that have a \textit{nice structure} and prove a variant of the interesting path counting lemma with respect to them. This is further discussed in Section \ref{overview:general}.

\paragraph{Finding interesting paths.} Even though Lemma \ref{overview_structural_lemma} bounds the number of legs of the spider that each leg can be interested in, we are still left with the job of identifying such interesting legs in order to complete the algorithm for this simple case of spider graph. One immediate approach is the following: Since an edge $e$ is interested in a path $P$ only if more than half of the edges that cover $e$ \textit{go towards} the path $P$, we can use sampling to identify the set of interesting legs w.r.t. an edge $e$. However, because of the nature of the sampling, we will not be able to pinpoint the set of interesting legs. Rather, we can obtain a set of legs which is a superset of the actually interesting legs. We denote that $e$ is \textit{potentially interested} in each leg of this set. As it turns out, this is sufficient for our purpose---we show that these paths still satisfy nice properties that allow us to prove that each leg $P$ of the spider is \emph{potentially interested} in $O(\log{n})$ legs. For simplicity of presentation, in this section we refer to all these paths as
paths $e$ is interested in. Combining everything, it is easy to see at this point how to implement the minimum 2-respecting cut algorithm in $O(D + \sqrt n)$ time complexity when the spanning tree is such a simple spider graph.

\subsubsection*{(III). Dealing with long paths via partitioning}   

As explained above, the short-paths routing trick and interesting path counting lemma allow us to find the min 2-respecting cut in a spider with legs of length $\sqrt{n}$. However, Claim \ref{spider_routing} relies on such short length of the legs. We next ease this restriction on the spider graph and explain how to handle a spider graph that may have longer legs using a \emph{partitioning} technique.

For simplicity, we start with a spider graph that is identical to the previous one, but we add to it one long leg of length $n$ (now the number of vertices is $2n$). The interesting path counting lemma (Lemma \ref{overview_structural_lemma}) still holds for this case, as the legs of the spider are still edge-disjoint. The only issue is that if we want to compare any leg of the spider to the new long leg, it requires time proportional to the length of the long leg, which is too expensive.

\paragraph{Comparing a short and a long leg.} 
Denote by $P_{long}$ the long leg of the spider, and by $P'$ some leg we want to compare to $P_{long}$. If we want to compare $P'$ and $P_{long}$ naively using Claim \ref{spider_routing}, we see that it requires $\Omega(n)$ time. To overcome it, a natural approach could be first to break $P_{long}$ to smaller sub-paths of length $\sqrt{n}$: denoted by $P_1,...,P_k$ (see Figure \ref{partitioning_spider_pic}), and then compare $P'$ to each one of the sub-paths $P_i$ separately. In doing so, we can route information from $P'$ to each one of the sub-paths $P_i$ separately, and then try to compute the cut values $\cut(e',e)$ for $e' \in P'$ and $e \in P_i$ internally inside $P_i$.\footnote{We assume for simplicity that there is an edge between $P'$ to each one of the sub-paths $P_i$, the case there is no edge is simpler.}
Unfortunately, we are faced with a delicate issue if we use this approach: $P_i$ alone does not have enough information to compute the cut values. To illustrate this, consider, for example, a pair of edges $e' \in P', e \in P_i$. Edges that cover $e',e$ may have both endpoints outside $P_i$ (see the left side of Figure \ref{partitioning_spider_pic}, where $P_i=P_2$), and hence $P_i$ cannot compute the value $\cov(e',e)$, without additional information from $P'$. More concretely, denote by $\extcov(e',P_i)$ the weight of the edges that cover $e' \in P'$ and the whole path $P_i$, and have both endpoints outside $P_i$.
If we want to compare all edges $e' \in P'$ to all sub-paths $P_i$, we need to send all the values $\extcov(e',P_i)$ from $P'$ to the sub-path $P_i$. Overall, since we have $\sqrt{n}$ edges in $P'$ and $\sqrt{n}$ sub-paths $P_i$, we need to send $\Omega(n)$ information from $P'$ to all other sub-paths in total. This may create $\Omega(n)$ congestion in $P'$ as we may need to collect $\Omega(n)$ information in one vertex in $P'$ before sending it---this is way more than the congestion we can afford. 
To overcome this issue, we partition $P'$.

\paragraph{Path-partitioning lemma.}
We show a path-partitioning lemma, that states that we can break the path $P'$ into sub-paths $P'_1,...,P'_k$ such that we only need to compare $P'_i$ to $P_i$ (see the right side of Figure \ref{partitioning_spider_pic} for illustration). Additionally, these sub-paths are almost disjoint. More concretely, we show the following:
\begin{itemize}
\item We can break the path $P'$ to subsets $P'_1,...,P'_k$, such that $\sum_{i=1}^k |P'_i| = O(\sqrt{n}).$ Here $|P'_i|$ refers to the number of edges in $P'_i$. 

\item If $\{e,e'\}$ define the min 2-respecting cut with $e \in P_{long}$ and $e' \in P'$, then there is an index $i$ such that $e \in P_i$ and $e' \in P'_i$. Hence, it is enough to solve the min 2-respecting cut problem on the pairs $\{P_i, P'_i\}_{i=1}^k$.

\end{itemize}

We prove the path-partitioning lemma using monotone structure of minimum cuts described in \cite{mukhopadhyay2019weighted} (see Section \ref{sec:mon_part}).
Based on this lemma, we can now deal with comparing $P'$ and $P_{long}$. As explained above, our goal is to route information of the form $\extcov(e',P_i)$ from $P'$ to each one of the sub-paths $P_i$. The crucial observation is that after the partitioning, each edge of $P'$ on average should be compared only to $O(1)$ sub-paths $P_i$, i.e, for each $e' \in P'$, we need to route $\extcov(e',P_i)$ for constant many $P_i$. Hence, the total amount of information to collect and send from $P'$ is now proportional to the number of edges in $P'$, $O(\sqrt{n})$, which leads to an efficient algorithm.\\[-7pt]

\paragraph{Dealing with two long paths.} If the spider graph has two long paths, we can use a variant of the path-partitioning lemma together with a divide-and-conquer approach described in \cite{mukhopadhyay2019weighted} to deal with comparing two long paths. We defer the elaboration on this case to the next section where we discuss the algorithm for a general spanning tree.

\setlength{\intextsep}{2pt}
\begin{figure}[h]
\centering
\setlength{\abovecaptionskip}{-2pt}
\setlength{\belowcaptionskip}{4pt}
\includegraphics[scale=0.4]{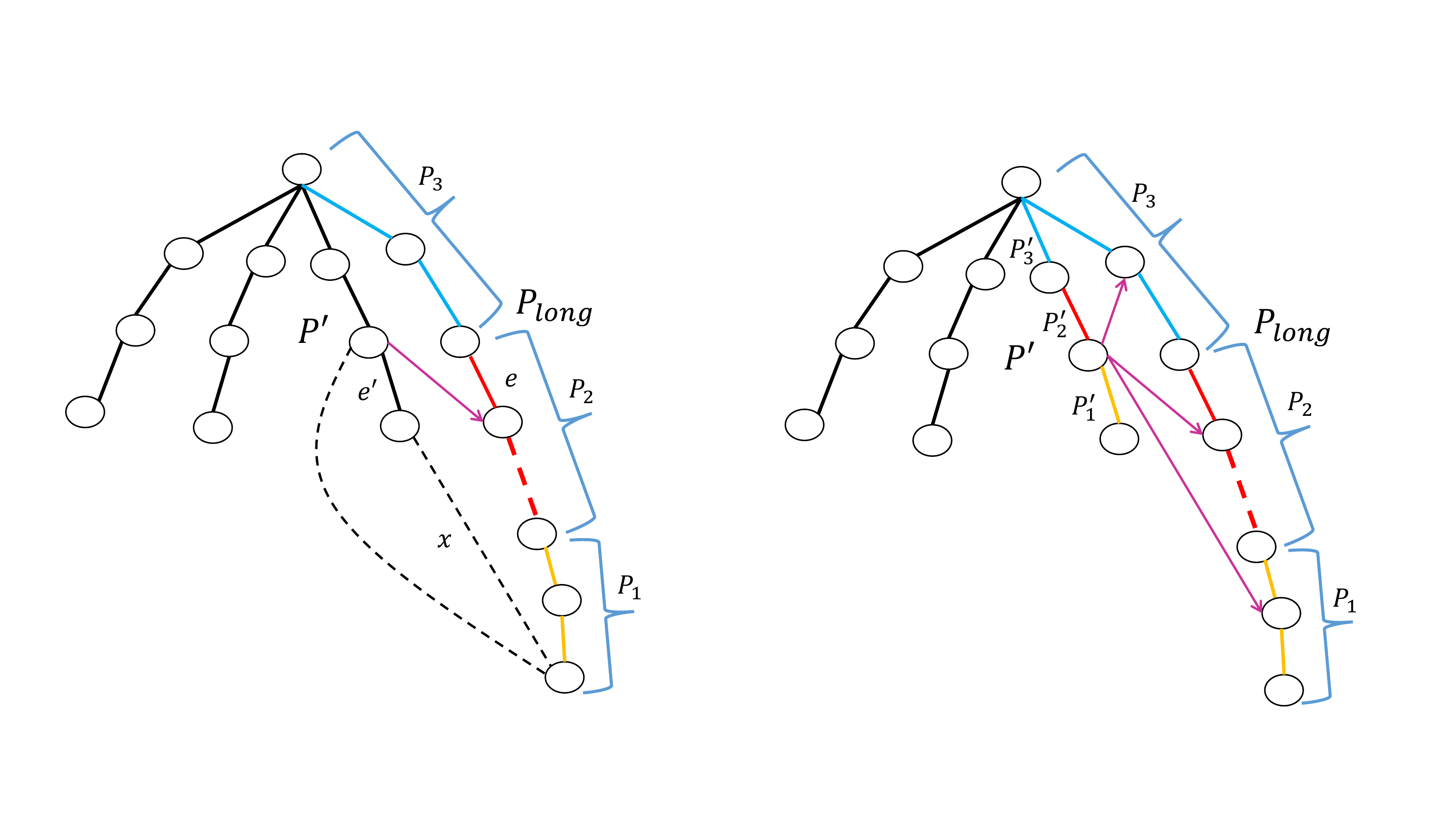}
 \caption{\small Illustration of the partitioning. In the left, note that the edge $x$ covers $e'$ and $e$ but have both endpoints outside $P_2$. The right part illustrates the partitioning, we only need to compare edges in $P'_i$ to $P_i$.}
\label{partitioning_spider_pic}
\end{figure}

\subsection{The algorithm for general graphs: Overview} \label{overview:general}

Up until now, we discussed 3 main ingredients that allow us to solve the spider example: a short-paths routing trick, an interesting path counting lemma, and a path-partitioning lemma. Next, we discuss variants of these tools that are useful for a general graph, as well as additional tools required such as \emph{tree decomposition}, a \emph{pairing theorem}, and a \emph{divide-and-conquer approach}. More details of these tools are provided in the next subsection. 




\subsubsection*{(I). Interesting path counting lemma for a general graph}

In a general graph, it is no longer true that each path $P$ in the graph is only interested in a small number of paths. However, we can show that if we restrict our attention to paths that have a certain structure we can still bound the number of interesting paths. To this end, we define the following: We say that two paths are pairwise orthogonal if the highest edges of each of these paths are on different root to leaf paths. Also, by denoting a path $P'$ to be completely above (or below) a path $P$, we mean that all vertices of $P'$ appear as ancestors (or descendants) of the top (or the bottom) vertex of  $P$ (See Figure  \ref{fig:non-splitting} for reference). For such paths, we show the following lemma: 

\begin{lemma}(Interesting path counting lemma for a general graph) A given path $P$ can be interested in at most $O(\log n)$ pairwise orthogonal paths that
are either orthogonal to $P$ or completely above or completely below $P$. 
\end{lemma}

\subsubsection*{(II). Fragment decomposition}
Of course, not all the paths in the graph are above, below or orthogonal to $P$.
Hence, to use the interesting path counting lemma, we break all the paths in the graph into paths that satisfy some nice structure. To do so, we bring to our construction a variant of a fragment decomposition from \cite{ghaffari2016near,Dory18} (see Section \ref{ssec:-decompFragment}).
At a high-level, we decompose our tree into $O(\sqrt{n})$ edge-disjoint fragments of size $O(\sqrt{n})$. Each fragment $F$ has a very specific structure: it has one main path, called the \emph{highway} of the fragment, between two vertices that are called the root $r_F$ and descendant $d_F$ of the fragment, and additional sub-trees attached to the highway that are contained inside the fragment. The paths in these sub-trees are called \emph{non-highways}. The only vertices that may be connected directly to other fragments are $r_F$ and $d_F$. The following properties are useful for us later:
\begin{enumerate}
\item Non-highways have small length and are completely contained in one fragment.
\item We have $O(\sqrt{n})$ different highways.
\end{enumerate}

While non-highways are contained in one fragment, highways can connect to highways in other fragments and create long paths of highways. We sometimes refer to highways in a single fragment as \emph{fragment highways}, and long paths composed of highways as \emph{super-highways}, to distinguish between the two. One can think of fragment highways and super-highways as short and long legs in the spider example, respectively.

\paragraph{Combining the interesting path counting lemma and fragment decomposition.}
The importance of the fragment decomposition comes from the fact that paths in different fragments do not intersect each other.
We can use this structure of the fragments and the interesting path counting lemma to prove that each non-highway or highway within a fragment is only interested in a small number of ancestor to descendant paths outside their fragment. 
To deal with cuts that have two edges in the same fragment, we exploit the small size of the fragments to compute the cuts efficiently. 


\subsubsection*{(III). Short-paths routing trick for a general graph}

A very basic building block in comparing paths that are interested in each other is to compare two sub-paths of length $O(\sqrt{n})$---we informally denote these paths as \textit{short} paths. We can extend the routing trick (Claim \ref{spider_routing}) to deal with comparing two such short paths (See Section \ref{sec:alg_short_path})---The algorithm is divided into cases depending on whether these short paths are non-highways or fragment highways.

The simplest case is that the paths are non-highways: Here we can show that we have an edge connecting any two non-highways interested in each other, and we can run an algorithm similar to Case 1 in the proof of Claim \ref{spider_routing} to compare such paths.
When one or two of the sub-paths are fragment highways, we may also be in a case that there is no edge connecting the two sub-paths $P',P$ we compare, but similarly to Case 2 in Claim \ref{spider_routing}, we show that we can exploit this, and divide the computation to simple internal computations in each one of $P$ and $P'$ and broadcast of $\tO(1)$ pieces of information over a global BFS tree.

\subsubsection*{(IV). Path-partitioning lemma for a general graph}

We use a variant of the path-partitioning lemma (See Section \ref{sec:mon_part}) to compare a \emph{short path} of length $O(\sqrt{n})$ (either non-highway or a fragment highway) to a long path composed of highways (or \textit{super-highways}). This technique also serves as a building block in comparing two super-highways.

\subsubsection*{(V). Pairing theorem and divide-and-conquer approach}

Lastly, to compare two super-highways we need a few additional tools. To give an idea of the technical bottleneck we face in this case, we mention one main difference from the spider case: In a general graph, it is no longer true that each path is interested in a small number of paths (this only holds when we limit the structure of the paths). This creates a problem when we want to compare \emph{long paths} of super-highways that are interested in each other.

To deal with it, we prove a \emph{pairing theorem} that allows us to pair-up the super-highways into pairs we need to compare such that, in each such pair, only some subset of highways are \emph{active}. We can show that each \emph{fragment highway} is only active in $\text{poly}\log{n}$ pairs, and we show how to compare them in a complexity that \textit{depends only on the number of active highways in each pair}. To compare each such pair of super-highways, we use a divide-and-conquer approach. We elaborate more on this in the next section.

\subsection{The algorithm for general graphs: More details}

Here we discuss the algorithm for general graphs in more detail.
From a high-level, the algorithm works as follows.

\medskip
\fbox{%
  \parbox{0.9\textwidth}{
\begin{enumerate}
\item We compute the fragment decomposition.
\item For each non-highway or fragment highway we compute the paths it is interested in.
\item We compare non-highways that are interested in each other by the \emph{short-paths routing trick}.\label{compare_nh}
\item We compare non-highways and highways that are interested in each other via the \emph{path-partitioning lemma}.\label{compare_nh_h}
\item We compare paths of highways (i.e., super-highways) that are interested in each other via a \emph{pairing theorem} and \emph{divide and conquer} approach.\label{compare_h}
\end{enumerate}
  }%
}

\bigskip\noindent
We next elaborate more on steps \ref{compare_nh}-\ref{compare_h}.

\paragraph{Step \ref{compare_nh}: Comparing non-highways.} Here we use the short length of non-highways and the interesting path counting lemma to get an efficient algorithm. The algorithm is similar to the spider case with short legs, uses a variant of the short-paths routing trick (Claim \ref{spider_routing}) and is based on the following ideas:
\begin{itemize}
\item \textbf{Bounding the number of comparisons.} From the interesting path counting lemma, we only need to compare each non-highway to $\text{poly}\log{n}$ different non-highways in other fragments.
\item \textbf{Working locally.} We can show that if a non-highway $P$ is interested in a non-highway in the fragment $F$ there is an edge $f$ between the sub-tree rooted at $P$ and the fragment $F$ (this follows from the fact that many edges that cover $P$ go towards $F$, and in particular there are such edges). Hence, we can use the edge $f$ to route information about cover values from $F$ to $P$ and then run computations similarly to Case 1 in Claim \ref{spider_routing} in $P$ to compute the cut values. This results in $\tilde{O}(\sqrt{n})$ complexity.
\item \textbf{Parallelizing the computations.} Since we only used local computations inside $P$, we can run such computations in parallel for orthogonal non-highways. Hence, for example, we can do the computations in parallel for non-highways in different fragments.
To work efficiently in parallel in different non-highways in the same fragment, we use a certain \emph{layering} of the non-highway paths (see Section \ref{ssec:-layerDecomp}).
\end{itemize}

\paragraph{Step \ref{compare_nh_h}: Comparing non-highways and highways.} The interesting path counting lemma implies that each non-highway is interested in $\text{poly}\log{n}$ super-highways. We deal with this case in a similar fashion as we dealt with short and long legs when the spanning tree is a spider graph. The main ingredient is a variant of the \emph{path-partitioning lemma} that allows us to break the non-highway into smaller sub-paths that each one of them is only compared to one fragment highway in the long path of highways. Similar to what we described in part (III) of the spider case, the property of the path-partitioning that we use here is that each edge of the non-highway is needed to be compared with only a small number (constant many) of fragment highways on average. Hence, for each fragment highway, we need to route a small amount of information on average and we let them compute the cut values locally.
We next elaborate on two issues: 
\begin{enumerate}
\item How to do many such computations in parallel?
\item What happens if a non-highway is interested in some fragment highway but there is no direct edge between them?
\end{enumerate} 

\textbf{Working in parallel.} To deal with the first issue, we use the interesting path counting lemma on the highways. Basically, it implies that each fragment highway is only interested in a small number of non-highways in other fragments. Hence, even if there are many non-highways interested in some fragment highway $P$, $P$ only participates in computations with non-highways it is interested in, which is enough for computing the min 2-respecting cut. This can be done efficiently, as each fragment highway now only participates in a small number of computations.\\[-7pt]

\textbf{The case there is no edge.} For the second issue, we show that dealing with the case there is no edge is actually easier (similarly to Case 2 in the short-paths routing trick). First, we broadcast $O(\sqrt{n})$ pieces of information to the whole graph, about the minimum cover values of an edge in each fragment highway. Then, based on this alone, each non-highway $P'$ can compute internally in $O(\sqrt{n})$ time the minimum 2-respecting cuts that have one edge in $P'$ and one edge in any fragment highway $P$ where there is no edge between $P'$ and the fragment of $P$.

\paragraph{Step \ref{compare_h}: Comparing highways.} 
The two main ingredients we use here are a \emph{pairing theorem} and \emph{divide-and-conquer} approach.\\[-7pt]

\textbf{Pairing theorem.} While the interesting path counting lemma implies that each fragment highway is interested in a small number of super-highways, this is not enough to get a fast algorithm. One issue is that the same super-highway may need to participate in too many computations.
To deal with it we prove a pairing theorem, with the following guarantees.
\begin{itemize}
\item We partition all the highways in the graph to pairs of super-highways $(P_{H},P'_{H})$, such that in each pair of super-highways we denote a subset of fragment highways that are \emph{active}.
\item Each fragment highway is active in $\text{poly}\log{n}$ pairs.
\item If the min 2-respecting cut has 2 edges in the fragment highways $P,P'$, there is a pair of super-highways $(P_{H},P'_{H})$, such that $P \in P_H, P' \in P'_H$, and both $P,P'$ are active in this pair. 
\end{itemize}
This basically allows us to divide the computation into a series of comparisons between different pairs of super-highways. Next, we explain how to compare two such super-highways.\\[-7pt]

\textbf{Comparing two super-highways.} The basic idea here is to use a variant of the path-partitioning lemma, together with a divide-and-conquer approach from \cite{mukhopadhyay2019weighted}. First, we show that if we want to compare one fragment highway to a super-highway, we can do it efficiently using a variant of the path-partitioning lemma. Next, we use it as a black-box to compare two super-highways. 
For simplicity, we assume that all the fragment highways in the super-highways we consider are active.
The algorithm works as follows (see Figure \ref{two_highway_intro}). Let $P_{H_1},P_{H_2}$ be the two super-highways we want to compare. We first compare the middle fragment highway $P \in P_{H_1}$ to the super-highway $P_{H_2}$. Say that the min 2-respecting cut we find has the second edge in the fragment highway $P' \in P_{H_2}.$ We now use the black-box algorithm to compare $P'$ to the super-highway $P_{H_1}$. After this, we either found the min 2-respecting cut of $P_{H_1}$ and $P_{H_2}$, in case it has one edge in $P$ or $P'$, or we can use a monotone structure of minimum 2-respecting cuts to break the problem to two smaller disjoint problems we can solve in parallel (see Figure \ref{two_highway_intro}). We continue in the same manner until we remain with disjoint problems where one of the sides has only one fragment, they can then be solved directly using the
black-box algorithm. Overall we have $O(\log n)$ iterations, each one takes $\tilde{O}(D + \sqrt{n})$ time using the black-box algorithm, as we work on different disjoint problems in parallel.\\[-7pt]  

\setlength{\intextsep}{2pt}
\begin{figure}[h]
\centering
\setlength{\abovecaptionskip}{2pt}
\setlength{\belowcaptionskip}{2pt}
\includegraphics[scale=0.5]{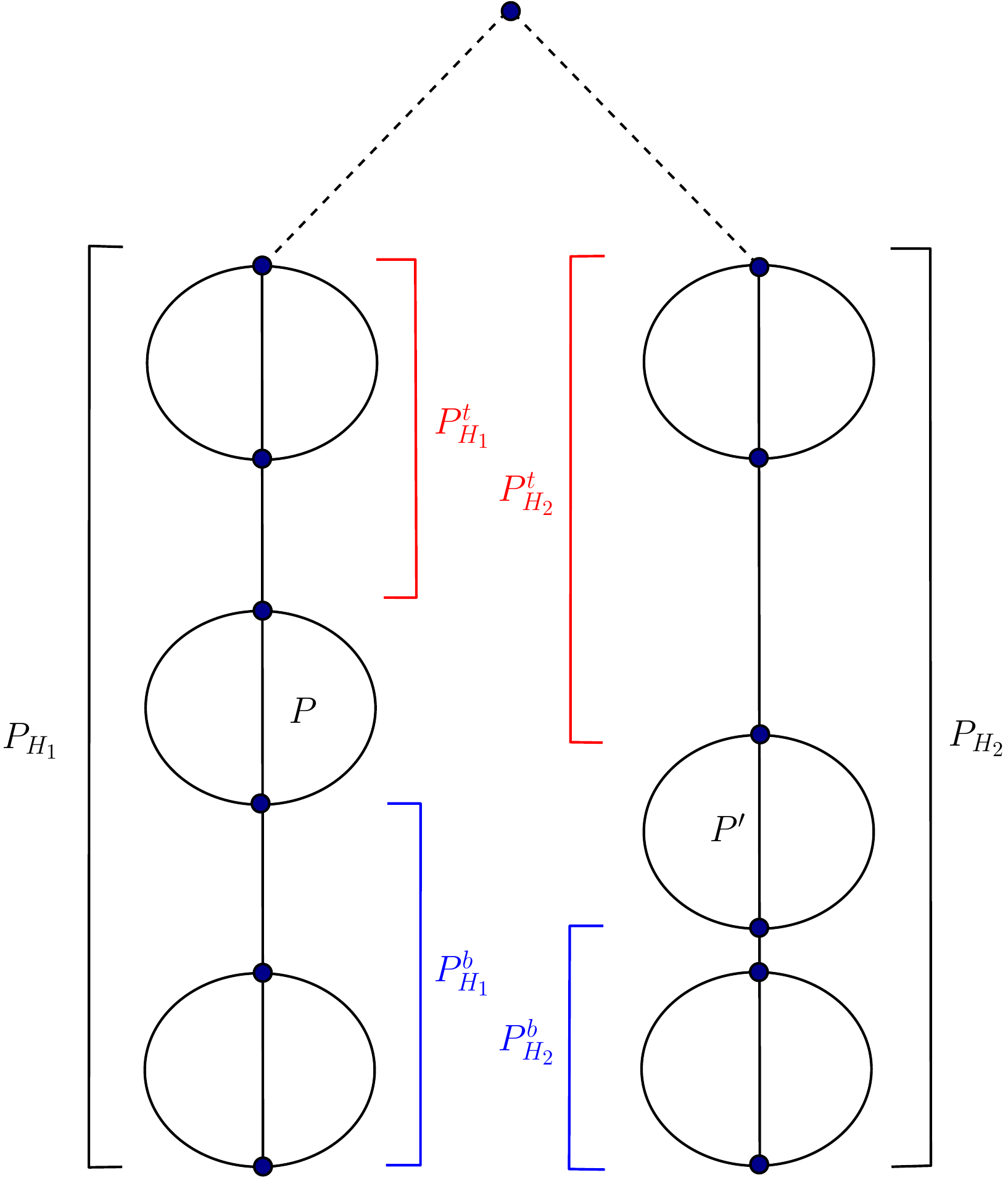}
 \caption{\small Illustration of the highway-highway case. After one iteration, we are left with the disjoint
red and blue problems we can solve in parallel.}
\label{two_highway_intro}
\end{figure}

\textbf{Working in parallel.} The algorithm for comparing super-highways boils down to many smaller computations where we compare two fragment highways (using a variant of the short-paths routing trick). Every time we compare two fragment highways $P$ and $P'$, the computation is divided to a local part where we route information between them and do local computations inside their fragments,\footnote{In the case there is no edge between them, the computation is easier and we do not have this part.} and a \emph{global} part, where global communication takes place in order to compute the cost of edges that cover both $P$ and $P'$ and have edges outside their fragments.\footnote{In the non-highway case, we didn't have this part because any edge that covers a non-highway has at least one endpoint in its fragment.} Hence, to get an efficient algorithm we should bound the number of pairs of fragment highways we compare, as for any such pair we need to use global communication of $O(1)$ information. Note that since the number of fragments is $O(\sqrt{n})$, there is a \emph{linear} number of possible pairs. 
We show that we can bound the amount of global communication by $\tilde{O}(\sqrt{n})$ using the following key ideas:
\begin{itemize}
\item We show that the amount of global communication required for comparing two super-highways is linear in the number of \emph{active} fragments in the pair.
\item The pairing theorem guarantees that each fragment is only active in $\text{poly}\log{n}$ pairs. Hence, comparing all pairs results in sending $\tilde{O}(\sqrt{n})$ global information.
\end{itemize} 
Using these ideas we get a complexity of $\tilde{O}(D+\sqrt{n})$.

\paragraph{Organization.} 
The paper is organized as follows. First, in Section \ref{sec:prelim}, we give some useful notation and claims.
In Section \ref{sec:building_blocks}, we discuss the fragment and layering decompositions. In Section \ref{sec:lemma_and_samp}, we explain how we compute and bound the number of paths each path is interested in, and prove the interesting path counting lemma. In Section \ref{sec:alg_short_path}, we show the short-paths routing trick for comparing two short paths. In Section \ref{sec:mon_part}, we discuss the variants of the path-partitioning lemma we use in our algorithm. Finally, in Section \ref{sec:2_respecting}, we combine all the ingredients to obtain our algorithm for finding the min 2-respecting cut. A schematic description of the algorithm appears in Section \ref{sec:schematic_description}.
\section{Preliminaries}\label{sec:prelim}

\subsection{The model and assumptions}\label{ssec:the_model}

Throughout the paper, we consider the CONGEST model of distributed computing. In this model, one is given a network on $n$ vertices in the form of a graph $G=(V,E)$. Initially, each vertex knows its own unique Id and the Id's of its neighbors in $G$. Communication takes place in synchronous rounds, i.e. in each round, each vertex can send a message of $O(\log n)$ bits to each of its neighbors. The given graph may be equipped with a weight function $w:E\to \mathbb{N}$, in which case each vertex knows also the weights of its incident edges. In case of the min cut problem, we assume that weights are integers and polynomially bounded, hence the weight of any given edge can be represented using $O(\log n)$ bits. At times, we refer to edges in the graph $G$ as performing computations, this means that one of the endpoints of a given edge is actually performing the computation. The specific endpoint is clear through context or specifically mentioned.


\subsection{2-respecting cuts \& tree packing}\label{ssec:def_claims} 

First of all, we discuss the reduction from finding the minimum cut in a given weighted graph $G=(V,E,w)$ to finding the minimum 2-respecting or 1-respecting cut in a given rooted spanning tree $T$ of $G$. We now define the relevant notions. 

For a given weighted graph $G=(V,E,w)$, and a cut $S\subseteq V$, we denote the \emph{value} of $S$ by $w(S)$ and define it to be $w(S)=\sum\limits_{e\in E(S,V\backslash S)} w(e)$. We denote by $E(S,V\backslash S)$ the set of edges of $G$ that cross the cut defined by $S$, i.e. $\set{(u,v)\in E\mid u\in S, v\not\in S}$.

\begin{definition}\label{def:k-min-respect}
Given a graph $G=(V,E)$ and a spanning tree $T=(V,E_T)$ of $G$, we say that a cut $S\subseteq V$ $k$-respects $T$ if it cuts at most $k$ edges of $T$, i.e., $|\set{e\in E_T\mid e\in E(S,V \setminus S)}|\leq k$. The minimum $k$-respecting cut is the cut $S$ with minimal value $w(S)$ among all $k$-respecting cuts.
\end{definition}

In this paper, we are interested in 2-respecting cuts. Figure \ref{fig:2-respect_example} illustrates some examples.
As mentioned in the introduction, the problem of finding a minimum cut of $G$ can be reduced to finding a 2-respecting cut w.r.t. a given spanning tree $T$. 
The seminal work of Karger \cite{Karger00} showed this reduction in the sequential setting. In this paper, we employ a theorem from \cite{DagaHNS19} which implements the reduction in the distributed setting for weighted graphs. More details about the reduction from min cut to 2-min respecting cut can be found in Appendix \ref{appen:redutction_to_respect}. 
\begin{theorem}[From \cite{DagaHNS19}]\label{thm:reduction_mincut_to_respect}
Given a weighted graph $G$, in $\tilde{O}(\sqrt{n}+D)$ rounds, we can find a set of spanning trees ${\cal T}=\set{T_1,...,T_k}$ for some $k=\Theta(\log ^{2.2} n)$ such that w.h.p. there exists a min-cut of $G$ which 2-respects at least one spanning tree $T\in \cal T$. Also, each node $v$ knows which edges incident to it are part of the spanning tree $T_i$, for $1\leq i\leq k$.
\end{theorem}


\begin{figure}
    \centering
    \includegraphics[scale=0.4]{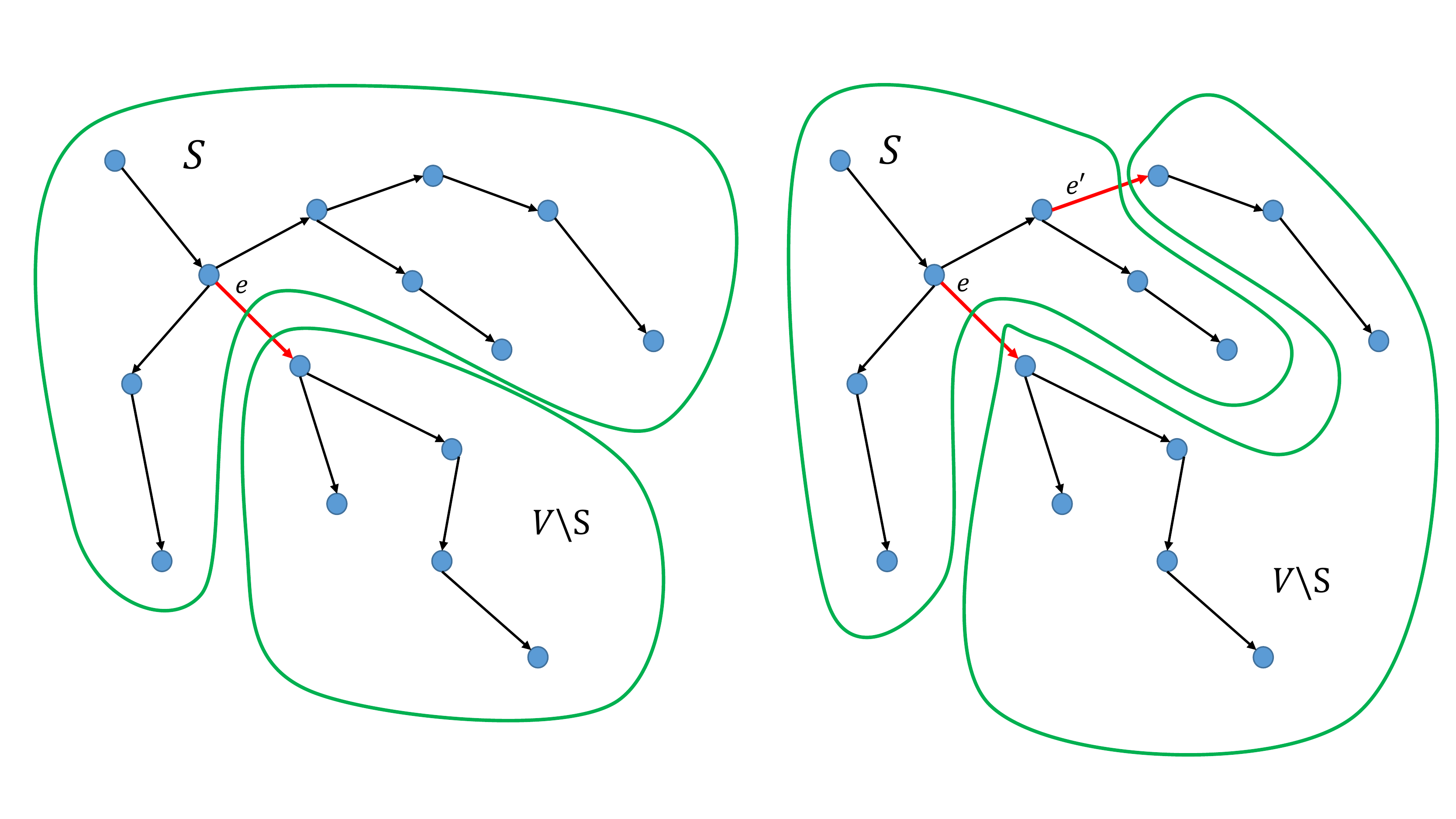}
    \caption{\small Examples of 1 respecting cut of the edge $e$ (Left), and 2-respecting cut of the edges $e,e'$ (Right), with non-tree edges omitted.}
    \label{fig:2-respect_example}
\end{figure}
\paragraph{}

 For a pair of tree edges $(e',e)$, we denote by $\ecut(e,e')$  the set of edges of $G$ that takes part in the 2-respecting cut defined by $e',e$ and $\cut(e',e)$ denotes the value of this 2-respecting cut, i.e., the total edge weight of the set $\ecut(e,e')$. See Figure \ref{fig:cover_and_cut_examples} for examples.

\subsection{Cover values} \label{sec:cover}

For a tree edge $e$, we say that an edge $x=\{u,v\}$ \emph{covers} $e$, if $e$ is in the unique $u-v$ path in the tree (See Figure \ref{fig:cover_and_cut_examples}). In particular, $e$ covers $e$, and all other edges that cover $e$ are non-tree edges. We denote by $\cov(e)$ the total weight of edges that cover $e$, and we denote by $\ECov(e)$ the set of edges that cover $e$. For two tree edges $(e',e)$, we denote by $\cov(e',e)$ the total weight of edges that cover both $e'$ and $e$, and we denote by $\ECov(e,e')$ the set of edges that cover both $e$ and $e'$.  
We denote by $p(v)$ the parent of $v$ in the tree. The following holds.

\begin{figure}[h]
    \centering
    \includegraphics[width=\textwidth]{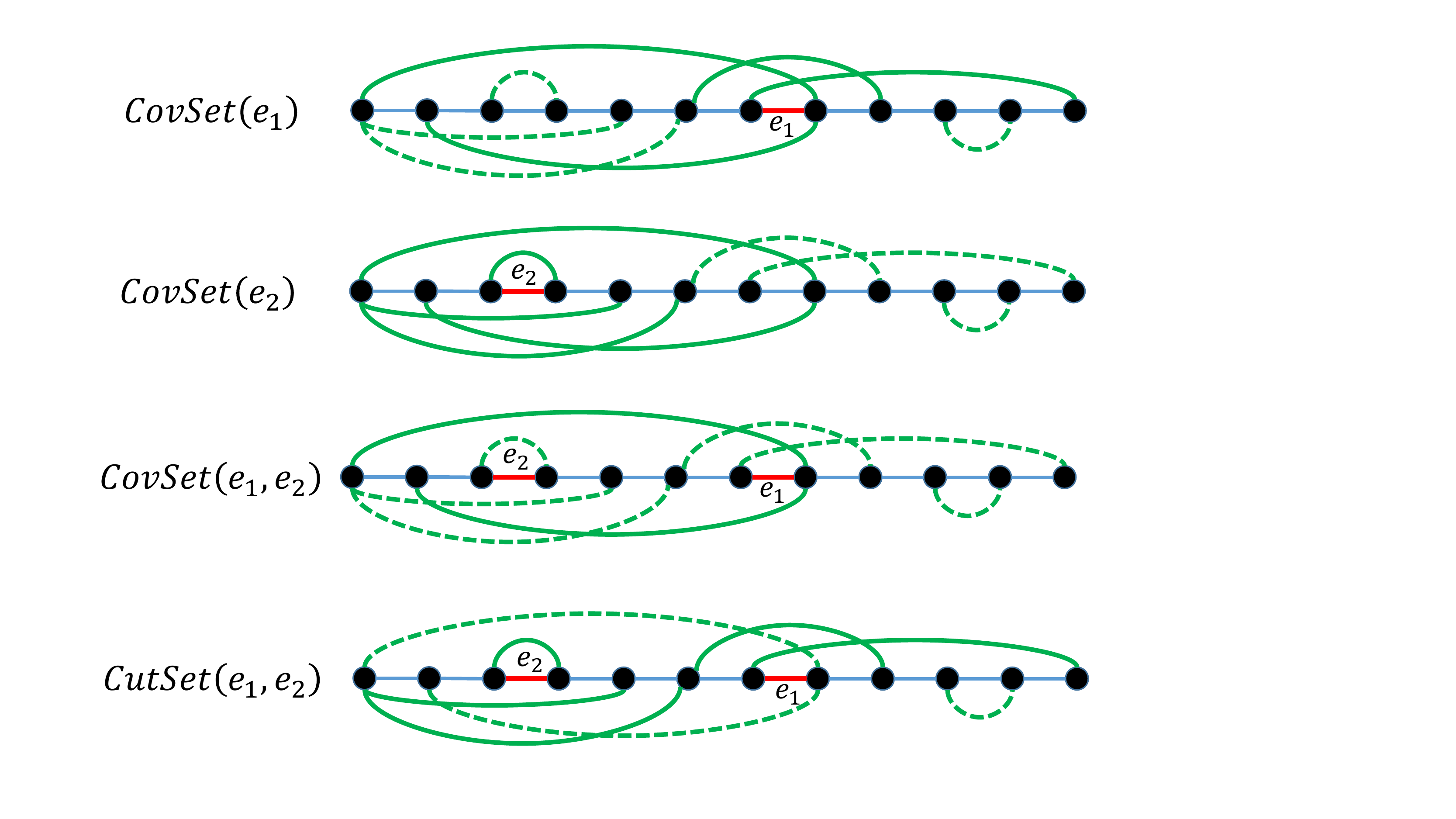}
    \caption{\small Examples of the notion of coverage, the tree $T$ is the central path in each figure. Green edges are non-tree edges. In each figure, the bold edges are the non-tree edges that cover the red tree edges. In the bottom figure, the bold edges represent the edges of $\cut(e_1,e_2)$, which are precisely edges that cover one of $e_1,e_2$, but not both.} 
    \label{fig:cover_and_cut_examples}
\end{figure}

\begin{claim} \label{claim_subtree}
Let $x$ be an edge that covers the tree edge $e=\{v,p(v)\}$, then $x$ has exactly one endpoint in the subtree $T_v$ rooted at $v$.
\end{claim}

\begin{proof}
This follows as removing $e$ from the tree leaves $T_v$ as one of the connected components. Any tree path that contains $e$ must have exactly one of its border vertices in this component, which shows that any edge that covers $e$ must have one endpoint in $T_v$. 
\end{proof}

We next show that the cut value can be expressed easily using the cover values of the related edges, this would be later very useful in our algorithm, when we compute the cover values in order to compute the cut value. The proof is based on showing that the edges that cross the cut defined by two edges $e,e'$ are exactly the edges that cover exactly one of $e,e'$.

\claimCov*

\begin{proof}
In order to prove the claim, it suffices to show that $\ecut(e,e')=\ECov(e)\triangle \ECov(e')$. Here $\triangle$ represents the symmetric difference between the sets. First we show that $\ecut(e,e') \subseteq \ECov(e)\triangle \ECov(e')$. To this end, let $x=\set{u,v}\in \ecut(e,e')$ be any edge in the set $\ecut(e,e')$. The unique path in $T$ between $u$ and $v$ must cross the 2-respecting cut defined by $(e,e')$, and since the only tree edges that cross said cut are $e,e'$, we can deduce that $e$ or $e'$ are on the unique path in $T$ between $u$ and $v$. Note that only one of $e,e'$ can be on this unique path, since otherwise one would get that $u,v$ are both on the same side of the cut. Thus, by definition of covering, the edge $x$ covers  exactly one of $e,e'$, and thus $x\in \ECov(e)\triangle \ECov(e')$.

Next we show that $\ECov(e)\triangle \ECov(e') \subseteq \ecut(e,e')$. Let $x\in \ECov(e)\triangle \ECov(e')$, w.l.o.g assume that $x\in \ECov(e),x\not\in \ECov(e')$. Denote $x=\set{u,v}$, since $x\in \ECov(e),x\not\in \ECov(e')$, we deduce that the unique path in $T$ between $u$ and $v$ goes through $e$, but not through $e'$, thus this path crosses the 2-respecting cut defined by $e,e'$ only once. From which we can deduce that $u,v$ are on different sides of the 2-respecting cut defined by $e,e'$. Thus $x=\set{u,v}\in \ecut(e,e')$ as required. 
\end{proof}

We denote by $\cut(e)$ the value of the 1-respecting cut defined by $e$, i.e., the cut obtained after removing $e$ from the tree. It is easy to see that $\cut(e)=\cov(e)$.

\begin{claim} \label{claim_1_respecting}
$\cut(e)=\cov(e)$.
\end{claim}

\begin{proof}
Let $e=\{v,p(v)\}$. As discussed in the proof of Claim \ref{claim_subtree}, removing $e$ from the graph leaves one component that is the subtree rooted at $v$, and the rest of the tree as the second component. The edges that cross the cut are exactly the edges that have exactly one endpoint in each one of the components. These are exactly all edges that cover $e$, as the tree paths defined by these edges move between these components, hence they must include $e$.
\end{proof}

\subsection{LCA labels} \label{sec:lca}

We use the tool of lowest common ancestor (LCA) labels to check easily if a tree edge is covered by some non-tree edge. We use the LCA labels from \cite[see Section 2.3.2]{dory2020} (See also section 5.2 in \cite{censor2019fast}), which adapt the sequential labeling scheme of  \cite{alstrup2004nearest} to the distributed setting. This allows to give any vertex in the graph a short label of $O(\log n)$ bits such that given the labels of two vertices $u,v$, we can infer the label of their LCA just from the labels. During the algorithm, when we send an edge, we always send its labels as well, which allows these computations. The time for computing the labels is $O(D+\sqrt{n}\log^{*}{n})$ as shown in \cite{dory2020,censor2019fast}. We next show that the labels allow to determine if a non-tree edge covers a tree edge. This is also used in \cite{censor2019fast,DBLP:conf/podc/DoryG19}.

\begin{claim} \label{claim_LCA_labels}
In $O(D+\sqrt{n}\log^{*}{n})$ time, we can assign all the vertices in the graph short labels, such that given the labels of a tree edge $e$ and a non-tree edge $x$, we can learn whether  $x$ covers $e$. Additionally, given the labels of two vertices $u,v$ we can deduce $\lca{u,v}$.
\end{claim}

\begin{proof}
Let $e=\{v,p(v)\}$ and $x=\{u,w\}$. From Claim \ref{claim_subtree}, we know that $x$ covers $e$ iff $x$ has exactly one endpoint in the subtree $T_v$ rooted at $v$. This can be easily checked using LCA labels. For any vertex $v'$, $v'$ is in the subtree rooted at $v$ iff $\lca{v',v}=v$, as $v$ is an ancestor of all vertices in $T_v$. Hence, to determine if $x$ covers $e$, we compute $\lca{u,v},\lca{w,v}$ and check whether the answer is $v$ in exactly one of the cases.
\end{proof}

LCA labels are also useful to infer which edges in the graph participate in a 2-respecting cut, as we show next.
\begin{observation} \label{obs:observation_learn_cut}
Given the labels of at most 2 edges $e,e'$ that define a 2-respecting cut, each vertex can learn exactly which of its incident edges cross the cut. This does not require any communication.
\end{observation}

\begin{proof}
This observation follows from LCA checks that $v$ can do (Claim \ref{claim_LCA_labels}). 
First, consider the simple case that the cut is defined by one edge $e$. Then, an edge $e'$ crosses the cut iff it covers $e$ which can be deduced from Claim \ref{claim_LCA_labels}.

We next focus on the case that there are two tree edges $(e, e')$ that the cut respects (i.e., when the 2-respecting cut is an \textit{exact 2-respecting cut}). 

Consider any edge $f= \{u,v\}$ which is incident on $v$. Note that, given $(e,e')$, $v$ can do an LCA check to find out which edges among $e$ and $e'$ are covered by $f$. The edge $f$ takes part in the cut iff $f$ covers \textit{exactly} one edge among $e$ and $e'$---this again can be computed inside $v$ without any communication using Claim \ref{claim_LCA_labels}. Hence $v$ can infer which edges incident to it are in the cut by local computation.
\end{proof}

For two tree edges $e=\set{p(v),v},e'=\set{p(v'),v'}$, denote by $\lca{e,e'}$ the vertex $v^*$ such that $v^*=\lca{v,v'}$. Note that by Claim \ref{claim_LCA_labels}, given the labels of $e,e'$, one can deduce $\lca{e,e'}$, since the lemma allows one to deduce $\lca{v,v'}$.

\subsection{Useful notation}

We next define a subtree $T(P)$ related to a path $P$, this is later useful for our algorithm.
We always assume that the corresponding spanning tree is rooted at a root vertex which we denote as $r$. For a path $P$ between an ancestor $r_P$ and a descendant of it in the tree, we denote by $T(P)$ the subtree that includes the path $P$, and all the subtrees rooted at vertices in $P \setminus r_P$, and we denote by $T(P^{\downarrow})$ the subtree $T(P) \setminus r_P$. For a tree edge $e=\{u,v\}$, we denote by $\dw{e}$ the tree rooted at $v$, where $v$ is the node farther from the root.  
In this work, whenever we mention a path $P$, we always assume that it is an ancestor-to-descendant path, i.e., the path occurs as a subpath of a root-to-leaf path of $T$.

\section{Tree decompositions}\label{sec:building_blocks}


Before we present the algorithm, we discuss in this section two tree decompositions that are crucial for our algorithm: a fragment decomposition and a layering decomposition.

\subsection{Fragment decomposition}\label{ssec:-decompFragment}

Here we discuss a decomposition of a tree $T$ into edge-disjoint components, each with small size. This is a variant of a decomposition that appeared first in \cite{ghaffari2016near}, and was refined and used also in \cite{Dory18, DBLP:conf/podc/DoryG19}. In these works, the tree is decomposed into $O(\sqrt{n})$ edge-disjoint fragments of diameter $O(\sqrt{n})$. In our variant, we also make sure that the \emph{size} of each fragment is $O(\sqrt{n})$.

The fragment decomposition is defined by $\nfrag = O(\sqrt{n})$ many tuples, each of the form $t_F = (r_F, d_F)$, with the following properties (see Figure \ref{fig:fragment_decomposition_zoom_in_on_a_single_fragment}):
\begin{enumerate}
    \item Each tuple $t_F=(r_F, d_F)$ represents an edge-disjoint fragment (subtree) $F$ of $T$ rooted at $r_F$ with diameter $\dfrag = O(\sqrt{n})$ and size  $\sfrag = O(\sqrt{n})$. The vertex $r_F$ is an ancestor of all vertices in the fragment $F$ in $T$. 
    
    \item Each fragment $F$ has a special vertex $d_F$ which is called the unique descendant of the fragment. The unique path between $r_F$ and $d_F$ is called the \textit{highway} of the fragment. Each fragment has a single highway path. The vertices $r_F$ and $d_F$ are the only two vertices of the fragment $F$ which can occur in other fragments.
    
    \item All edges that are not part of the highway, are called \emph{non-highway} edges. Each non-highway path is completely contained inside a single fragment.  
    
    \item Each edge of $T$ takes part in exactly one fragment $F$.
\end{enumerate}


We show in Lemma \ref{lemma:strong_decomp_construct} how to compute the fragment decomposition employed in this paper.
We denote by $\cF$ the set of fragments in the fragment decomposition.
\begin{figure}
    \centering
    \includegraphics[scale=0.3]{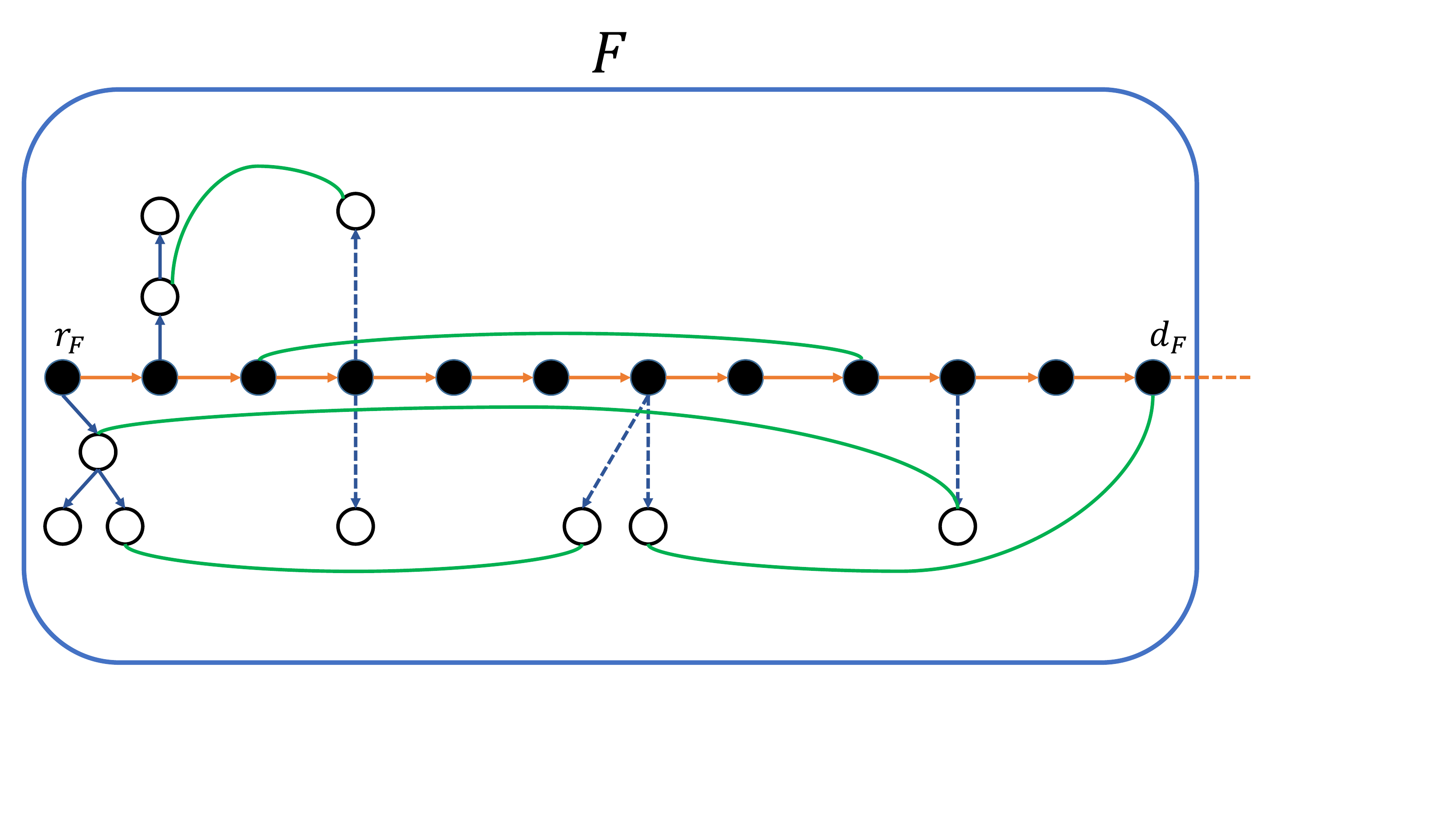}
    \caption{\small The internal topology of a single fragment. The path comprised of black nodes is the highway of the fragment $F$. Orange edges denote highway edges. Blue edges denote non-highway edges. Empty nodes are nodes that are adjacent to no highway edge. Dotted edges indicate arbitrarily long paths. Green undirected edges correspond to possible \emph{non-tree} edges in the graph $G$. }
    \label{fig:fragment_decomposition_zoom_in_on_a_single_fragment}
\end{figure}
Given such a decomposition of a tree $T$, we can associate a virtual \textit{skeleton tree} $T_S$ naturally to the decomposition in the following way (see Figure \ref{fig:skeleton_tree_example}):
\begin{enumerate}
    \item $T_S$ has $\nfrag + 1$ many vertices: For each vertex that is either $r_F$ or $d_F$ in one of the fragments, there is a vertex in $T_S$,
    
    \item The edges in $T_S$ correspond to the highways of the fragments, i.e., there is an edge $\{u,v\}\in T_S$ where $u$ is a parent of $v$ iff $u=r_F$ and $v=d_F$ for some fragment $F$.
\end{enumerate}

\begin{figure}
    \centering
    \includegraphics[scale=0.3]{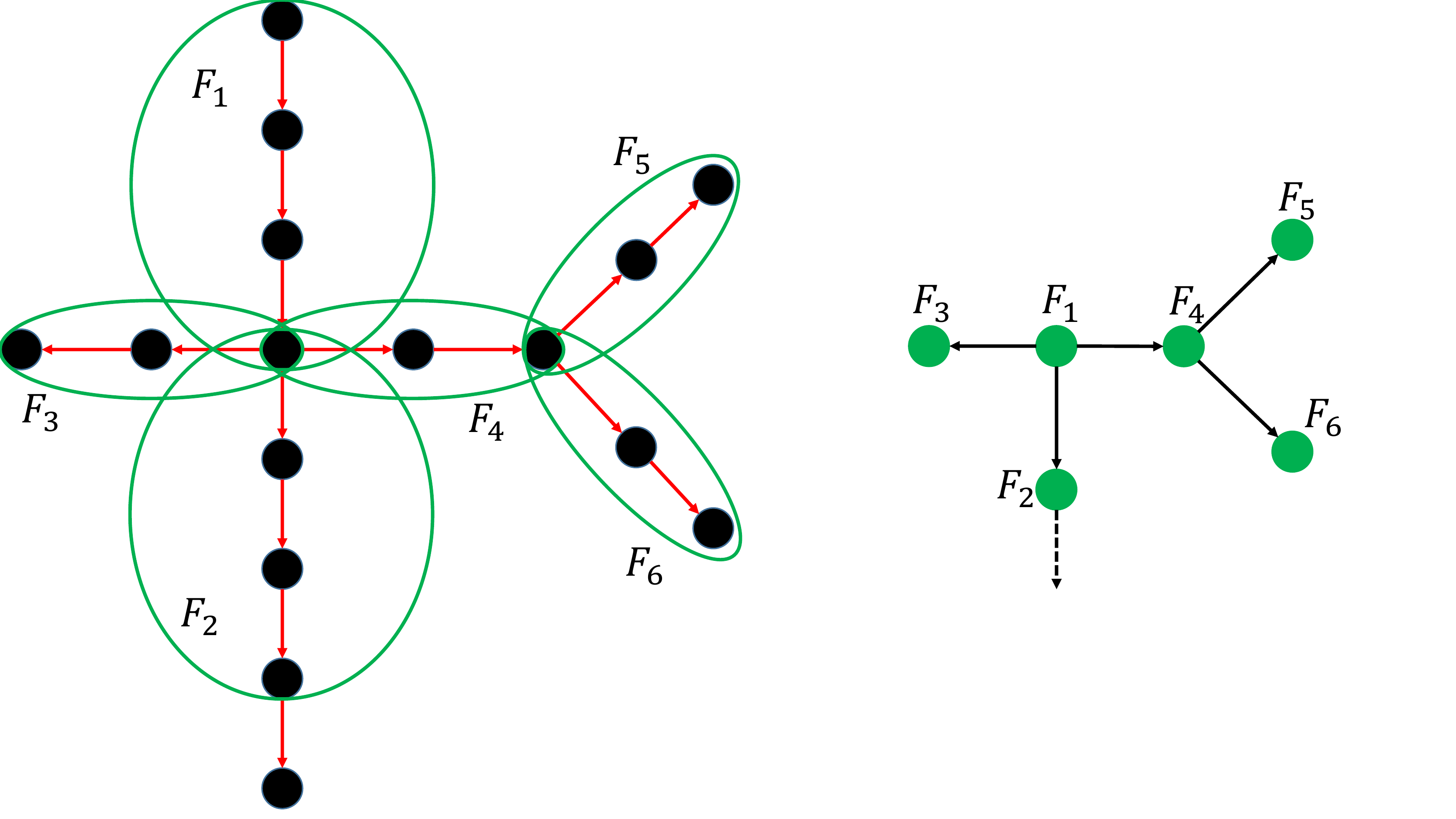}
    \caption{\small An example of a skeleton tree. Fragments are circled in green with non-highway edges and vertices omitted. The skeleton tree has a unique vertex corresponding to each fragment, and two vertices of the skeleton tree are connected by an edge if the corresponding fragments share a (highway) vertex.}
    \label{fig:skeleton_tree_example}
\end{figure}


In Appendix \ref{sec:appen_building_blocks}, we prove the following lemma.
\begin{restatable}{lemma}{FragDecomp}\label{lemma:strong_decomp_construct}
Given a graph $G=(V,E)$ and a rooted spanning tree $T$ of $G$, one can compute in $\tO(\sqrt{n}+D)$ rounds a fragment decomposition of $G$ with $\nfrag=\dfrag=\sfrag=O(\sqrt{n})$.  In particular, each vertex $v$ learns the following information about each fragment $F$ it belongs to:

\begin{enumerate}
    \item The identity $(r_F,d_F)$ of the fragment $F$.
    \item The complete structure of the skeleton tree $T_S$.
    \item All the edges of the unique path connecting $v$ and $r_F$, and also the edges of the unique path connecting $v$ and $d_F$.
    \item All the edges of the highway of the fragment $F$.
\end{enumerate}
\end{restatable}
 
 The full proof of this lemma can be found in Appendix \ref{sec:appen_building_blocks}. On a high level, the construction is divided into two main steps. First, we  decompose the tree into $O(\sqrt{n})$ edge disjoint components each of size $O(\sqrt{n})$. Then, we further breakdown each component to enhance them with the properties depicted in bullets 2-4 at the beginning of the section, without forfeiting the size and number of components guarantees. This second part follows \cite{ghaffari2016near, Dory18}.

\paragraph{Routing Information.} The following lemma shows an efficient way to route information from non-tree edges to tree edges they cover. This allows for example to compute efficiently the value $\cov(e)$ for all edges $e$. For a proof see Claim 2.6 in \cite{dory2020}.
In addition, in \cite{ghaffari2016near,Dory18} this lemma is proven for the special case where each tree edge wants to learn about the \emph{best} edge that covers it according to some criterion (see Section 4.2 in \cite{ghaffari2016near}, and Section 3.1 (II) in \cite{Dory18}). The exact same argument works also for the more general case when each tree edge wants to compute some commutative aggregate function of the edges that cover it (for example, the sum of their weights).
Denote by $C_t\subseteq E$ the set of non-tree edges that cover a given tree edge $t$. 

\begin{lemma}(\cite{dory2020})\label{lemma:routing_lemma_aggregate}
Assume that each non-tree edge $e$ has some information $m_e$ of $O(\log n)$ bits, and let $f$ be
a commutative function with output of $O(\log n)$ bits. In $O(D+\max\{\dfrag, \nfrag\})$ rounds, each tree edge $t$, learns the output of $f$ on the inputs $\set{m_e}_{e\in C_t}$. 
\end{lemma}

From this lemma, one can deduce the following claims.

\begin{claim}\label{claim:learn_cov_value} \label{claim_learn_cov}
All tree edges $e$ can learn their cover values $\cov(e)$. This is done simultaneously for all edges in $O(D+\sqrt{n})$ rounds.
\end{claim}

\begin{proof}
This follows from Lemma \ref{lemma:routing_lemma_aggregate}, where the information for non-tree edges is their weight, and the function is sum. After applying the lemma, each tree edge learns the total cost of non-tree edges that cover it. Adding to it to weight of $e$ gives $\cov(e)$.
\end{proof}

\begin{claim}\label{claim:route_info}
Assume that for each tree edge $e$, there is a unique non-tree edge $e'\in \ECov(e)$ that wants to send $t$ bits of information to $e$. Then this routing can be done in parallel for all tree edges in $O((D+\max\set{\dfrag,\nfrag})\cdot \frac{t}{\log n})$ rounds.
\end{claim}
\begin{proof}
Each non-tree edge holding information will mark itself and the information, and the function $f$ will be the function that chooses arbitrarily (e.g. the first) marked input and outputs it. Since each tree edge $e$ has a unique edge $e'\in \ECov(e)$ that is marked, applying Lemma \ref{lemma:routing_lemma_aggregate}  $ \frac{t}{\log n}$ times suffices for $e$ to learn the $t$ bits of information that $e'$ holds. 
\end{proof}

\subsection{Layering decomposition}\label{ssec:-layerDecomp}

 Here, we present a decomposition of the edges of a given tree $T$ into $O(\log n)$ layers. This is also known as the \textit{bough decomposition} in some literature \cite{Karger00, geissmann2018parallel}. Such a decomposition was previously employed in the context of distributed computing by $\cite{DBLP:conf/podc/DoryG19}$. We borrow the following definition from Karger.

\begin{definition}[Bough]
A bough is a maximal path starting at a leaf and traveling upwards until it reaches a vertex with more than one child, i.e., a junction vertex. 
\end{definition}

\begin{figure}
    \centering
    \includegraphics[width=\textwidth]{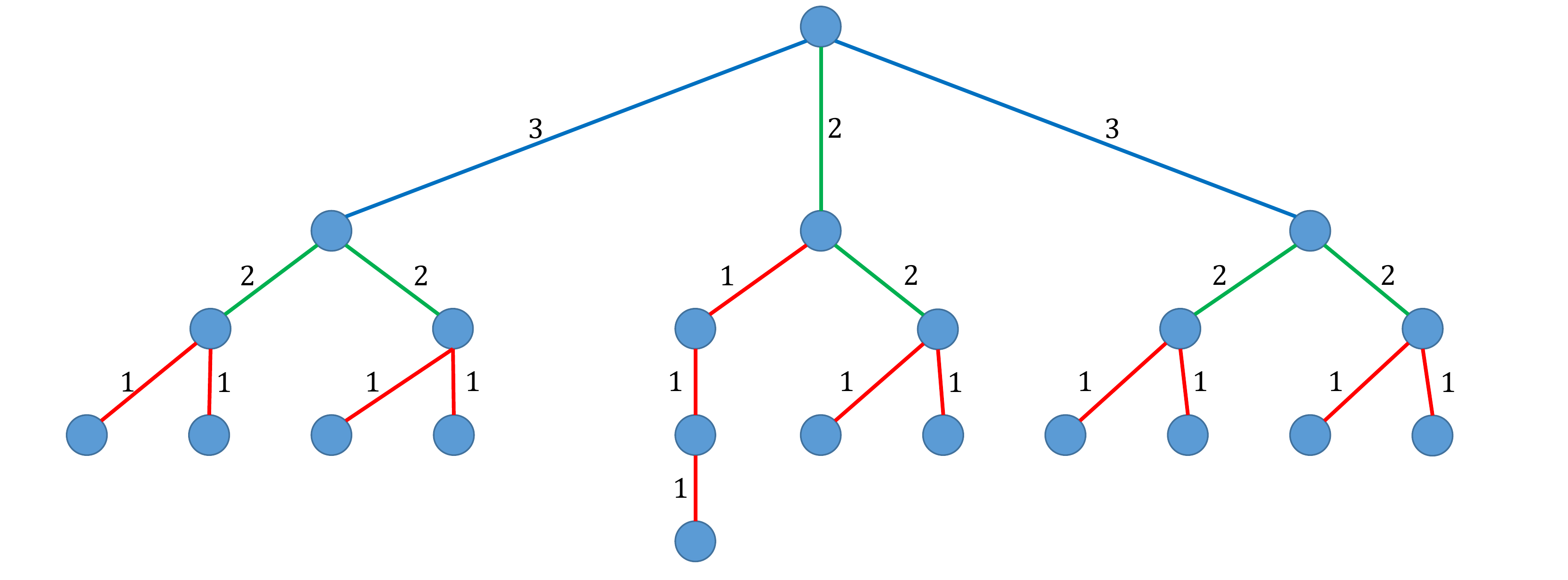}
    \caption{\small Example of a layering decomposition on a given tree. \textcolor{red}{Red} edges are edges of layer 1, \textcolor{green!50!black}{green} edges are of layer 2, and \textcolor{blue}{blue} edges are of layer 3.}
    \label{fig:layering_exmaple}
\end{figure}

We are interested in the following layering algorithm. Given a graph $G$ and a spanning tree $T$ of $G$ rooted at a root vertex $r$, the layering algorithm can be described as follows:

\begin{description}
\item[Initialization:] Start with $T_0 = T$.
\item[Round $i$ description:] In round $i$, do the following:
\begin{itemize}
    \item Consider all boughs of $T_{i-1}$ and include all edges of such boughs in $E^i$. 
    \item Contract these boughs of $T_{i-1}$ to obtain $T_i$.
\end{itemize}
\item[Stop condition:] Continue until $T_i$ consists of only root vertex $r$.
\end{description}
\paragraph{}
For any given $i$, we call the edges in $E^i$ the edges in \emph{layer} $i$ (See Figure \ref{fig:layering_exmaple}). An immediate observation from this procedure is as follows.
\begin{observation}\label{obs:-no_down_layer}
Given a tree edge $e\in T$, denote its layer by $i$. Then, all edges in $\dw{e}$ are in layer at most $i$. Furthermore, at most one edge adjacent to $e$ which is in $\dw{e}$ is in layer exactly $i$. Furthermore, denote by $j$ the highest layer of an edge in $\dw{e}$ which is adjacent to $e$, then $j=i$ iff there is exactly one edge of layer $j$ adjacent to $e$ in $\dw{e}$. 

\end{observation}

Next, we bound the number of layers in this decomposition.
\begin{claim}\label{claim:bound_number_layers}
The number of layers produced by the above procedure is $L=O(\log n)$.
\end{claim}
\begin{proof}
 Denote by $\ell_i$ the number of leaves in the tree $T_i$. A key observation is that for all $1\leq i\leq L-1$, it holds that $2\ell_{i+1}\leq \ell _i$. This is since each leaf $v$ in $T_{i+1}$ is an ancestor of at least 2 leaves in $T_i$. Since $v$ was a junction in $T_i$. Now since $\ell_1\leq n$, and the fact that any tree has at least $2$ leaves, we deduce that after $L=O(\log n)$ rounds we are left with an empty graph, as required.
\end{proof}



\subsection{Combining layering with fragment decomposition}\label{ssec:way_todo_layering}

In our algorithm, we use two layering decompositions, one for non-highways and one for highways. We next explain this in detail.
\begin{enumerate}
    \item \textbf{Layering for non-highways.} We decompose the non-highways in each fragment to layers, according to the layering algorithm described above. The union of all such layerings computed is the layering for non-highways. I.e., edges in layer $i$ include non-highway edges in layer $i$ in all different fragments.
    \item \textbf{Layering for highways.} We decompose the highways into layers by simulating the layering algorithm in the skeleton tree. Here, we ignore completely all non-highway edges, and just run the layering algorithm in the skeleton tree that its edges correspond to highways. Since all vertices know the complete structure of the skeleton tree this can be simulated locally by each vertex without communication. As each fragment highway corresponds to one edge in the skeleton tree, it follows that by the end of the computation each fragment highway has a layer number. 
\end{enumerate}

The above description results in the following claim.

\begin{lemma}\label{lemma:layer_decomposition_highway}
We can compute a layering for the highways without communication. By the end of the computation, all vertices know the layering.
\end{lemma}
The above Lemma holds simply because all vertices know the topology of the skeleton tree by Lemma \ref{lemma:strong_decomp_construct}. 
We next state that we can compute the layering for the non-highways. The full of proof of the following lemma is deferred to appendix \ref{sec:appen_building_blocks}.

\begin{restatable}{lemma}{LayerDecompNonHighway}\label{lemma:layer_decomposition_inside_fragment_nonhighway}
We can compute a layering for the non-highways in $O(\dfrag)$ time. By the end of the computation all vertices know the layer numbers of non-highway edges adjacent to them. 
\end{restatable}

This is achieved by a simple aggregate computation inside each fragment in parallel, along the non-highway trees of each fragment from the leaves to the highway path.



We also define the notion of a \emph{maximal} path of layer $i$.
\begin{definition}\label{def:maximal_path_layer_i}
Given a (non) highway path $P$ in a tree $T$, we say that $P$ is a maximal path of layer $i$ if the following holds.
\begin{enumerate}
    \item $P\subseteq E^i$ according to either the skeleton decomposition or the non-highway decomposition.
    \item for all (non) highway paths $P'$ such that $P\subsetneq P'$, it holds that $P'\not\subseteq E^i$.
\end{enumerate}
\end{definition}




\subsection{Information of edges}\label{sec:-info_tree_dge}

The goal of this section is to explain in detail, what is the information each edge in the tree holds. Later, when we introduce the sampling procedure, and say that an edge $e$ learns about an edge $e'$, it learns not only the Id of $e'$, but also all the information that $e'$ holds about the fragment and layer decompositions. Furthermore, $e$ might need to spread said information to edges in its vicinity, depending on the properties of $e$. All of these issues are addressed in detail in this section.

From a high level perspective, our algorithm requires each tree edge $e$ to know whether $e$ in on a highway or not, the fragment of $e$, and its layer in it's respective layering decomposition (see Section \ref{ssec:-layerDecomp}). This is captured in the following definition. 

\begin{definition}[Information of edge]\label{def:}
The information of a tree edge $e$, denoted by $\info(e)$, consists of the following:
\begin{enumerate}
    \item The id of the edge $e$.
    \item The value $\cov(e)$ and $|\ECov(e)|$.
    \item Whether $e$ is a highway edge, or a non-highway edge.
    \item $e$'s layer in its respective layering decomposition (see Section \ref{ssec:-layerDecomp}). 
    \item  The id of its own fragment.
\end{enumerate}
For non-tree edge $e=\{u,v\}\in E$, one can associate the information $\info(e_1)\cup \info(e_2)\cup \set{Id(e)}$ where $e_1=\{p(v),v\},e_2=\{p(u),u\}$ with $e$. Here $p(v),p(u)$ are the parent vertices of $v,u$ in the tree. Denote this information by $\info(e)$ for a non-tree edge.
For any edge $e$, we denote by $|\info(e)|$ the amount of bits required in order to store $\info(e)$.
\end{definition}

More formally, our goal in this section is to prove the following theorem. In the following theorem, $E^i$ refers to the set of non-highway edges of layer $i$. (See Section \ref{ssec:way_todo_layering}, Lemma \ref{lemma:layer_decomposition_inside_fragment_nonhighway})
The proof of the following theorem is deferred to Appendix \ref{sec:appen_building_blocks}.

\begin{restatable}{theorem}{LearnInfoEdge}\label{thm:-info_theorem_tree_edge}
Consider a rooted tree $T$ of a graph $G$, a layering decomposition $E^1,...,E^\ell$, $\ell=O(\log n)$, as described in Section \ref{ssec:way_todo_layering}, and a fragment decomposition with parameters $\nfrag,\sfrag$. In $\Tilde{O}(\nfrag+\sfrag+D)$ rounds, each tree edge $e$ can learn $\info(e)$. Furthermore, it holds that $|\info(e)|=O(\log n)$. 
\end{restatable}

\begin{observation}\label{obs:knowing_allfragments_from_one}
Each tree edge $e$ knows the id and layers of all fragments in the path from the root to it. In general, given a root to descendant path $P$ that ends in the edge $e'$, every tree edge $e$ can deduce all the fragments that $P$ intersects from knowing the fragment of $e'$.
\end{observation}
\begin{proof}
 Note that $e$ knows the entire topology of the skeleton tree (See Lemma \ref{lemma:strong_decomp_construct}), and computed the layering of it internally using Lemma \ref{lemma:layer_decomposition_highway}, hence $e$ knows the layers of all fragments that intersect $P$, since it also knows it's own fragment. The same argument holds for other root to descendant paths, assuming $e$ knows the fragment in which the path ends.
\end{proof}




\section{Finding, bounding, and routing interesting paths}\label{sec:lemma_and_samp}
 Given a spanning tree $T$ of $G$, a trivial way of finding a minimum 2-respecting cut is to compute $\cut(e,e')$ for every pair of edges $(e,e')$ in $T$ which requires $O(n^2)$ many comparisons. As observed in \cite{mukhopadhyay2019weighted}, many of these comparisons are unnecessary, as many such pairs cannot possibly yield a minimum 2-respecting cut. To this end, \cite{mukhopadhyay2019weighted} has formulated a necessary condition for such pairs to be a potential candidate for a minimum 2-respecting cut which we will recap briefly now. First, note that $\cut(e,e') = \cov(e) + \cov(e') - \cov(e,e')$ as stated in Claim \ref{claim_cover}. Also assume that it is easy to compute the values of $\cov(e)$ for all $e \in T$. Note that $\cov(e)$ actually represents the value of a 1-respecting cut which cuts the tree edge $e$, as shown in Claim \ref{claim_1_respecting}. For $(e,e')$ to be a minimum exact-2-respecting cut (i.e., cut value smaller than any 1-respecting cut), it needs to happen that
 \[
     \cut(e,e') = \cov(e) + \cov(e') - \cov(e,e') < \min\{\cov(e), \cov(e')\}. \tag{$\ast$} \label{eq:cover}
 \]
 In other words, the 2-respecting cut defined by $(e,e')$ should be smaller than the cuts defined by $e$ and $e'$ individually. Reorganizing the previous inequality we get $\cov(e,e') > \frac 1 2 \cdot \max\{\cov(e), \cov(e')\}$. This means that, for a tree edge $e$, the potential pairings $(e,e')$ which can yield a 2-respecting cut smaller than $\cov(e)$ are the pairing for which $\cov(e,e') > \cov(e)/2$. We denote this event as $e$ being \textit{interested} in the edge $e'$. In this section we expound the notion of interesting edges and extend it to the notion of interesting paths. There are three main subsections in this section: In the first part (in Section \ref{ssec:-samplingProcedure}) an algorithm that finds for each edge a set of paths that includes the path that edge is interested in. Then, in section \ref{ssec:-interesting_lemma}, we prove our \emph{interesting path counting lemma} that shows that in fact, intuitively, the number of paths interested in one another is small, and can be bounded from above by a poly-logarithmic factor in $n$. Then, in section \ref{ssec:path_pruning}, we show how to turn the information obtained in section \ref{ssec:-samplingProcedure}, into knowing the paths each edge needs to know for the algorithm. Additionally in that section, we connect our combinatorial lemma of bounding the number of paths interested in one another to the algorithmic building blocks we use in the algorithm. Furthermore, we show how to route  the paths that each vertex needs to know for the algorithm to said vertex.


\subsection{Interesting edges and paths}\label{sec:interesting-defn}

 We begin with basic definitions and observations. Most of the definitions and observations used in this section are inspired by similar definitions previously stated in \cite{mukhopadhyay2019weighted}. We, however, introduce simpler notation and terminology for this work.

\begin{definition}[Interesting edge]\label{def:interest-edges}
Given two tree edges $e,e'\in T$, we say that $e$ is interested in $e'$ if $\cov(e,e')>\frac{\cov(e)}{2}$.
\end{definition}

This definition, along with the observation made previously in Equation (\ref{eq:cover}), immediately gives the following claim.
\begin{claim} \label{clm:min-cut-interesting}
A pair of tree edges $(e,e')$ is a candidate for exact-2-respecting min-cut (i.e., has cut value smaller than all 1-respecting cuts) if $e$ and $e'$ are interested in each other.
\end{claim}

We now proceed to define what is meant by an edge $e$ being interested in a path. But first, we define the notion of orthogonality between edges of $T$ and, subsequently, between paths of $T$.


\begin{definition}[Orthogonal edges]\label{def:orthogonal}
Given two tree edges $e,e'$, we say that $e,e'$ are orthogonal if they are not on the same root to leaf path, and we denote $e\perp e'$.
\end{definition}

Combining Definition \ref{def:interest-edges} and \ref{def:orthogonal}, we can make the following observation.

\begin{observation}\label{obs:interested-edges}
Given a tree edge $e\in T$, let $e',e''$ be tree edges such that there is no simple path in the tree $T$ in which $e,e',e''$ all take part, then it cannot hold that $e$ is interested in both $e'$ and $e''$. Furthermore, $\ECov(e,e')\cap \ECov(e,e'')=\emptyset$.

\end{observation}
\begin{proof}
First of all, Observe that $\ECov(e,e')\cap \ECov(e,e'')=\ECov(e)\cap \ECov(e')\cap \ECov(e'')$. Now, if $ECov(e)\cap \ECov(e')\cap \ECov(e'')\neq \emptyset$, this means that for some non-tree edge $e'=(u,v)$, the unique path in $T$ from $u$ to $v$ includes all of $e,e',e''$. This is a contradiction, thus we deduce that  $\ECov(e,e')\cap \ECov(e,e'')=\ECov(e)\cap \ECov(e')\cap \ECov(e'')=\emptyset$.

 This means that $\cov(e) \geq \cov(e,e') + \cov(e,e'')$. Suppose $e',e''$ are such that $e$ is interested in both $e$ and $e''$, i.e., $\cov(e,e')>\frac{\cov(e)}{2}$, and $\cov(e,e'')>\frac{\cov(e)}{2}$. This means $\cov(e,e') + \cov(e,e'') > \cov(e)$ which is an immediate contradiction. 
\end{proof}

We extend the definition of orthogonal edges to orthogonal paths which we define below.
\begin{definition}[Orthogonal paths]\label{def:orthogonal_paths}
Given two ancestor to descendant paths $P,P'$, we say that $P$ and $P'$ are orthogonal and denote $P\perp P'$ if for all pairs of edges $e\in P,e'\in P'$, it holds that $e\perp e'$.
\end{definition}




Note that if $e''$ is on the unique tree path between $e$ and $e'$, then $\cov(e,e'')\geq \cov(e,e')$. This is true since all edges that cover both $e$ and $e'$, also cover $e''$, by definition. Thus, one can make the following observation

\begin{observation}\label{obs:-int_ancestors}
Given an edge $e$, if $e$ is interested in some edge $e'$, Then $e$ is interested in all the edges in the tree path from $e$ to $e'$. 

\end{observation}

We also make the following observation.
\begin{observation} \label{obv:path_disjoint}
If $P$ and $P'$ are two orthogonal paths in the tree, then $T(P)$ and $T(P')$ are edge disjoint, and $T(P^{\downarrow})$ and $T(P'^{\downarrow})$ are disjoint.
\end{observation}

At this point we introduce another notation which is defined as follows. 

\begin{definition}\label{def:cov-samp-edge-path}
Given a tree edge $e$ and an ancestor-to-descendant path $P$ which is either orthogonal to $e$ or completely above $e$ in a root-to-leaf path which contains $e$ or completely inside $e^\downarrow$, we define $\ECov(e,P)$ to be the set of edges $f = \{u,v\} \in E(G)$ such that the unique $u$ to $v$ path in $T$ contains both $P$ and $e$ and $\cov(e,P)$ to be the cumulative weight of the edges of the set $\ECov(e,P)$.
\end{definition}

 Note that when $P$ is orthogonal to $e$, then $f$ covers both $e$ and the \textit{lowest} edge of $P$, whereas when $P$ occurs as an ancestor of $e$, $f$ covers both $e$ and the \textit{highest} edge of $P$.  Because of Observation \ref{obs:-int_ancestors}, we extend Definition \ref{def:interest-edges}  to the following:

\begin{definition}[\Interest path]\label{def:-interest_path_cover}
Given a tree edge $e$, and some ancestor to descendant path $P$ in the tree as in Definition \ref{def:cov-samp-edge-path}, we say that $e$ is interested in $P$ if $\cov(e,P) > \cov(e)/2$, and denote the set of all such paths as $\interesting(e)$. Given two ancestor to descendant paths $P_1,P_2$ in the tree, we say that $P_1$ is interested in $P_2$ if there is an edge $e\in P_1$ such that $P_2\in \interesting(e)$. For an ancestor to descendant path $P$ in the tree, we denote by $\interesting(P)$ the set of ancestor to descendant paths in the tree that $P$ is interested in.

\end{definition}

Note that apart from the path above $e$, there is a unique maximal ancestor-to-descendant path $P$ in which $e$ is interested in. The uniqueness comes from the fact that, for such path $P$, $\cov(e,P) > \cov(e)/2$. Having more than one such path will result in a contradiction because $\cov(e,P)$ of each such path $P$ contributed more than half of the value of $\cov(e)$. This means that all edges $e'$ that $e$ is interested in belong to the unique maximal path $P$ that $e$ is interested in. Hence, similar to Claim \ref{clm:min-cut-interesting}, we can make the following claim.
\begin{claim}\label{clm:clm:min-cut-interesting-path}
A pair of tree edges $(e, e')$ is a candidate for exact-2-respecting min-cut (i.e., has cut value smaller than all 1-respecting cuts) if $e'$ is in a path $P' \in \interesting(e)$ and $e$ is in a path $P \in \interesting(e')$. 
\end{claim}

\begin{proof}
If $(e,e')$ is a candidate for exact-2-respecting min-cut, then by Claim \ref{clm:min-cut-interesting}, it implies that $e$ and $e'$ are interested in each other. As discussed before, all edges which $e$ is interested in belong to a path $P' \in \interesting(e)$ and similarly all edges $e'$ is interested in belong to a path $P \in \interesting(e')$. Hence the claim follows.
\end{proof}

\subsection{Finding interesting paths}\label{ssec:-samplingProcedure}

 The main lemma of this section is the following sampling lemma in which each tree edge gets to know a small set of paths, one of which is an \textit{interesting} path w.r.t. $e$. \footnote{Readers familiar with standard sampling techniques can draw some similarities with that of Karger \cite{karger1999random}. Our sampling differs in two ways: (i) the sampling probability depends on the value of $\cov(e)$ for different tree-edge $e$ where as Karger samples each edge with similar probability, and (ii) because this is a distributed implementation, one needs to be careful about how the information of the sampled edges is routed to the vertices responsible for computation.} We will abuse definition and, instead of a vertex knowing another vertex, we will denote a tree-edge knows another tree-edge (or path in the tree) with the assumption that it is a vertex of the tree-edge which does the computation. More formally, let the vertices agree on a total ordering of the set $V$ apriori and let $e=\{u,v\}$ 
be an edge with $u \prec v$ in that ordering of $V$. Unless specified otherwise, any computation purportedly done by $e$ is actually done by $u$. 

\begin{restatable}{lemma}{FindIntSample} \label{lemma:find-int-samp}
There is a distributed sampling procedure which takes $\tO(\sqrt n + D)$ rounds in which every tree-edge $e$ learns about a set of paths $\intpot{e}$ such that with high probability:
\begin{enumerate}
    \item Any path $P$ with $\cov(e,P) > \cov(e)/2$ is in $\intpot{e}$, and \label{item:samp-1}
    \item Any path $P' \in \intpot{e}$, has $\cov(e,P') \geq \cov(e)/6$. \label{item:samp-2}
\end{enumerate}

\end{restatable}

We give an overview of the sampling procedure here. For a more detailed proof of the lemma, see Appendix \ref{sec:appen-samp}.  The idea is simple: we treat every weighted edge $e$ with weight $w(e)$ to be $w(e)$ many parallel unweighted edges. This way, any weighted graph $G$ is viewed as a \textit{unweighted graph with multi-edges}. Each edge $e$ samples $O(\log n)$ distinct non-tree edges from the set $\ECov(e)$, denoted by $\ECovs(e)$, and whose total weight is denoted by $\covs(e)$. Then, each edge $e$ makes decisions about other interesting paths based on the sampled edges, i.e., $e$ declares a path $P$ to be in $\intpot{e}$ if at least 1/3 of the sampled edges by $e$ also covers $P$. Moreover, $e$ also puts the path $P_e$ that contains $e$ into the set $\intpot{e}$ as well. If $e$ can sample these $O(\log n)$ edges uniformly at random from $\ECov(e)$, then by standard concentration argument both conditions of Lemma \ref{lemma:find-int-samp} hold. More formally, we define the set $\intpot{e}$ in the following way, where $P_e$ denotes that path starting from the root to $e$.

\[
\intpot{e}=\left\{P \mid \covs(e,P)\geq \frac{\covs(e)}{3}\right\}\cup P_e.
\]

The question now boils down to how each tree-edge $e$ can sample $O(\log n)$ distinct non-tree edges from $\ECov(e)$. To this end, we assume that every tree edge $e$ knows the value of $\cov(e)$ by Claim \ref{claim_learn_cov}. The sampling procedure runs in $O(\log n)$ iterations. In iteration $j$, only tree-edges $e$ such that $\cov(e) \in [2^{j-1}+1, 2^j]$ sample their corresponding set of non-tree edges---we call these tree edges as \textit{active tree-edges} in that iteration. The point to note here is that, if all active tree-edges want to sample one non-tree edge each from the set of edges that they cover, they can uniformly sample from those sets with probability $2^{-j}$ simultaneously because of their similar cover-values. More concretely, if every non-tree edges \textit{samples itself} with probability $2^{-j}$ in iteration $j$, then every active tree-edge $e$ has at least one sampled non-tree edge in the set $\ECov(e)$ with high probability. Moreover, because all active tree-edges have similar cover-value, for any such edge $e$, the sampled non-tree edge in $\ECov(e)$ is distributed uniformly in the set $\ECov(e)$. The active tree-edges do such sampling enough number of times to get $O(\log n)$ many distinct samples each, and then declare the set of interesting paths based on what they have sampled (as outlined previously).

Of course, implementation of this procedure in the distributed setting has a few more additional details. Note that the information about the sampled non-tree edges needs to reach the vertex of the corresponding active tree-edge $e$ which is responsible for the computation. So we implement each iteration in logarithmic many rounds where, in each round, the active tree-edges get information of one sampled edge each from the sets $\ECov(e)$. In each round, we use Lemma \ref{lemma:routing_lemma_aggregate} to route the information from the \textit{smallest} sampled non-tree edge (w.r.t. some pre-agreed ordering) to the active tree-edge. This information is small enough for us to apply Lemma \ref{lemma:routing_lemma_aggregate}.


Lastly, we prove the following claim, which is extensively used in our algorithm.
\begin{claim} \label{claim_edge}
If $P$ and $P'$ are two paths between an ancestor to a descendant, not in the same root to leaf path, and $P'$ is potentially interested in $P$, then there is an edge between $T(P^{\downarrow})$ and $T(P'^{\downarrow})$.
\end{claim} 
\begin{proof}
Let $e'\in P'$ be some edge that is potentially interested in $P$. Thus some edge between $T(P^{\downarrow})$ and $T(e'^{\downarrow})$ was sampled, thus in particular there is an edge between $T(P^{\downarrow})$ and $T(P'^{\downarrow})$, as required.

\end{proof}

\subsection{Structural lemma  for bounding number of interesting paths}\label{ssec:-interesting_lemma}

So far, we have defined the notion of a potentially interesting edge and a potentially interesting path, for a given tree edge and a given ancestor to descendant path. Our next goal will be to  bound the pairs of paths $P_1,P_2$ such that $P_1$ is potentially interested in $P_2$. Note that although Property \ref{item:samp-2} of Lemma \ref{lemma:find-int-samp} gives us an upper bound on $\intpot{e}$, i.e. the number of paths that the tree edge $e$ is potentially interested in, the best bound one can hope for on $\intpot{P}$ for a given ancestor to descendant path $P$ is $O(n)$. 
Figure \ref{fig:technical_lemma_bad_case_naive} illustrates an example for when a given ancestor to descendant path $P$ can satisfy $\intpot{P}=\Omega(n)$. 

\begin{figure}
    \centering
    \includegraphics[width=\textwidth]{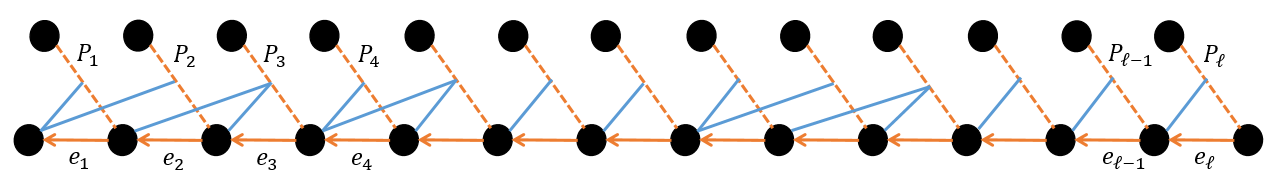}
    \caption{\small The path $P$ in the figure consists of the of the edges $e_1,...,e_\ell$, where the vertex on the right from $e_\ell$ is the root. The dotted orange path labeled by $P_i$ represents an ancestor to descendant path in $\dw{e_{i+1}}$. The blue edges correspond to heavy (weight-wise) clusters of non-tree edges leaving each node to the respective path, that cause the corresponding tree edge to be interested in said path. More formally, it can be the case that $e_i$ is interested in $P_i$ for each $i\in [\ell ]$. Thus the path $P$ is interested in $\ell$ distinct ancestor to descendant paths. The example is complete by noting that it is possible that $\ell=\Omega(n)$.}
    \label{fig:technical_lemma_bad_case_naive}
\end{figure}

However, if one restricts the discussed set of paths $\cP$ to a set that satisfies a specific property, one can bound $\intpot{P}\cap \cP$ very nicely. As Figure \ref{fig:non-splitting} suggests, the kind of paths $\cP$ that we are going to restrict ourselves to is a natural generalization of the set of paths considered in Definition \ref{def:cov-samp-edge-path}. In the discussion that follows, we will denote a path $P_1$ is an ancestor of another path $P_2$ (or, equivalently, $P_2$ is a descendant of $P_1$) if $P_1$ is contained completely inside the path connecting the root to $P_2$. 

\begin{figure}[h!]
    \centering
    \includegraphics[width=0.7\textwidth]{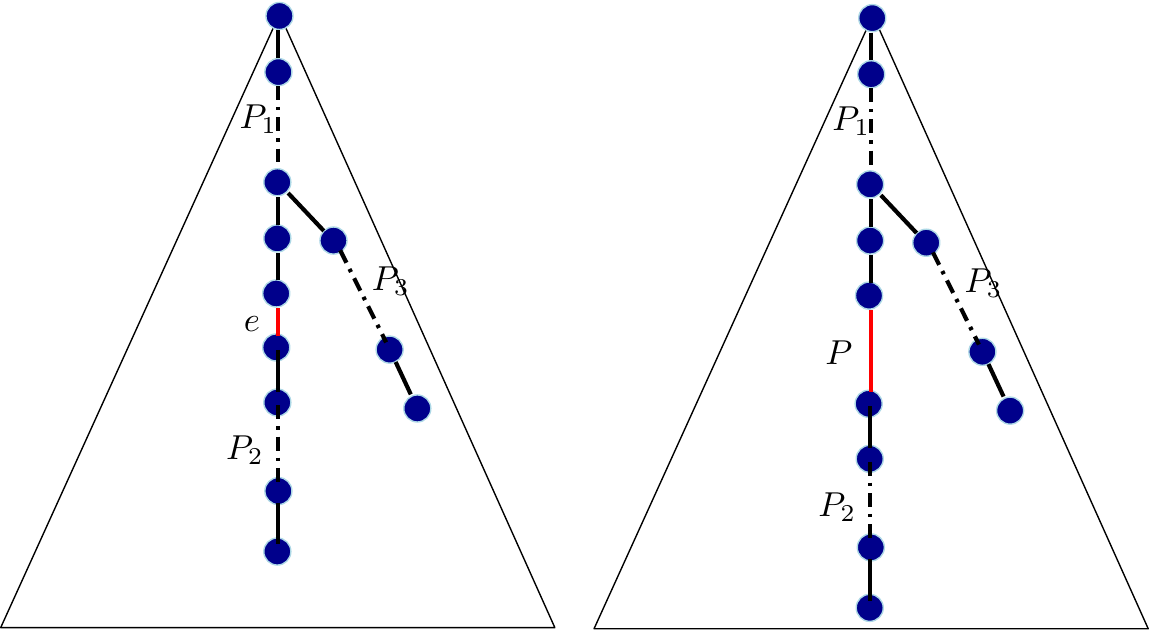}
    \caption{\small Left: $P_1, P_2, P_3$ are the types of paths considered in Definition \ref{def:cov-samp-edge-path} for which we define the notion of $e$ being interested in $P_i$. Right: A generalization. Given $P$ we are interested in non-splitting paths $P_1, P_2, P_3$. $P$ occurs completely in the path from root to $P_3$, and is absent completely from the path from root to either $P_1$ or $P_2$.} 
    \label{fig:non-splitting}
\end{figure}

Consider any path $P$ in $T$, and let us count the number of paths $P'$ such that either $P'$ occurs as a descendant or an ancestor of $P$, or $P'$ is orthogonal to $P$ (i.e., in the path from the root to $P'$, either $P$ is present fully, or $P$ is absent completely---See Figure \ref{fig:non-splitting}; $P'$ is any of $P_1, P_2$ and $P_3$). We denote \textit{$P$ is non-splittable w.r.t. $P'$}. Note that non-splitability is a symmetric property, i.e., if $P$ is non-splittable w.r.t. $P'$, then $P'$ is non-splittable w.r.t. $P$. For a set of paths $\cP$, we say $P$ is non-splitable w.r.t. $\cP$ if $P$ is non-splitable w.r.t. every path in the set $\cP$. Among the paths w.r.t. which $P$ is non-splitable, we want to count the number of paths $P'$ such that $P$ is interested in $P'$. In the structural lemma that follows, we bound the number of such $P'$s given a $P$. We actually show a stronger property: We show that even if we consider such $P'$s  such that $P$ is only \textit{potentially interested} in $P'$, the number of such $P'$s is still bounded. 



\begin{lemma}[Interesting path counting lemma] \label{lemma:-cover_all_or_nothing_paths}
Let $P$ be some ancestor to descendant path in the tree $T$ and let $\cP$ be set of paths such that (i) $P$ is non-splittable w.r.t. $\cP$, (ii) all paths of $\cP$ orthogonal to and ancestor of $P$ are pairwise orthogonal, and all paths of $\cP$ descendant of $P$ are pairwise orthogonal as well. Then, w.h.p., it holds that $B_{path}  \overset{\textrm{def}}{=} |\set{P'\in \cP| P'\in \intpot{P}}|=\potintnum$. 
\end{lemma}




Before going into the proof, let us provide some intuition for it. We focus on the case that the paths in $\cP$ are orthogonal to $P$. For simplicity, let us also assume that we want to count the number of paths that $P$ is \textit{interested} in (instead of $P$ being \textit{potentially interested} in). This simplifies the intuition, and going from \textit{interested} to \textit{potentially interested} is not hard.

At a high-level, if an edge $e \in P$ is interested in some path $P_i \in \cP$, we have that $\cov(e,P_i) \geq \cov(e)/2$. We would next go over the path $P$ from the lowest vertex towards the highest vertex and start counting the number of non-tree edges that start from the subtree rooted at the vertex and end somewhere outside. We will actually count only the number of such edges that covers some orthogonal path in $\cP$, and, depending on this number, will start populating the set $\intpot{P}$. The crucial observation here is the following: each time we reach some edge $e = \{u,v\}$ that is  interested in some new path $P_i$ in $\cP$ (i.e., $P_i$ is different from the paths in $\cP$ that were already seen to be in $\intpot{P}$ and $\cov(e,P_i) \geq \cov(e)/2$), we know that the total weight of such `new' non-tree edges that starts at the sub-tree rooted at $v$ and covers $P_i$ has to be at least the total weight of non-tree edges we have counted so far---otherwise, $e$ would not be interested in $P_i$. So, every time we encounter such an edge $e$ by traversing $P$ from the lowest vertex to the highest vertex, the total weight of non-edges that we count doubles. As the total weight of all non-tree edges are bounded by some polynomial in $n$, this can happen only logarithmic many times. This means that the number of paths in $\cP$ that $P$ is interested in can be at most logarithmic in $n$. More or less the same argument holds for the case when $\cP$ includes paths that are ancestor or descendant of $P$.

As mentioned before, this is a simplification and we want to bound the number of paths in $\cP$ that $P$ is \textit{potentially interested} in. We show next that the simple idea described above can be modified to deal with this case as well.


\begin{proof}
Let us first count the number of $P'$ such that (i) $P$ is potentially interested in $P'$ and (ii) $P$ does not occur in the path from root to $P'$ (i.e., $P'$ is either an ancestor of or orthogonal to $P$).\footnote{To be honest, we are not really interested in the case when $P'$ is an ancestor of $P$ because there cannot be more than one such $P'$ who are orthogonal to each other.} For the rest of the proof, readers are advised to refer to Figure \ref{fig:bounding-interesting}. Let us order the edges of $P$ from the lowest to the highest (i.e., from farthest from the root to the closest to the root) as $e_1,...,e_z$. Consider a path $P_i$ from $\cP$ that some edge in $P$ is potentially intersted in. Note that all edges that are ancestors of this particular edge are also potentially interested in $P_i$. Let $e_i$ be the lowest such edge that is potentially interested in some path $P_i$. Denote by $\cov(e_i,\cP)$ the value $\sum_{P^*\in \cP} \cov(e_i,P^*)$. $\ECov(e_i,\cP)$ is defined similarly.
Next we describe an argument that we repeat several times during the claim.
 
%
\begin{figure}
\centering
\begin{minipage}{.5\textwidth}
  \centering
  \includegraphics[scale =0.8]{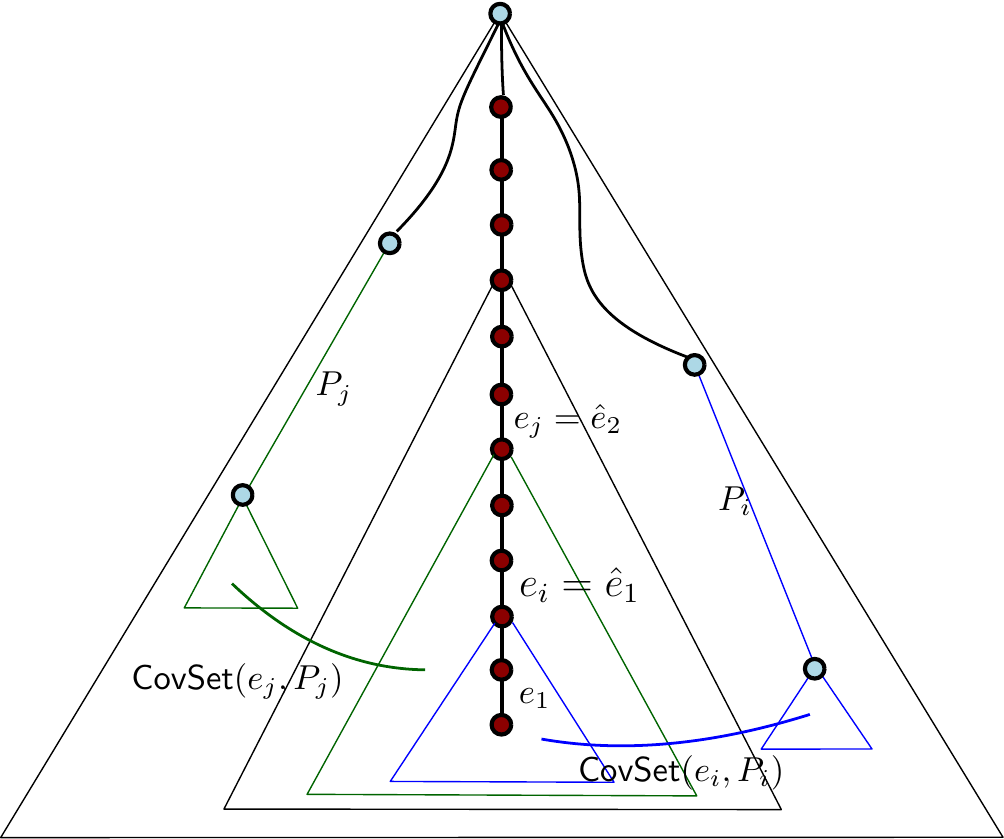}
\end{minipage}%
\begin{minipage}{.5\textwidth}
  \centering
  \includegraphics[scale =0.8]{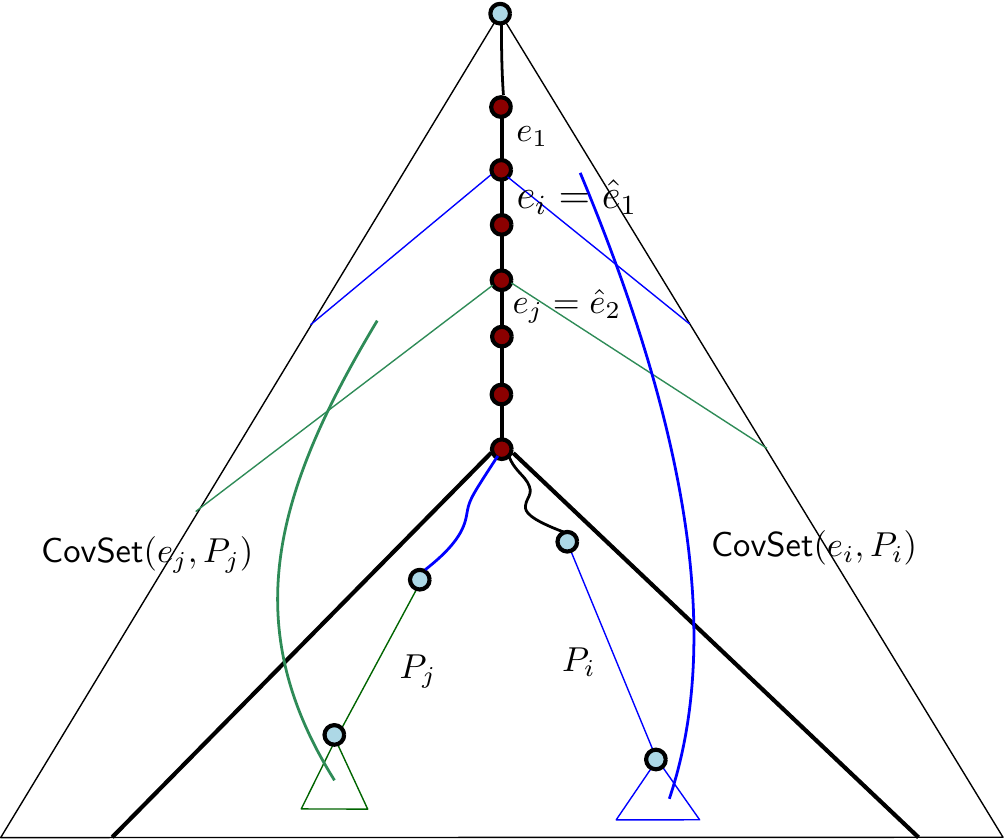}
\end{minipage}
\caption{\small Bounding the number of potentially interesting paths when $\cP$ is orthogonal to $P$: $P$ is the central path with red vertices. Edge $e_i$ is interested in path $P_i$ and edge $e_j$ is interested in path $P_j$. Left: $\cP$ is orthogonal to $P$. Right: $\cP$ is descendant of $P$.}
\label{fig:bounding-interesting}
\end{figure}

Having fixed $e_i$, we now go up in the path $P$, and focus on edges that are potentially interested in different paths in $\cP$. 
Let $e_{j_1}$ be the closest edge to $e_i$ in $P$ such that $e_{j_1}$ is potentially interested in a path $P_{j_1}\in \cP$ that no prior edge in the ordering is potentially interested in. This gives us, by Claim \ref{claim:samp_correct}, that w.h.p, $\cov(e_{j_1},P_{j_1})>\frac{\cov(e_{j_1})}{6}$. We make the following two observations at this point.

\begin{observation}
$\ECov(e_i,P_i)\subseteq\ECov(e_{j_1},\cP)$.
\end{observation}
This follows from the following reasoning: Since $e_{j_1}$ is an ancestor of $e_i$, and by non-splittability, we deduce that $\ECov(e_i,\cP)\subseteq\ECov(e_{j_1},\cP)$, which in turn gives us that $\cov(e_i,\cP)\leq \cov(e_{j_1},\cP)$, and in particular gives us that $\ECov(e_i,P_i)\subseteq\ECov(e_{j_1},\cP)$. 

\begin{observation}
$\ECov(e_i,P_i)$ and $\ECov(e_{j_1},P_{j_1})$ are disjoint.
\end{observation}
The reason is as follows. We notice that since $P_i,P_{j_1}$ are orthogonal, and since $P$ is either  completely above these paths or orthogonal to them, the edges of $P$ can not appear on the unique tree paths between pairs of edges from $P_i,P_{j_1}$ (See Observation \ref{obs:interested-edges}). This allows us to deduce that $\ECov(e_i,P_i,),\ECov(e_{j_1},P_{j_1})$ are disjoint. 

Now, from these two observations, we obtain that the following holds w.h.p.
\begin{gather*}
    \cov(e_{j_1},\cP)\geq \cov(e_{j_1},P_{j_1})+\cov(e_i,P_i)\geq \frac{\cov(e_{j_1})+\cov(e_i)}{6} \\ \geq \frac{\cov(e_{j_1},\cP)+\cov(e_i\cP)}{6} \geq \frac{2\cov(e_i,\cP)}{6}.
\end{gather*}


Had we been interested in paths in $\cP$ that $P$ is simply interested in, we would have got $\cov(e_{j_1},\cP)\geq \cov(e_i,\cP)$. This would be enough to complete the argument as mentioned in the overview of the proof. Unfortunately, we cannot have such strong claim when we are dealing with \textit{potentially interested} paths. Hence we have to consider 7 such edges $e_{j_1},...,e_{j_7}$. Formally, let us consider $e_{j_1},...,e_{j_7}$ with $j_7>...>j_1 > i$ such that $e_{j_1},...,e_{j_7}$ are the seven closest edges (by order) to $e_i$ in $P$ such that $e_{j_k},k\in \set{1, \cdots, 7}$ are potentially interested in a different path $P_{j_k}\in \cP$ from all previous paths considered (including all of $\cP_i$). If such edges don't exist, the lemma follows immediately. Otherwise, we make the following claim.

\begin{claim} \label{clm: cov-potint}
$\cov(e_{j_7},\cP) \geq \frac{7}{6} \cov(e_i,\cP)$ holds w.h.p.
\end{claim}

Let us first assume Claim \ref{clm: cov-potint}, and prove the lemma. The total weight of edges that cover $e_{j_7}$ and $\cP$ is larger by a constant factor (more than 1) from the total weight of edges that cover $e_i$ and $\cP$. We next repeat this argument as we traverse $P$ from $e_{j_7}$ upwards. Note that the total weight of edges cannot grow by such a multiplicative factor more than $O(\log n)$ times, as weights are polynomially bounded. Hence we get a set of $O(\log n)$ many edges  in $P$ which are potentially interested in new paths in $\cP$ than what their descendants in $P$ are interested in. Hence the lemma follows.

The case where we consider $\cP$ to be a set of orthogonal paths appearing completely below $P$ follows the same proof with a change in the ordering of the edge of $P$ (See Figure \ref{fig:bounding-interesting}), and replacing `ancestor' with `descendant'. We want to order the edges in the opposite order such that the lowest edge of $P$ gets the smallest index, and we now take $e_i$ to be the \emph{highest} edge that is potentially interested in some path $P_i\in \cP$, and $e_{j_1},...,e_{j_7}$ satisfy  $j_7<j_6<...<j_1<i$. This is because of the way the set $\ECov(e_i,P^\ast), P \in \cP$ is defined for this case: Any edge $f \in \ECov(e_i,P^\ast)$ has one end point as a descendant of $P^\ast$ as before, but has the other end point \textit{outside the set of descendants of $e_i$}, i.e., the unique tree path between the two end-point of $f$ includes both $e_i$ and $P^\ast$. Once we are set with this change in ordering, rest of the proof is similar to the previous case and is omitted.
\end{proof}

Now we prove Claim \ref{clm: cov-potint}.
\begin{proof}[Proof of Claim \ref{clm: cov-potint}]

By definition of potentially interested and Property \ref{item:samp-2} of Lemma \ref{lemma:find-int-samp}, that w.h.p., $$\cov(e_{j_k}, P_{j_k}) > \cov(e_{j_k}) /6.$$ 
As $e_{j_1},...,e_{j_7}$  are ancestors of $e_i$ and by non-splittability, any edge $e \in\ECov(e_i, P')$ such that $P'\in \cP$ also covers $e_{j_k}$ for all $k\in \set{1, \cdots, 7}$ i.e., $\ECov(e_i, \cP) \subset \ECov(e_{j_k}, \cP)$. This immediately gives the following observation:

\begin{observation}
$\cov(e_i, \cP) \leq \cov(e_{j_k}, \cP)$ for all $k \in \set{1, \cdots, 7}$.
\end{observation}

Note that $\cov(e_{j_k}, \cP)\leq \cov(e_{j_k})$ for all $k\in \set{1, \cdots, 7}$. Now one can deduce the following observation from the fact that $P_i$ and $P_{j_k}, k\in \set{1, \cdots, 7}$ are pairwise orthogonal, and the fact that all edges of $P$ are not on the unique paths in the tree between all pairs among $P_i$ and $P_{j_k}, k\in \set{1, \cdots, 7}$ (See Observation \ref{obs:interested-edges}).

\begin{observation}
$\ECov(e_i,P_i),\ECov(e_{j_k},P_{j_k}),k\in \set{1, \cdots, 7}$ are pairwise disjoint.
\end{observation}
Due to non-splittability, we also know that $\ECov(e_{j_k},P_{j_k})\subseteq \ECov(e_{j_7},\cP).$
Combining these observations 
we can deduce that w.h.p.
\[
\cov(e_{j_7},\cP)\geq \sum\limits_{k=1}^7 \cov(e_{j_k},P_{j_k})\geq \sum\limits_{k=1}^7 \frac{\cov(e_{j_k})}{6} \geq \sum\limits_{k=1}^7 \frac{\cov(e_{j_k},\cP)}{6} \geq \frac{7}{6} \cdot \cov(e_i,\cP). 
\]

\end{proof}

\subsection{Learning the interesting paths}\label{ssec:path_pruning}
In this section we explain how given the sampled edges by each edge $e$ in the sampling procedure, $e$ can internally construct the set $\intpot{e}$. Furthermore, we define the type of paths we work with in the algorithm, and we prove that for each such path $P$, all of its edges can efficiently learn the set $\intpot{P}$. We start with the definition of knowing a path. 

\begin{definition}[Knowing a path]\label{def:edge_knows_path}
Given a tree edge $e$ and some ancestor to descendant path $P'$ which is either a highway path or a non-highway path, we say that $e$ knows $P$ if the following holds.
\begin{enumerate}
    \item $e$ knows whether $P'$ contains non-highway edges.
    \item $e$ knows the lowest fragment $F$ that intersects $P$ (and immediately from  using LCA computation, the highest fragment as well).
\end{enumerate}
\end{definition}

Note, that since the id of each  fragment is $O(\log n)$ bits, the information about each path is $O(\log n)$ bits as well, since the first condition requires a single bit.
Furthermore, from here on in we identify each ancestor to descendant path with the lowest fragment $F$ it intersects. We interchangeably refer to paths using either the standard notion, or using the lowest fragment of the path.

We now define the types of paths that we work with in the algorithm.
 
\begin{definition}\label{def:fragment_and_super_highways}
Given a highway path $P$, we call $P$ a fragment highway if $P=P_F$ for some fragment $F$. Here, $P_F$ is the highway path of the fragment $F$. We call a highway path $P$ a super highway if $P$ is the union of two or more fragments highways. We call $P$ a highway bough if $P$ is a super highway, and $P$ is maximal with respect to the layering of the skeleton tree (See Lemma \ref{lemma:layer_decomposition_highway}).
\end{definition}

Now we address the issue of translating the sampled edges into knowing the paths that each path is potentially interested. Specifically, the missing piece is how, given a tree edge $e$, one can make $e$ know each path in the set $\intpot{e}$.

So far, we explained how $e$ uses Lemma \ref{lemma:routing_lemma_aggregate} in order to learn $\info(e^*)$ for all edges $e^*\in \ECovs(e)$. We now explain why this information suffices for $e$ to construct internally without further communication the set $\intpot{e}$. 
We aim to prove the following lemma. Recall that for a given tree edge $e$ we have that
\[
\intpot{e}=\left\{P \mid \covs(e,P)\geq \frac{\covs(e)}{3}\right\}\cup P_e
\]
Here, $P_e$ is the path from the root to $e$. 
We next show that all non-highways and highways can learn about the paths they are potentially interested in, the proof is deferred to Appendix \ref{sec:appen_five}.

\begin{restatable}{lemma}{PathParsing}\label{lemma:parse_paths}

All non-highway boughs, and all fragment highways can learn the paths (See definition \ref{def:edge_knows_path}) each of them is potentially interested in $\Tilde{O}(\dfrag)$ rounds. Furthermore, in $\Tilde{O}(\dfrag)$ rounds, all non-highway boughs $P$ can simultaneously send $\intpot{P}$ to all vertices in $T(P)$, and all fragment highways $P$( of a fragment $F_P$) can send $\intpot{P}$ to all vertices in $F_P$.
\end{restatable}

Note that the above lemma is far from trivial, in particular, there is no clear way to translate the set $\ECovs(e)$ for some given tree edge $e$ into the set of potentially interesting paths. A naive attempt might be to just take as our set of potentially interesting paths all fragments (which represent paths)  that contain some vertex incident to an edge in $\ECovs(e)$. Note however that this is problematic since this might put paths $P$ in $\intpot{e}$ whose weight fraction $\frac{\cov(e,P)}{\cov(e)}$ is very small (See Figure \ref{fig:bounding-interesting}), since as far as we know, the only edge in $\ECov(e,P)$ is the sampled edge that included $P$ in $\intpot{e}$, and this edge can have very low weight. This is a problem since this means that we are not potentially interested in $P$ and thus we would not be able to apply the interesting path counting lemma (See Section \ref{ssec:-interesting_lemma}) later in the algorithm, which is crucial for the algorithm's fast running time. 

Another naive attempt might be to just take as our set of potentially interesting paths all fragments (which represent paths)  that contain some vertex incident to an edge in $\ECovs(e)$ \emph{and} $e$ is potentially interested in the paths that the fragments represent. Note however that this set might be empty. If one denotes by $P$ a path that $e$ is potentially interested in, it might be that all edges in $\ECovs(e)$ that have an incident vertex in $T(P)$ sot not intersect $P$, but are actually connected to some low vertices in $T(P)$, which are in a different fragment than the one that represents $P$. An example of this is depicted in Figure \ref{fig:LCA_figure} In this case, how can we find the fragment that represents $P$ given only the edges $\ECovs(e)$?

\paragraph{LCA claim.} To overcome this obstacle we prove a nice structural claim regarding the connection between the edges in $\ECovs(e)$ and the lowest vertices of the paths $P$ that $e$ is potentially interested in. Intuitively, the claim says that for each path in $\intpot{P}$, one can identify the lowest vertex $v$ by considering the LCAs of \emph{pairs} of vertices which are incident to edges in $\ECovs(e)$. This claim allows us to deduce from $\ECovs(e)$ the specific fragments that represent all paths that $e$ is potentially interested in.

\begin{figure}[h!]
    \centering
    \includegraphics[width=\textwidth, height=9cm]{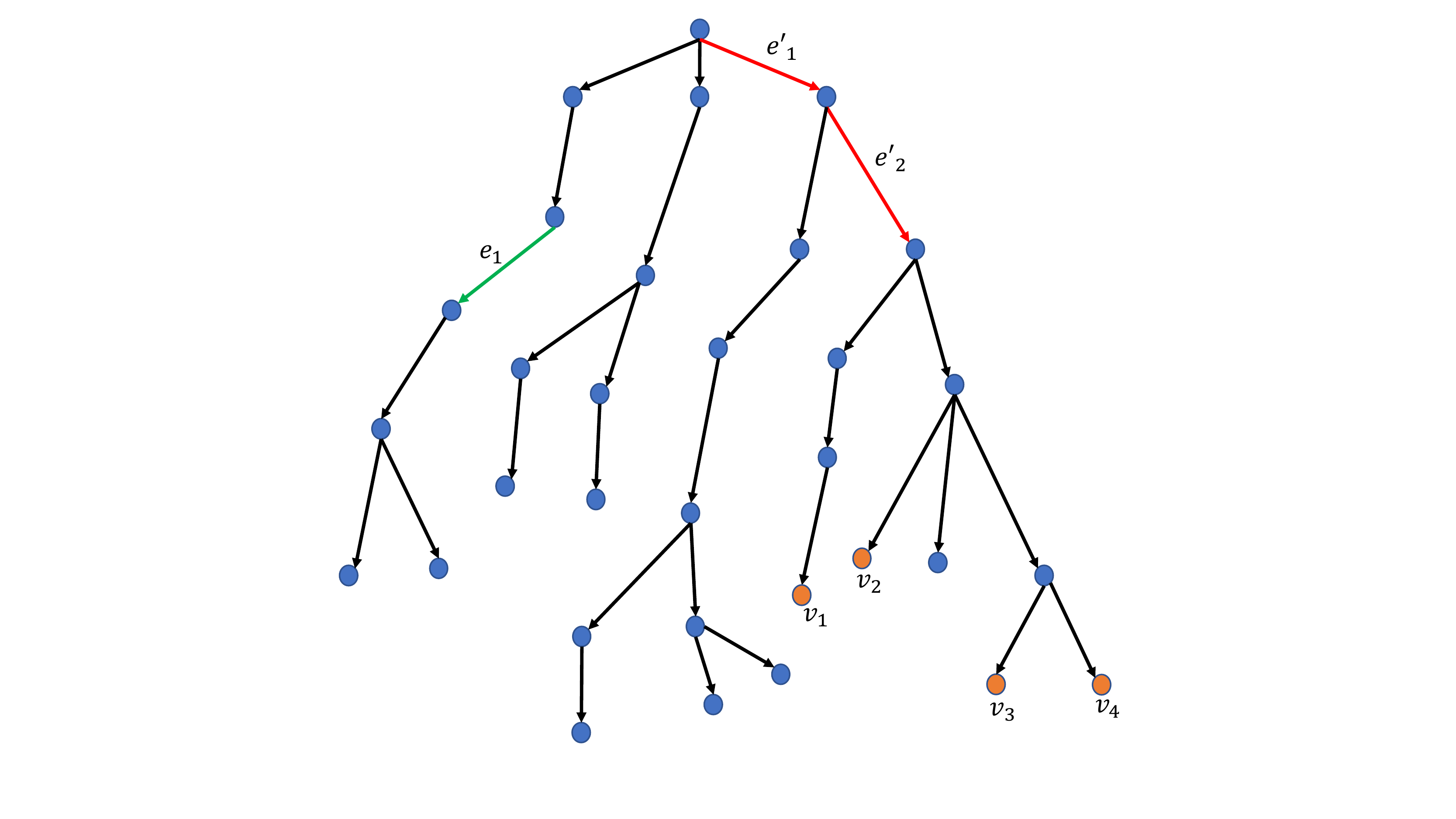}
    \caption{\small  In this figure we have the green edge $e_1$, and the path $P$ which consists of $e'_1,e'_2$, and satisfies $P\in \intpot{e}$. $v_1,v_2,v_3,v_4$ denote the vertices in $T(P)$ that have incident edges in $\ECovs(e_1)$. Note that there exists a pair of vertices (e.g. $v_1,v_3$) such that $\lca{v_1,v_3}$ is exactly the lowest vertex of $P$.}
    \label{fig:LCA_figure}
\end{figure}

\paragraph{Corollaries of interesting path counting lemma (Lemma \ref{lemma:-cover_all_or_nothing_paths}) and Lemma \ref{lemma:parse_paths}).}
Lastly, we present some  corollaries of  Lemma \ref{lemma:-cover_all_or_nothing_paths} and Lemma \ref{lemma:parse_paths} , which are used later in the algorithm.

The following corollary stems from the fact that any 2 non-highway paths in \emph{different fragments} are orthogonal.

\begin{corollary}\label{corol:-layering_interested}
For each $1\leq i\leq L$, where $L$ is the number of layers in the layering decomposition of the non-highways (See Section \ref{ssec:way_todo_layering}), the following holds: Consider a non-highway $P$ which is a bough path of layer $i$. Then $P$ is potentially interested in at most $B_{path}$ non-highway paths (represented by the ID of their respective fragments) which are not in the fragment of $P$. 

Furthermore, for each non-highway bough $P$ of layer $i$, all the vertices in $T(P)$ know (See definition \ref{def:edge_knows_path}) all of the $B_{path}$ non-highway paths (represented by the ID of their respective fragments) which are not in the fragment of $P$, that $P$ is potentially interested in. 
\end{corollary}

The following corollary stems from the interesting path counting lemma since the considered super highway paths set denoted by $\cP$ in the lemma are orthogonal. 

\begin{corollary} \label{corol:-_layer_highway}

Let $P$ be some non-highway path $P$, and let $\cP$ be some set of pairwise orthogonal super highway paths. Then $|\intpot{P}\cap \cP|=B_{path}$.

Furthermore, for each non-highway bough $P$ of layer $i$, all the vertices in $T(P)$ know (See definition \ref{def:edge_knows_path}) all of the $B_{path}$ super highway paths (represented by the ID of their lowest fragment) which are not in the fragment of $P$, that $P$ is potentially interested in. 
\end{corollary}

The following corollary stems from the fact that any 2 non-highway paths in \emph{different fragments} are orthogonal.
 
\begin{corollary}\label{corol:-highway_non_highway}
Let $P$ be be some highway path contained in a single fragment $F$. Denote by \ $\cF$ the set of fragments $F'$ that contain some non-highway path that is in  $\intpot{P}$. Then $|\cF_i|=B_{path}$.

Furthermore, for such highway path $P$, all the vertices in $F$ (the fragment of $P$) know (See definition \ref{def:edge_knows_path}) all of the $B_{path}$ fragments  that contain a non-highway that $P$ is potentially interested in.
\end{corollary}

\subsection{The highway pairing theorem}\label{ssec:bridge_the_gap}

We again begin with a corollary of Lemma \ref{lemma:-cover_all_or_nothing_paths}.

\begin{corollary}\label{corol:-_fragment_highway}
Let $P$ be some highway path contained a single fragment $F$. Denote by $\cP$ some set of pairwise orthogonal highway paths that $P$ is potentially interested in. Then $|\cP|=B_{path}$.
\end{corollary}

Using this corollary, we also aim to prove the following important theorem, which is essential to achieving an optimal running time for the algorithm. Before stating the theorem, we introduce some relevant notation and definitions.  Recall the definition of a bough highway path (Definition \ref{def:fragment_and_super_highways}) to be a highway path that constitutes a maximal path in the skeleton tree with respect to some layer $i$ in the skeleton tree.



\begin{observation}\label{obs:orthogonal_maximal_fragment_highways}
Given two bough highway paths $P_1,P_2$ of layer $i$, then $P_1$ and $P_2$ are orthogonal.
\end{observation}
\begin{proof}
Consider the highest fragments in $P_1,P_2$, denoted by $F_1,F_2$ respectively. Denote by $e_1,e_2$ the corresponding edges of $F_1,F_2$ in the skeleton tree, since both are of layer $i$, then $e_1,e_2$ are orthogonal. Thus $F_1,F_2$ are orthogonal, and this concludes the proof.
\end{proof}

Denote by $\cP$ the set of bough highway paths.
Each bough highway path can be spread across multiple fragments. For such a path $P$, we define by $\cF_P$ the set of corresponding fragments of $P$, i.e. $\cF_P=\set{F\in \cF\mid E(F)\cap P\neq \emptyset}$. For a fragment $F \in \cF_P$, we denote the subpath of $P$ contained in $F$ to be $P_F$. Given a pair $(P_0,P_1)$ of bough highway paths, for a fragment $F\in \cF_{P_i}$, $P_F$ is called \emph{active}, if $P_F$ is potentially interested in $P_{1-i}$. We also abuse notation and say that $P_F\subseteq (P_1,P_2)$ if $P_F\subseteq P_1$ or $P_F\subseteq P_2$.

\begin{figure}
    \centering
    \includegraphics[width=0.7\textwidth]{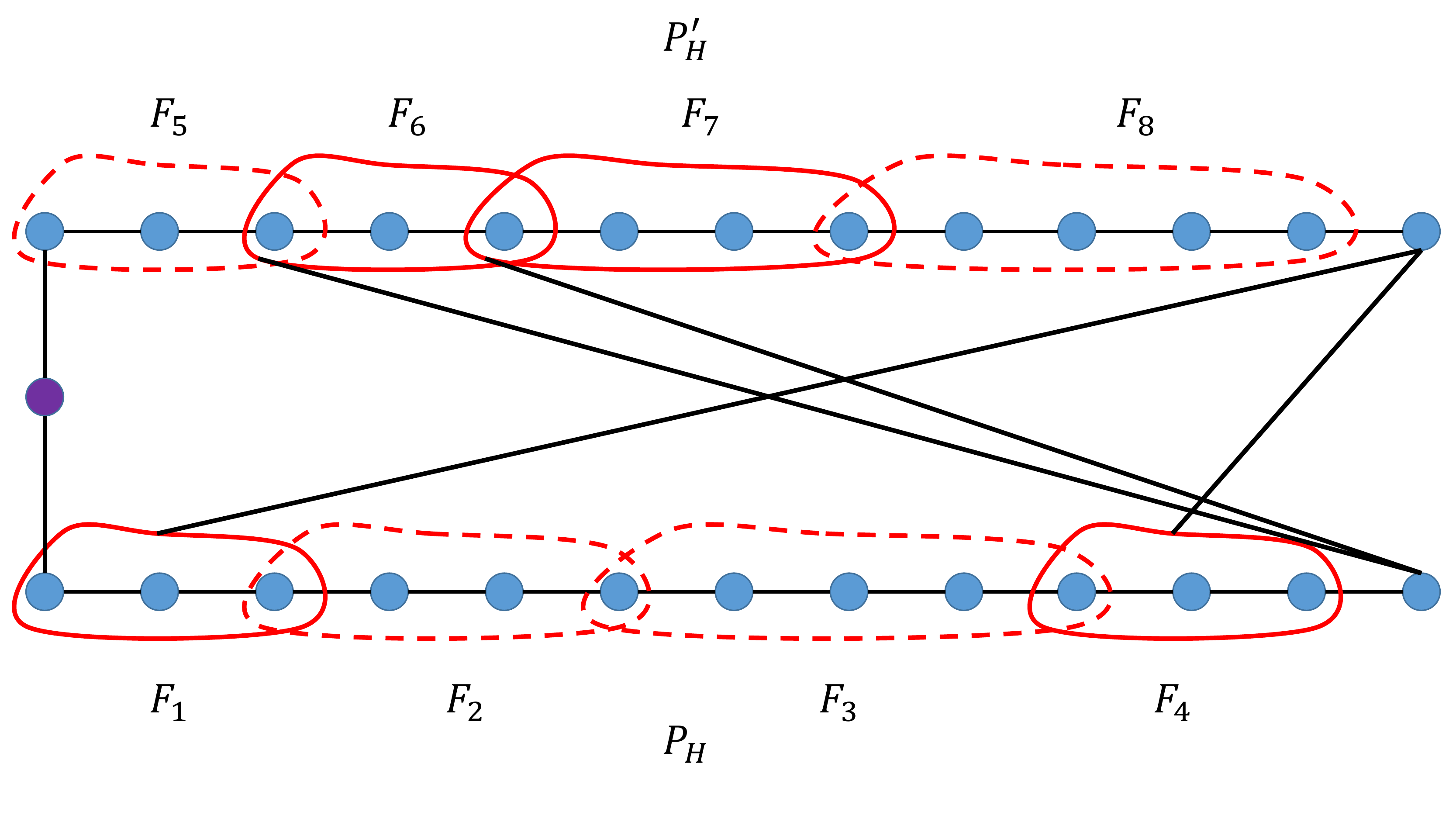}
    \caption{\small Example of two highways, with the purple vertex being the root of the tree. Fragments are circled in red. Dotted fragments are inactive with respect to this pair. Black lines represent a high amount of weight in terms of edges going from the respective fragment to the respective subtree of the vertex on the right.}
    \label{fig:highway_active_fragments_examples}
\end{figure}




At the end of the previous section, we presented corollaries that show that for non-highway paths, their entire subtree can learn the paths that the paths are potentially interested in. This was a valid option due to the small size of each fragment, and since for each non-highway path $P$, $T(P)$ is completely contained in the fragment of $P$.


Highways however, need a different treatment, more specifically our goal is to make sure that we don't compare too many pairs of fragment highways to one another, and also, we want to make the pairs of fragment highways we need to compare during the algorithm global information.
A naive attempt for achieving that might be making the pairs of fragment highways that are potentially interested in one another global information. This, however, can cause a lot of congestion since the number of such pairs can be linear in $n$ since each fragment highways can be potentially interested in $\Omega(\sqrt{n})$ other fragment highways. To circumvent this, we employ the interesting path counting lemma (Lemma \ref{lemma:-cover_all_or_nothing_paths}) as depicted in the corollary above, that each fragment highway $P$, can only be potentially interested in at most $O(\log n)$ fragment highways that are pairwise orthogonal. Using this corollary, instead of pairing up fragments, we employ the layer decomposition of the skeleton tree (See Section \ref{ssec:-layerDecomp}) in order to pair up \emph{highway boughs} that are potentially interested in one another. The number of such pairs $(P_1,P_2)$ is sufficiently small to make these pairs global information, including the information on which fragments in $P_1,P_2$ are active with respect to this pair. 
Formally, we prove the following theorem.
\begin{restatable}{theorem}{HighwayPairingTheorem}\label{thrm:pairing_up_highways}
Given the set of bough highways paths $\cP$ as defined in Definition \ref{def:fragment_and_super_highways}, one can construct a set $R\subseteq \cP \times \cP$ 
such that the following holds.
\begin{enumerate}
    \item If $(P_1,P_2)\in R$, then $P_1$ is potentially interested in some highway sub-path of $P_2$, and vice versa.
    \item For any fragment $F$, it holds that $R_F=|\set{(P_1,P_2)\in R \mid P_F\subseteq (P_1, P_2); P_F \text{ is active} }|=B_{path}\cdot \log n$. 
    \item  If $e_1\in P,e_2\in P'$ are highway edges on different bough highway paths $P,P'$ that define the min 2-respecting cut of $T$, then w.h.p. $(P,P')\in R$. Furthermore, $e_1$ is in the active fragments of $P$, and $e_2$ is in the active fragments of $P'$.
    \item All the vertices in the graph $G$ can learn the set $R$ in time $\Tilde{O}(D+\sqrt{n})$. Furthermore, for each pair $(P_1,P_2)\in R$, all the vertices of $G$ know who are the active fragments in the pair $(P_1,P_2)$.
    \item For all $(P_1,P_2)\in R$, it holds that $P_1$ and $P_2$ don't split one another (See Section \ref{ssec:-interesting_lemma}).
\end{enumerate}
\end{restatable}

The proof of the theorem is deferred to Appendix \ref{sec:appen_five}. Here we give a short intuition to its correctness.

In short, the proof goes as follows. We begin with the set of pairs of bough highwways $(P_1,P_2)$ such that $P_1$ is potentially interested in $P_2$ and vice versa. Already this pairing satisfies properties 1-4. Lemma \ref{lemma:-cover_all_or_nothing_paths} and its corollary stated at the beginning of the section are crucial for proving the properties. Then, to make sure property 5 is satisfied, and careful trimming procedure is done which does not hurt all of the other properties.

\section{Algorithms for short paths and routing trick}\label{sec:alg_short_path}
 During our algorithm, many times we consider two ancestor to a descendant sub-paths $P',P$ of length $O(\dfrag)$, and find the min 2-respecting cut that has one edge in $P'$, and one edge in $P$. In this section, we describe how to compare two such sub-paths. The basic idea is to use an edge between the paths for routing information between them, and then use internal aggregate computations. In the case there is no edge, we use a \emph{routing trick} to bound the amount of global communication.
We divide to cases according to whether $P'$ and $P$ are in a highway or not. In all our claims we assume that each edge $e$ knows the value $\cov(e)$, which can be obtained from Claim \ref{claim_learn_cov}.
Before explaining the algorithm, we start with some simple claims that are useful later.

 \subsection{Preliminaries: Basic subroutines on a tree} \label{sec:preliminaries_paths}

During our algorithm many times we run basic computations in trees, mostly in the trees $T(P)$ of non-highways, on the trees $F_P$ defined by fragments, and on a BFS tree of the graph. We next discuss such computations and explain how to run many such computations in a pipelined manner.

\paragraph{Broadcast.} In a broadcast computation, the root of the fragment has a message of size $O(\log n)$ to pass to the whole tree. 
This requires time proportional to the diameter of the tree, as the root starts by sending the message to its children, that in the next round pass it to their children, and so on. Note that the computation requires sending only one message on each one of the edges of the tree.

\paragraph{Aggregate computation in a tree.} 
In an aggregate computation in a tree $T'$ (sometimes called convergecast), we have an associative and commutative function $f$ with inputs and output of size $O(\log n)$ (for example, $f$ can be sum or minimum).
Each vertex $v \in T'$ initially has some value $x_v$ of $O(\log n)$ bits, and the goal of each vertex is to learn the output of $f$ on the inputs $\set{x_v}_{v\in T'_{v}}$, where $T'_{v}$ is the subtree rooted at $v$. 

Computing an aggregate function is done easily, as follows. We start in the leaves, each one of them $v$ sends to its parent its value $x_v$, each internal vertex applies $f$ on its input and the inputs it receives from its children, and passes the result to its parent. The time complexity is proportional to the diameter of the tree. Moreover, the algorithm requires sending only one message on each one of the tree edges, which allows to pipeline such computations easily. For a proof, see Appendix \ref{sec:app_short_paths_prelim}.

\begin{restatable}{claim}{pipeline}\label{claim_pipeline}
Let $T'$ be a tree of diameter $D_{T'}$, and assume we want to compute $c_1$ broadcast computations and $c_2$ aggregate computations in $T'$. Then, we can do all computations in $O(D_{T'} + c_1 + c_2)$ time. Moreover, we can work in parallel in trees that are edge-disjoint.
\end{restatable}

\paragraph{Tree computations in a fragment $F_P$.}
When we work on fragments $F_P$ during the algorithm, we use aggregate computations in two different directions. Recall that each fragment $F_P$, has a highway $P$, which is a path between the root $r_P$ and descendant $d_P$ of the fragment, and additional sub-trees attached to $P$ that are contained entirely in the fragment. We will need to do \emph{standard} aggregate computations as described above, where $r_P$ is the root, but also computation \emph{in the reverse direction}, where we think about $d_P$ as the root and orient all edges in the fragment accordingly. See Figure \ref{orientation_pic} for an illustration. 

\setlength{\intextsep}{2pt}
\begin{figure}[h]
\centering
\setlength{\abovecaptionskip}{-2pt}
\setlength{\belowcaptionskip}{6pt}
\includegraphics[scale=0.5]{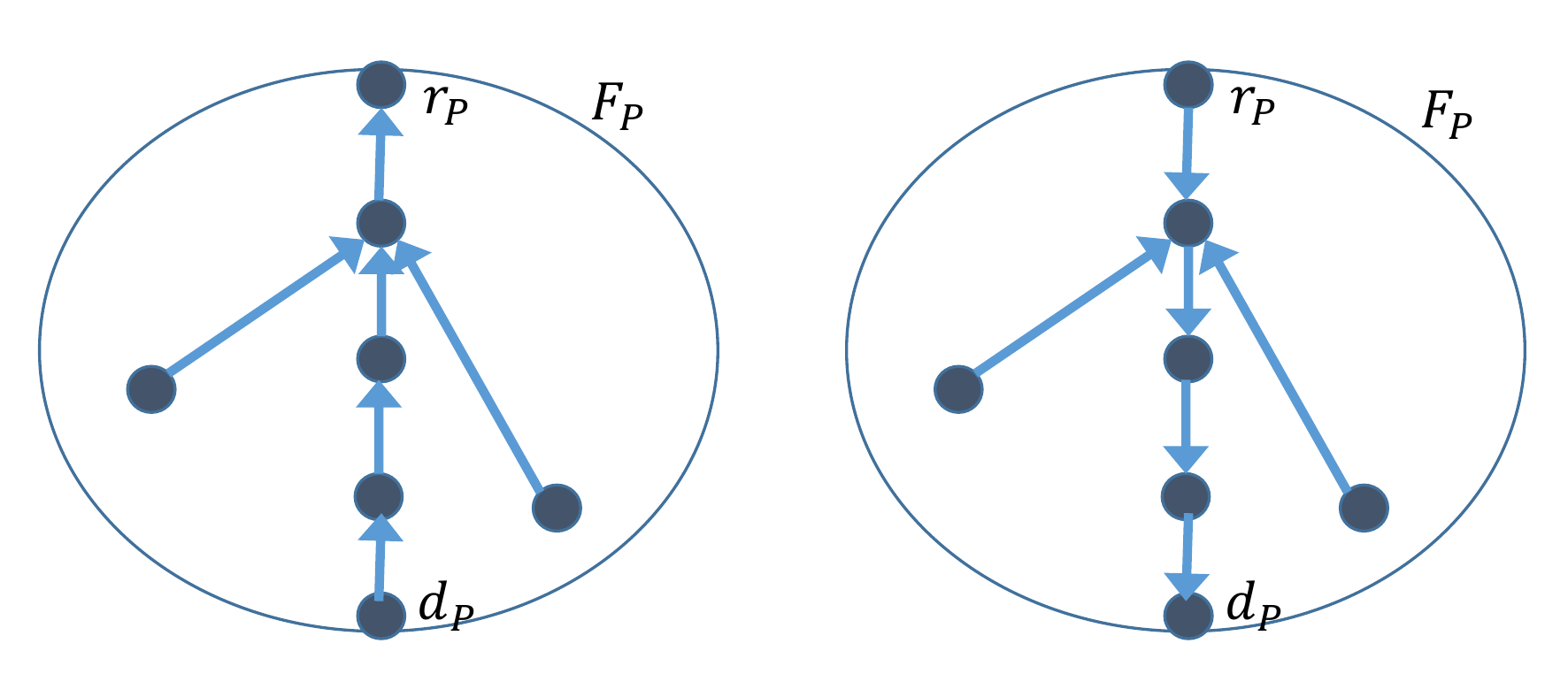}
 \caption{\small Possible orientations of the fragment $F_P$. On the left appears a standard orientation where $r_P$ is the root, and on the right appears an orientation where $d_P$ is the root. Note that the only edges that changed their orientation are highway edges.}
\label{orientation_pic}
\end{figure}

We next show that we can pipeline aggregate computations in both directions and broadcast computations efficiently. For a proof, see Appendix \ref{sec:app_short_paths_prelim}.

\begin{restatable}{claim}{pipelineFragment} \label{claim_pipeline_fragment}
Let $F_P$ be a fragment with its highway denoted by $P$, and assume we want to compute $c_1$ broadcast computations, $c_2$ aggregate computations, and $c_3$ aggregate computations in the reverse direction in $F_P$. Then, we can do all computations in $O(\dfrag + c_1 + c_2 + c_3)$ time. Moreover, we can work in parallel in different fragments.
\end{restatable}

We next discuss additional notation and claims required.
We denote by $p(v)$ the parent of $v$.
For the algorithm, we first make sure that the highest and lowest edges in each highway are known to all vertices, which takes $O(D+\nfrag)$ time.

\begin{claim} \label{claim_high_low}
In $O(D+\nfrag)$ time all vertices can learn the highest and lowest edges in each highway.
\end{claim}

\begin{proof}
This requires collecting and sending $O(1)$ information per fragment, which can be done in $O(D+\nfrag)$ time using upcast and broadcast in a BFS tree.
\end{proof}

\subsection{Simple cases where $P'$ is a non-highway} \label{sec:non_highway_simple}

We start with the following basic claim that is useful for cases involving a non-highway.
We later use it to compare a non-highway edge to other edges in its fragment, as well as to edges in another fragment, assuming there is an edge between the fragments.

\remove{
In this section we discuss the following cases:
\begin{enumerate}
\item $P'$ and $P$ are non-highways.
\item $P'$ is a non-highway and $P$ is the fragment highway in the fragment of $P'$.
\item The two edges defining the cut are in the same non-highway path $P'$.
\end{enumerate}

We start with the following claim that is useful for cases involving a non-highway.

}

\begin{claim} \label{claim_non_highway}
Let $P'$ be a non-highway path, and let $e$ be an edge outside $T(P')$, such that $\{e,\cov(e)\}$ is known to all vertices in $T(P')$. 
Then, using one aggregate computation in $T(P')$, each edge $e' \in P'$ can compute the value $\cut(e',e).$ In addition, for different paths $P'$ that are not in the same root to leaf path, these computations can be done in parallel.
\end{claim}

\begin{proof}
Let $e' \in P'$. By Claim \ref{claim_cover}, $\cut(e',e)=\cov(e')+\cov(e)-2\cov(e',e)$. Also, $e'$ knows $\cov(e)$ and $\cov(e')$, hence to compute $\cut(e',e)$, the edge $e'$ should compute $\cov(e',e).$ As any edge that covers a tree edge $e'=\{v',p(v')\}$ where $p(v')$ is the parent of $v'$, has one endpoint in the subtree of $v'$ by Claim \ref{claim_subtree}, it follows that the edges that cover $e'$ and $e$ are exactly all the edges that cover $e$ and have exactly one endpoint in the subtree of $v'$. This subtree is contained in $T(P')$. Note that any edge $x$ that covers $e'$ and $e$ cannot have both endpoints in $T(P')$, as otherwise the tree path covered by $x$ is contained entirely in $T(P')$, but as $e$ is outside $T(P')$, $x$ cannot cover $e$ in this case, which leads to a contradiction.
Hence, to learn about the total cost of the edges that cover $e$ and $e'$ we need to do one aggregate computation in $T(P')$: letting each vertex $v' \in T(P'^{\downarrow})$ learn about the total cost of edges incident on the subtree of $v'$ that cover $e$. This can be done using a convergecast in $T(P')$, where each vertex sends to its parent the total cost of edges in its subtree that cover $e$, this includes the sum of costs of edges it receives from its children and the cost of such edges adjacent to it. To implement this efficiently we need an efficient way to determine for each non-tree edge whether it covers a specific edge $e$, this can be done using Claim \ref{claim_LCA_labels}. Hence, all the edges $e' \in P'$ can compute $\cut(e',e)$ using one aggregate computation in the subtree of $P'$.
Since the whole computation was done in $T(P')$, the same computation can be done in parallel for other paths not in the same root to leaf path with $P'$.
\end{proof}

We can use Claim \ref{claim_non_highway} to compare a non-highway path $P'$ in the fragment $F_{P'}$ to the edges of a fragment $F \neq F_{P'}$, assuming there is an edge between $T(P'^{\downarrow})$ and $F$. The time is proportional to the size of the fragments $\sfrag=O(\sqrt{n})$. For a proof, see Appendix \ref{sec:app_short_paths_nh},

\begin{restatable}{claim}{CompareFragments}\label{claim_orthogonal_nh}
Let $P'$ be a non-highway in the fragment $F_{P'}$ and let $F \neq F_{P'}$ be another fragment. Assume that there is an edge $f$ between $T(P'^{\downarrow})$ and $F$, that is known to all vertices in $T(P')$.
Also, assume that at the beginning of the computation all vertices in $F$ know all values $\{e,\cov(e)\}_{e \in F}$. Then, in $O(\sfrag)$ time, all edges $e' \in P'$ can compute the values $\{e,\cut(e',e)\}_{e \in F}.$ The computation can be done in parallel for different paths $P'$ not in the same root to leaf path.
\end{restatable}

\remove{
\begin{claim} \label{claim_orthogonal_nh}
Let $P'$ be a non-highway in the fragment $F_{P'}$ and let $F \neq F_{P'}$ be another fragment. Assume that there is an edge $f$ between $T(P'^{\downarrow})$ and $F$, that is known to all vertices in $T(P')$.
Also, assume that at the beginning of the computation all vertices in $F$ know all values $\{e,\cov(e)\}_{e \in F}$. Then, in $O(\sfrag)$ time, all edges $e' \in P'$ can compute the values $\{e,\cut(e',e)\}_{e \in F}.$ The computation can be done in parallel for different paths $P'$ not in the same root to leaf path.
\end{claim}

\begin{proof}
As all vertices in $F$ know the values $\{e,\cov(e)\}_{e \in F}$, and the edge $f$ has an endpoint in $F$ and an endpoint in $T(P')$, it knows this information, and can pass it to all vertices in $T(P')$ using $O(\sfrag)$ aggregate and broadcast computations in $T(P').$ Note that since the subtrees $T(P'^{\downarrow})$ are disjoint for non-highways not in the same root to leaf path, the edge $f$ is different for different such paths $P'$, which allows working in parallel as needed. As $F$ and $P'$ are in different fragments, and $P'$ is a non-highway, we have that all edges $e \in F$ are not in $T(P')$, hence we can use Claim \ref{claim_non_highway} to let all edges $e' \in P'$ learn the values $\{e,\cut(e',e)\}$, this takes $O(\sfrag)$ time using pipelining, and can be done in parallel for different non-highways $P'$ not in the same root to leaf path. 
\end{proof}
}

Similar ideas allow us to compare the edges of a non-highway $P'$ to all edges that are above or orthogonal to them in the same fragment. This in particular allows to compare a non-highway to all edges of the highway of the same fragment. For a proof, see Appendix \ref{sec:app_short_paths_nh}.

\begin{restatable}{claim}{NhFragment}\label{claim_non_highway_fragment}
Let $P'$ be a non-highway in the fragment $F_{P'}$ and assume that all vertices in $T(P')$ know the complete structure of $P'$. In $O(\sfrag)$ time, all edges $e' \in P'$, can compute the values $\cut(e',e)$ for all edges $e \in F_{P'}$ that are either above or orthogonal to them in the fragment $F_{P'}$. The computation can be done in parallel for different paths $P'$ not in the same root to leaf path.
\end{restatable}

\subsection{$P'$ is a non-highway and $P$ is a highway}

Here we focus on the case that we compare a non-highway $P'$ to a fragment highway $P$ in a different fragment. The case that $P$ is in the same fragment was already discussed in the previous section.
Let $e' \in P', e \in P$, we next look at the value $\cov(e',e).$ See Figure \ref{non_highway_highway} for an illustration.
We need the following definitions. 
\begin{itemize}
\item $\cov_{F}(e',e)$ is the cost of all edges that cover $e'$ and $e$ and have one endpoint in $T(P'^\downarrow)$ and one endpoint in the fragment $F_P$ of $P$.
\item $\extcov(e',P)$ is the cost of all edges that cover $e'$ and the whole highway $P$, and have one endpoint in $T(P')$ and both endpoints outside $F_P$.
\end{itemize}

We show that all all the edges that cover $e'$ and $e$ are in one of the above categories. For a proof, see Appendix \ref{sec:app_short_paths_nhh}.

\setlength{\intextsep}{0pt}
\begin{figure}[h]
\centering
\setlength{\abovecaptionskip}{-2pt}
\setlength{\belowcaptionskip}{6pt}
\includegraphics[scale=0.4]{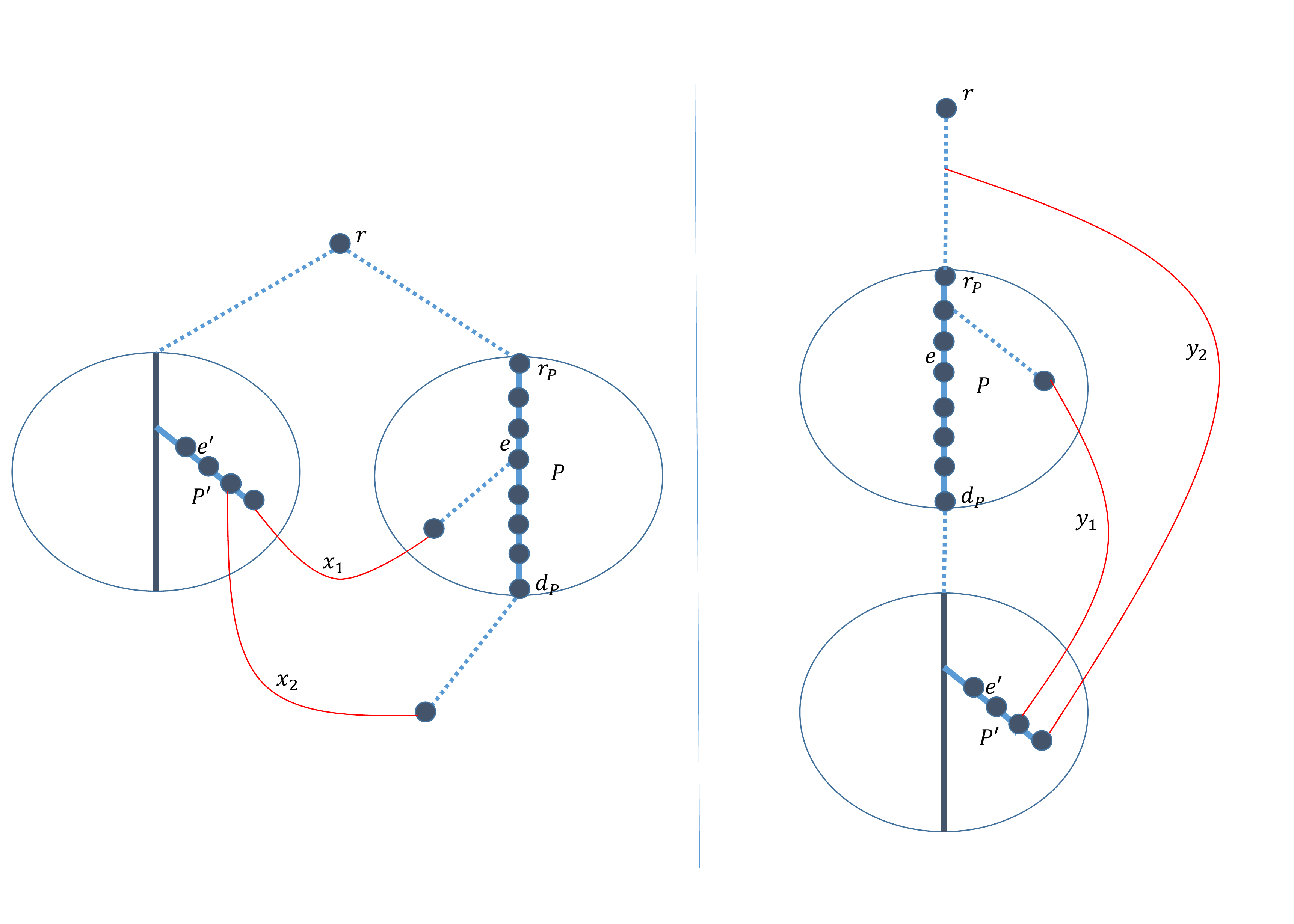}
 \caption{\small Examples of a non-highway $P'$ and a highway $P$. On the left, $P'$ and $P$ are orthogonal, and on the right, $P$ is above $P'$. The edges $x_1$ and $y_1$ are examples of edges counted in $\cov_F(e',e)$, and the edges $x_2$ and $y_2$ are examples of edges counted in $\extcov(e',P).$}
\label{non_highway_highway}
\end{figure}

\begin{restatable}{claim}{coverNhH}\label{cover_non_highway_highway_new}
Let $e' \in P', e \in P$ where $P'$ is a non-highway, and $P$ is a highway in a different fragment $F_P$. Then $\cov(e',e) = \cov_{F}(e',e) + \extcov(e',P)$.
\end{restatable}

In Appendix \ref{sec:app_short_paths_nhh} we explain how $\cov_{F}(e',e)$ and $\extcov(e',P)$ are computed. Intuitively, $\cov_{F}(e',e)$ can be computed by an aggregate computation similar to the one described in the proof of Claim \ref{claim_non_highway}, either inside $T(P')$ or in $F_P$. The value $\extcov(e',P)$ can be computed by an aggregate computation inside $T(P').$


\begin{restatable}{claim}{CovFnh} \label{claim_covF_nh}
Let $P'$ be a non-highway and let $P$ be a highway of the fragment $F_P$. 
Given an edge $e \in P$, using an aggregate computation in $T(P')$, all edges $e' \in P'$ can compute the value $\cov_{F}(e',e)$.
\end{restatable}

\begin{restatable}{claim}{CovFh} \label{claim_covF_h}
Let $P'$ be a non-highway and let $P$ be the fragment highway of the fragment $F_P$. 
Given an edge $e' \in P'$, using an aggregate computation in $F_P$, all edges $e \in P$ can compute the value $\cov_{F}(e',e)$.
\end{restatable}


\begin{restatable}{claim}{ClaimExtnh}\label{claim_cov_eP_nh}
In $O(\dfrag + \nfrag)$ time, all edges $e'$ in the non-highway $P'$ can compute the values $\extcov(e',P)$ for all highways $P$. Moreover, this computation can be done in all non-highways in the same layer simultaneously.
\end{restatable}

\subsubsection*{Comparing two paths using a routing trick}
We next explain how given a non-highway $P_0$ and a highway $P_1$ we compute the values $\{\cut(e',e)\}_{e' \in P_0, e \in P_1}.$ 
There are 2 cases, either there is an edge $f$ between $T(P_0)$ and $F_{P_1}$. In this case, we use $f$ to route information between $P_0$ and $P_1$, and compute the cut values using aggregate computations inside one of them. In the case there is no edge, we show a \emph{routing trick} to limit the amount of global communication. Basically we show that if we broadcast $O(\sqrt{n})$ cover values to the whole graph, it is enough to deal with all pairs $P_0,P_1$ from this case.\\[-7pt]

\textbf{Case 1: There are no edges between $T(P_0)$ and $F_{P_1}.$}
Let $P$ be a highway, we denote by $e^P_{min}$ the edge $e \in P$ such that $\cov(e)$ is minimal. First, we make sure that the values $\cov(e^P_{min})$ are known to all vertices.

\begin{claim} \label{claim_cov_min}
In $O(D+\nfrag + \dfrag)$ time all vertices learn the values $\{e^P_{min}, \cov(e^P_{min})\}$ for all highways $P$.
\end{claim}  

\begin{proof}
First, in each highway, finding the edge $e^P_{min}$ requires an aggregate computation inside the fragment, which takes $O(\dfrag)$ time.
Next we send this information to all vertices.
As there is only one edge $e^P_{min}$ per highway, and there are $\nfrag$ highways, this only requires collecting and broadcasting $\nfrag$ information in a BFS tree, which takes $O(D+\nfrag)$ time.
\end{proof}

\begin{lemma} \label{lemma_nh_no_edge}
Let $P_0$ be a non-highway, and $P_1$ be a highway in a different fragment, such that there are no edges between $T(P_0^\downarrow)$ and $F_{P_1}$. Then if each edge $e' \in P_0$ knows the values $\cov(e'),\extcov(e',P_1)$, as well as the values $\{e^{P_1}_{min}, \cov(e^{P_1}_{min})\}$, using one aggregate and broadcast computations in $P_0$, all vertices in $P_0$ learn about the values $\{e',e,\cut(e',e)\}$ for edges $e' \in P_0, e \in P_1$ such that $\cut(e',e)$ is minimal.
\end{lemma}

\begin{proof}
Let $e' \in P_0, e \in P_1$.
From Claims \ref{claim_cover} and \ref{cover_non_highway_highway_new}, we have $$\cut(e,e') = \cov(e) + \cov(e') - 2\cov(e,e'),$$ $$\cov(e',e) = \cov_{F}(e',e) + \extcov(e',P_1).$$
Now if there are no edges between $T(P_0^\downarrow)$ and $F_{P_1}$, from the definition $\cov_{F}(e',e) = 0$ for any pair of edges $e' \in P_0, e \in P_1$, as this is the sum of costs of edges with one endpoint in $T(P_0^\downarrow)$ and the second in $F_{P_1}$.
Hence, in this case, we have $$\cut(e,e') = \cov(e) + \cov(e') - 2\extcov(e',P_1).$$ Note that the expression $\cov(e') - 2\extcov(e',P_1)$ does not depend on the specific choice of $e \in P_1$, hence to minimize the expression $\cut(e',e)$ for a specific $e' \in P_0$, we just need to find an edge $e \in P_1$ where $\cov(e)$ is minimal, this is the edge $e^{P_1}_{min}$. As each edge $e' \in P_0$ knows the values $e^{P_1}_{min},\cov(e^{P_1}_{min}), \cov(e'),\extcov(e',P_1)$ it can compute $\cut(e',e)$ for an edge $e \in P_1$ such that $\cut(e',e)$ is minimal. To find the minimum over all choices of $e'$ we just need one aggregate computation in $P_0$ to find the minimum value computed, we can then broadcast the information to let all vertices in $P_0$ learn the values $\{e',e,\cut(e',e)\}$ for edges $e' \in P_0, e \in P_1$ such that $\cut(e',e)$ is minimal.
\end{proof}

\textbf{Case 2: There is an edge between $T(P_0)$ and $F_{P_1}.$}

\begin{lemma} \label{lemma_paths_nh_h}
Let $P_0$ be a non-highway in the fragment $F_{P_0}$ and let $P_1$ be the fragment highway of a different fragment $F_{P_1}$, and assume that there is an edge $f$ between $T(P_0^\downarrow)$ and $F_{P_1}$. 
Let $E_0 \subseteq P_0$ be a set of edges in $P_0$ we compare to $P_1$, and $E_1 \subseteq P_1$ be a set of edges in $P_1$ that we compare to $P_0$.
Additionally, assume that at the beginning of the computation the following information is known:
\begin{itemize}
\item All vertices in $T(P_0)$ know the identity of the edge $f$, The identity of all the edges in $E_0$, and for each edge $e' \in E_0$, the values $\cov(e'),\extcov(e',P_1)$.
\item All vertices in $F_{P_1}$ know the identity of the edge $f$, the identity of all the edges in $E_1$, and for each edge $e \in P_1$, the value $\cov(e)$.
\end{itemize}

We can compute the values $\{\cut(e',e)\}_{e' \in P_0, e \in P_1}$ in the following two ways.
\begin{enumerate}
\item In $O(|E_1|)$ aggregate and broadcast computations in $T(P_0)$, where at the end of the computation each edge $e' \in E_0$ would know the values $\cut(e',e)$ for all edges $e \in E_1$. \label{P0_compute}
\item In $O(|E_0|)$ aggregate and broadcast computations in $F_{P_1}$, where at the end of the computation each edge $e \in E_1$ would know the values $\cut(e',e)$ for all edges $e' \in E_0$. \label{P1_compute}
\end{enumerate}
\end{lemma}

\begin{proof}
We start by proving Case \ref{P0_compute} where the computations are done in $T(P_0)$.
First, we use the edge $f$ that has an endpoint in $T(P_0)$ and an endpoint in $F_{P_1}$ to pass information from $F_{P_1}$ to $T(P_0)$.
As $f$ has an endpoint in $F_{P_1}$ it knows the values $\{e,\cov(e)\}_{e \in E_1}$ and can pass them to all vertices in $T(P_0)$ using $O(|E_1|)$ aggregate and broadcast computations in $T(P_0)$. Let $e' \in E_0, e \in E_1.$
From Claims \ref{claim_cover} and \ref{cover_non_highway_highway_new}, we have $$\cut(e,e') = \cov(e) + \cov(e') - 2\cov(e,e'),$$ $$\cov(e',e) = \cov_{F}(e',e) + \extcov(e',P_1).$$ An edge $e' \in E_0$ already knows $\cov(e')$ and $\extcov(e',P_1),$ and from the broadcast it also knows the value $\cov(e)$ for all edges $e \in E_1.$ Hence, to compute the value $\cut(e',e)$ for a specific edge $e \in E_1$, it only needs to compute $\cov_{F}(e',e)$. By Claim \ref{claim_covF_nh}, all edges $e' \in P_0$ can compute the value $\cut(e',e)$ for a fixed $e$ using one aggregate computation in $T(P_0).$ To compute this value for all edges $e \in E_1$, we pipeline $|E_1|$ such aggregate computations, which concludes the proof of Case \ref{P0_compute}.   

We next discuss Case \ref{P1_compute}, where the computations are done in $F_{P_1}.$
Here we use $f$ to pass the information $\{e',\cov(e'),\extcov(e',P_1)\}_{e' \in E_0}$ from $T(P_0)$ to $F_{P_1}$, which requires $O(|E_0|)$ broadcast and aggregate computation in $F_{P_1}$. Now, after this each edge $e \in E_1$ knows the value $\cov(e)$ as well as the values $\{e',\cov(e'),\extcov(e',P_1)\}_{e' \in E_0}$. As discussed in the proof of Case \ref{P0_compute}, the only information missing to compute $\cut(e',e)$ is $\cov_{F}(e',e)$. Given an edge $e' \in E_0$, all edges $e \in P_1$ can compute $\cut(e',e)$ using one aggregate computation in $F_{P_1}$ by Claim \ref{claim_covF_h}. To do so for all edges $e' \in E_0$, we pipeline $|E_0|$ such aggregate computations, which completes the proof.
\end{proof}

\subsection{$P'$ and $P$ are highways}

Let $e' \in P', e \in P$ be two tree edges in the highways $P',P$ such that $F_{P'},F_P$ are the fragments of $P'$ and $P$, respectively. The value $\cov(e',e)$ is broken up to the following different parts. See Figure \ref{highway_highway} for an illustration.

\begin{enumerate}
\item The cost of edges that cover entirely the highways $P$ and $P'$, with endpoints outside $F_{P'} \cup F_P$: $\extcov(P',P).$ 
\item The cost of edges with one endpoint in $F_{P'}$ and one endpoint outside $F_{P'} \cup F_P$ that cover $e'$ and the whole highway $P$: $\extcov(e',P).$
\item The cost of edges with one endpoint in $F_{P}$ and one endpoint outside $F_{P'} \cup F_P$ that cover $e$ and the whole highway $P'$: $\extcov(e,P').$
\item The cost of edges that cover $e',e$ and have endpoints in both $F_P$ and $F_{P'}$: $\cov_{F}(e',e)$.
\end{enumerate}

To see that these are all the options we use the structure of the decomposition. 

\setlength{\intextsep}{0pt}
\begin{figure}[h]
\centering
\setlength{\abovecaptionskip}{-2pt}
\setlength{\belowcaptionskip}{6pt}
\includegraphics[scale=0.4]{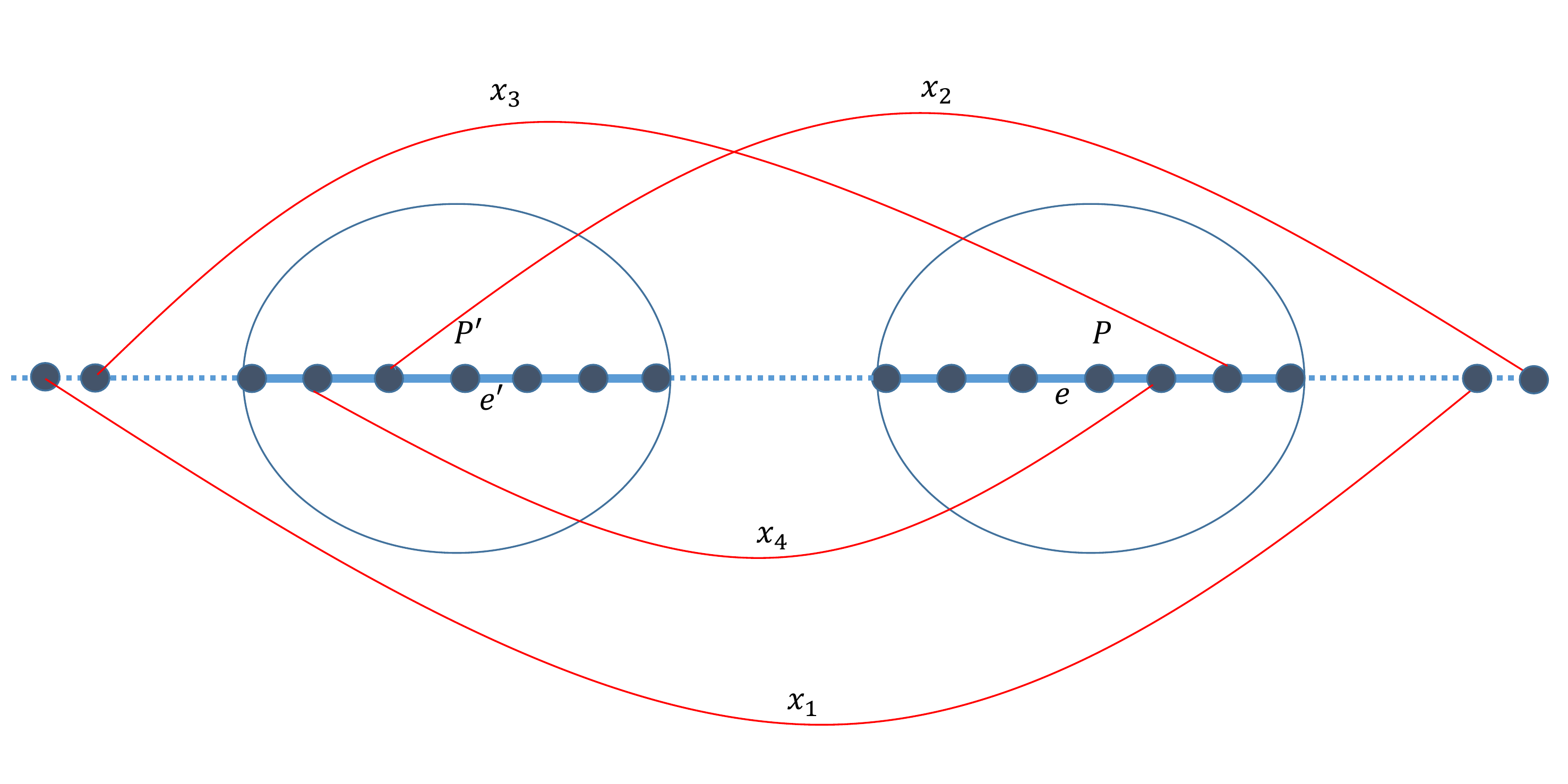}
 \caption{\small Example of two highways $P'$ and $P$, for the special case that the tree is a path, and $P$ and $P'$ are subpaths of it. The edges $x_1,x_2,x_3,x_4$ are counted in $\extcov(P',P),\extcov(e',P),\extcov(e,P'),\cov_F(e',e)$, respectively.}
\label{highway_highway}
\end{figure}

In Appendix \ref{sec:app_short_paths_hh}, we show that these are indeed all options.

\begin{restatable}{claim}{covHighway}\label{claim_cov_highways}
Let $e' \in P', e \in P$ be two tree edges in the highways $P',P$. Then, $$\cov(e',e) = \extcov(P',P) + \extcov(e',P) + \extcov(e,P') + \cov_{F}(e',e).$$
\end{restatable}

In Appendix \ref{sec:app_short_paths_hh}, we explain how we compute the different ingredients in the expression $\cov(e,e').$ Intuitively, $\cov_{F}(e',e)$ can be computed by an aggregate computation in $F_{P}$ or $F_{P'}$, $\extcov(e',P)$ can be computed by an aggregate computation inside $F_{P'}$, and $\extcov(e,P')$ can be computed using an aggregate computation inside $F_{P}$.
The value $\extcov(P,P')$ on the other hand is computed using an aggregate computation over a BFS tree. Note that this requires \emph{global} communication, as the edges counted here are not known to vertices in $F_P \cup F_{P'}$. Hence, to get an efficient algorithm, in our final algorithm we need to make sure not to compare too many different pairs of fragment highways.


\begin{restatable}{claim}{claimCovH}\label{clm:two-short-highway}
Let $P,P'$ be two highways. Using one aggregate and one broadcast computation in a BFS tree, all vertices can learn $\extcov(P,P')$. Computing this value for $k$ different pairs takes $O(D+k)$ time.
\end{restatable}


\begin{restatable}{claim}{ClaimCovHF}\label{claim_covl}
Let $P,P'$ be highways of the fragments $F_P$ and $F_{P'}$, respectively, and let $e \in P$.
Using one aggregate computation in $F_{P'}$, all edges $e' \in P'$ can compute the value $\cov_{F}(e',e)$.
\end{restatable}


\begin{restatable}{claim}{ClaimCoveP}\label{claim_cov_eP}
In $O(\dfrag + \nfrag)$ time, all edges $e$ that are in a highway can compute the values $\extcov(e,P)$ for all highways $P$.
\end{restatable}

\subsubsection*{Comparing two paths using a routing trick}
We next explain how given two highways $P_0,P_1$, we compute $\cut(e',e)$ for edges $e' \in P_0, e \in P_1$. 
Again, we break into two cases according to the existence of an edge between $F_{P_0}$ and $F_{P_1}$. In the case there is no edge, we prove that the cut value can be broken up to two parts, one only depends on information known to $P_0$ and one only depends on information known to $P_1$, we exploit it to limit the amount of global communication required.\\[-7pt]

\textbf{Case 1: there is no edge between $F_{P_0}$ and $F_{P_1}$.}

\begin{lemma} \label{lem:paths_h_h-noedge}
Let $P_0,P_1$ be two disjoint highways, and assume that there is no edge between $F_{P_0}$ and $F_{P_1}$. 
Then, given the values $\extcov(P_0,P_1)$, using one aggregate computation in $F_{P_0}$ and in $F_{P_1}$, 
and by communicating $O(\log n)$ information over a BFS tree, all vertices learn the values $e',e, \cut(e',e)$ for edges $e' \in P_0, e \in P_1$ such that $\cut(e',e)$ is minimal.
\end{lemma}

\begin{proof}
Let $e' \in P_0, e \in P_1$. By Claim \ref{claim_cover}, we have $\cut(e,e') = \cov(e) + \cov(e') - 2\cov(e,e').$ Also, by Claim \ref{claim_cov_highways}, $\cov(e',e) = \extcov(P_0,P_1) + \extcov(e,P_0) + \extcov(e',P_1) + \cov_{F}(e',e)$. Now since there are no edges between $F_{P_0}$ and $F_{P_1}$, by definition $\cov_{F}(e',e) = 0$ as this is the sum of costs of edges between $F_{P_0}$ and $F_{P_1}$ that cover $e'$ and $e$. Also, $\extcov(P_0,P_1)$ does not depend on the specific choice of $e'$ and $e$. Hence, if we want to minimize the expression $\cut(e,e')$, it is equivalent to minimizing the expression $(\cov(e) - 2\extcov(e, P_0)) + (\cov(e') - 2\extcov(e',P_1))$. Note that the expression $\cov(e) - 2\extcov(e, P_0)$ does not depend on $e'$, and the expression $\cov(e') - 2\extcov(e',P_1)$ does not depend on $e$. Hence, minimizing the whole expression is equivalent to minimizing each one of the expressions separately.
We next show that we can find $e \in P_1$ such that $\cov(e) - 2\extcov(e, P_0)$ is minimal, as well as compute this value using one aggregate computation in $P_1$. Similarly, we can find $e' \in P_0$ such that  $\cov(e') - 2\extcov(e',P_1)$ is minimal using one aggregate computation in $P_0.$
Then, if $P_0$ sends the message $P_0, e', \cov(e') - 2\extcov(e',P_1)$, and $P_1$ sends the message $P_1, e, \cov(e) - 2\extcov(e, P_0)$, for the minimal edges found, using a BFS tree, all vertices can compute the values $e',e,\cut(e',e)$ for the edges $e' \in P_0, e \in P_1$ such that $\cut(e',e)$ is minimal. As discussed above, this value equals to $(\cov(e) - 2\extcov(e, P_0)) + (\cov(e') - 2\extcov(e',P_1)) - 2\extcov(P_0,P_1)$, and we assume that $\extcov(P_0,P_1)$ is known.

Hence, to complete the proof, we explain how to compute $\cov(e) - 2\extcov(e, P_0)$ in $P_1$ (the equivalent computation in $P_0$ is done in the same way). From Claim \ref{claim_cov_eP}, all edges $e \in P_1$ know the value $\extcov(e,P_0)$, and they also know $\cov(e)$. Hence, each edge $e \in P_1$ knows the value $\cov(e) - 2\extcov(e, P_0)$. To find the edge $e$ that minimizes this expression, we only need to run one aggregate computation in $P_1$ for finding the minimum.
\end{proof}

\textbf{Case 2: there is an edge between $F_{P_0}$ and $F_{P_1}.$}

\begin{lemma} \label{lemma_paths}
Let $P_0,P_1$ be two disjoint highways, and assume that there is an edge $f$ between $F_{P_0}$ and $F_{P_1}$. 
Let $E_0 \subseteq P_0$ be a set of edges in $P_0$ we compare to $P_1$, and $E_1 \subseteq P_1$ be a set of edges in $P_1$ that we compare to $P_0$.
For $i \in \{0,1\}$, at the beginning of the computation, the following information is known by all vertices in $F_{P_i}$: 
\begin{enumerate}
\item The identity of all the edges $E_i$.
\item For each edge $e \in E_i$, the values $\cov(e),\extcov(e,P_{1-i})$.
\item The value $\extcov(P_0,P_1).$
\item The identity of the edge $f$.
\end{enumerate}
We can compute the values $\{\cut(e',e)\}_{e' \in P_0, e \in P_1}$ in the following two ways.
\begin{enumerate}
\item In $O(|E_1|)$ aggregate and broadcast computations in $F_{P_0}$, where at the end of the computation each edge $e' \in E_0$ would know the values $\cut(e',e)$ for all edges $e \in E_1$. 
\item In $O(|E_0|)$ aggregate and broadcast computations in $F_{P_1}$, where at the end of the computation each edge $e' \in E_1$ would know the values $\cut(e',e)$ for all edges $e \in E_0$. 
\end{enumerate}
\end{lemma}

\begin{proof}
We focus on the case where the computations are done in $F_{P_1}$, the second case is symmetric.
We work as follows. First, we use the edge $f$ to pass information from $F_{P_0}$ to $F_{P_1}$. Note that all vertices in $F_{P_0}$, know all the values $\{\cov(e),\extcov(e,P_1)\}_{e \in E_0}$. Since $f$ has one endpoint in  $F_{P_0}$ it knows this information and can pass it to all vertices in $F_{P_1}$ using $O(|E_0|)$ aggregate and broadcast computations. Now, for each edge $e \in E_0$, we run one aggregate computation in $F_{P_1}$, that allows each edge $e' \in P_1$ compute $\cut(e,e').$ This is done as follows. First, by Claim \ref{claim_cover}, we know that $\cut(e,e') = \cov(e) + \cov(e') - 2\cov(e,e').$ Also, by Claim \ref{claim_cov_highways}, $\cov(e',e) = \extcov(P_0,P_1) + \extcov(e',P_0) + \extcov(e,P_1) + \cov_{F}(e',e)$. Now, for $e' \in E_1$, it already knows $\cov(e'), \extcov(e',P_0), \extcov(P_0,P_1)$ at the beginning of the computation. Also, it learns the values $\cov(e),\extcov(e,P_1)$ as $f$ passed this information to $F_{P_1}$. Hence, the only thing missing to complete the computation is computing $\cov_{F}(e',e)$ which requires one aggregate computation in $F_{P_1}$ per edge $e \in E_0$ by Claim \ref{claim_covl}. To compute the values $\cut(e',e)$ for all $e \in E_0$, we pipeline $O(|E_0|)$ such computations, which completes the proof.
\end{proof}

\paragraph{Both edges in the same highway.}
We can deal with the case that both edges are in the same fragment highway using similar ideas, the details are deferred to Appendix \ref{sec:both-edge-highway}.
\section{Monotonicity and Partitioning} \label{sec:mon_part}

\subsection{Monotonicity} \label{sec:monotone}

We next discuss another crucial building block for our algorithm, \textit{monotonicity}. This property shows that the minimum 2-respecting cuts in the graph behave in a certain monotone structure, which can be exploited to obtain a fast algorithm. We start with describing the property in Claim \ref{claim_monotonicity}, and later explain how to exploit it to obtain a certain partitioning.
This property is also discussed in \cite{mukhopadhyay2019weighted}, where it is phrased in a slightly different manner related to relevant matrices (see Claim 3.6 in \cite{mukhopadhyay2019weighted}). For completeness, we next provide a self-contained proof that fits our description of the property.

\begin{claim} \label{claim_monotonicity}
Let $P_0, P_1$ be two ancestor to descendant paths in the tree, such that $P_0$ and $P_1$ are either orthogonal or one of them is strictly above the other in the tree. Let $v_0 \in P_0$ be the endpoint in $P_0$ that is closest to $P_1$, and $v_1 \in P_1$ be the endpoint of $P_1$ that is closest to $P_0$, and let $t$ be some vertex in the tree path between $v_0$ and $v_1$. Let $E_0 \subseteq P_0, E_1 \subseteq P_1$ be subsets of edges. The following holds. Let $e^0_1,e^0_2$ be edges in $E_0$, where $e^0_2$ is closer to $t$. Denote by $e^1_1, e^1_2$ the edges in $E_1$ such that $\cut(e^0_1,e^1_1)$ is minimal and $\cut(e^0_2,e^1_2)$ is minimal, taking the edges closest to $t$ if there is more than one option. 
Then either $e^1_2=e^1_1$ or $e^1_2$ is closer to $t$ compared to $e^1_1.$
\end{claim}

\setlength{\intextsep}{2pt}
\begin{figure}[h]
\centering
\setlength{\abovecaptionskip}{-2pt}
\setlength{\belowcaptionskip}{6pt}
\includegraphics[scale=0.5]{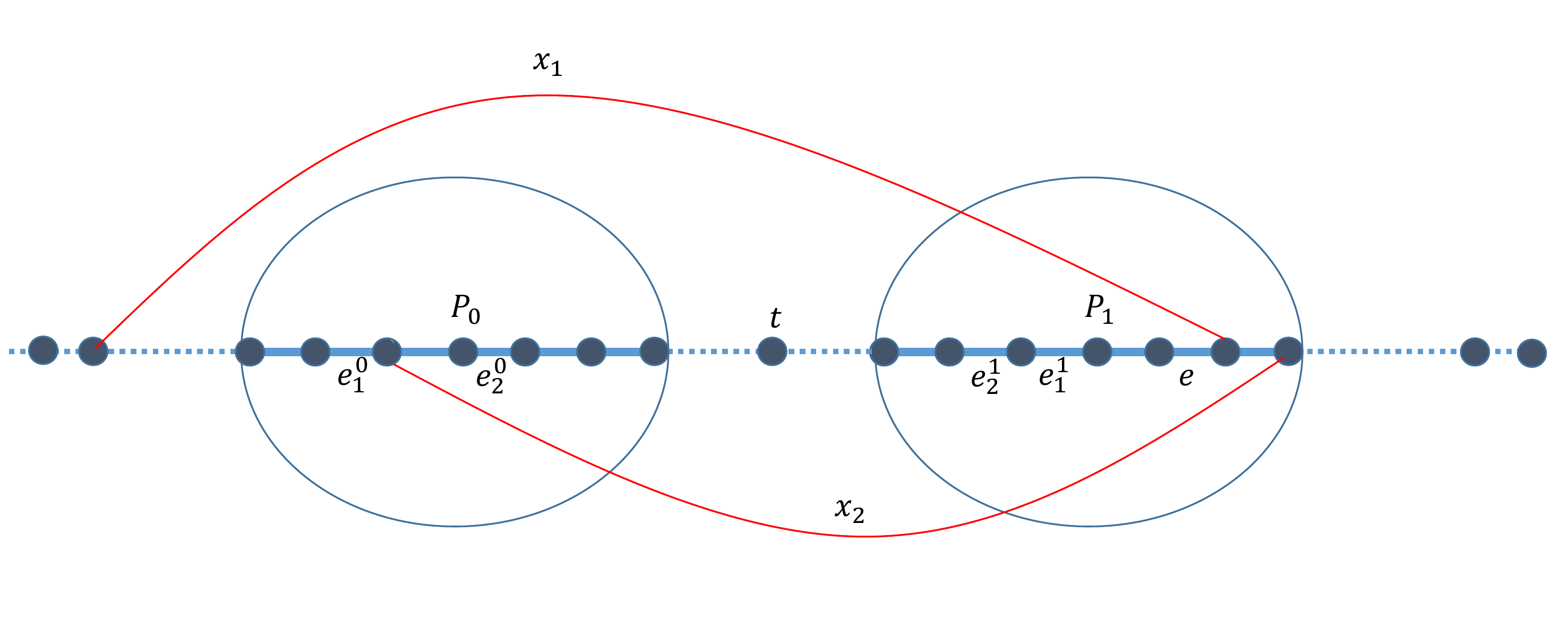}
 \caption{\small An illustration for monotonicity. If $e^0_2$ is closer to $t$ compared to $e^0_1$, then $e^1_2$ is closer to $t$ compared to $e^1_1$. The edge $x_1$ is an example of an edge that covers $e$ and $e^0_2$ and also covers $e^0_1$, and the edge $x_2$ is an example of an edge that covers $e$ and $e^0_2$ but not $e^0_1$. Note that any edge that covers $e$ and $e^0_2$ also covers any edge in $P_1$ closer to $t$ compared to $e$.}
\label{monotonicity_pic}
\end{figure}

\begin{proof}
Let $e \in E_1$. By Claim \ref{claim_cover}, $\cut(e^0_1,e) = \cov(e^0_1)+\cov(e) - 2\cov(e^0_1,e).$ Since $\cut(e^0_1,e^1_1)$ is minimal, it follows that $\cov(e^1_1) - 2\cov(e^0_1,e^1_1) \leq \cov(e) - 2\cov(e^0_1,e)$ for all $e \in E_1$.
Similarly, for $e \in E_1$, we have $\cut(e^0_2,e) = \cov(e^0_2)+\cov(e) - 2\cov(e^0_2,e).$ To find an edge $e \in E_1$ that minimizes this expression, we need to minimize $\cov(e) - 2\cov(e^0_2,e)$. We write it as $C_1(e)+C_2(e)$, where $C_1(e) = \cov(e) - 2\cov(e^0_1,e)$, $C_2(e) = - 2\cov(e^0_2,e) + 2 \cov(e^0_1,e)$. From the discussion above, we have that $C_1(e^1_1) \leq C_1(e)$ for all $e \in E_1$. We next take a closer look at $C_2(e)$. Note that  $\cov(e^0_2,e) - \cov(e^0_1,e)$ is the cost of all edges that cover $e^0_2$ and $e$ but do not cover $e^0_1$ (since $e^0_2$ is closer to $t$ than $e^0_1$, we have that any edge that covers $e^0_1$ and $e$ also covers $e^0_2$ as it is on the path between $e^0_1$ and $e$, but there may be additional edges that cover $e^0_2$ and $e$). This expression is monotonic in the following sense. If $e', e \in E_1$ where $e'$ is closer to $t$, any edge that covers $e^0_2$ and $e$ also covers $e'$ as its on the path between $e^0_2$ and $e$, hence this expression increases when we go towards $t$. As $C_2(e) = -2 (\cov(e^0_2,e) - \cov(e^0_1,e))$, we have that $C_2(e)$ is monotonically decreasing when we go towards $t$. To sum up, if we look at the sum $C_1(e) + C_2(e)$ that we want to minimize, for all edges $e$ that are farther from $t$ compared to $e_1^1$, we have that $C_1(e^1_1) + C_2(e^1_1) \leq C_1(e) + C_2(e)$. As we want to find the edge $e^1_2$ that minimizes the expression and is closest to $t$, it follows that $e^1_2$ is either equal to $e^1_1$ or closer to $t$ compared to it. This completes the proof.
\end{proof}

\subsection{Partitioning} \label{sec:partitioning}

Here we use monotonicity to obtain a certain partitioning. This partitioning is useful for the case that we have one short path $P'$ of length $O(\dfrag)$, and we should compare it to a long path composed of highways, that may have length $\Omega(n)$. This would be later crucial for the algorithm where we look for the min 2-respecting cut that has at least one edge in a highway. 

We start by discussing the case that the short path $P'$ is a non-highway, and later discuss the case it is a highway.
Let $P'$ be a non-highway that we want to compare to a long path $P_H$ composed of highways. For our algorithm, we need to look at two cases, that $P_H$ is either completely orthogonal to $P'$ or completely above $P'$. 
While we can compare $P'$ to a specific highway $P \in P_H$ in $O(\dfrag)$ time using an algorithm comparing two short paths, the challenge here is that we want to do many such computations efficiently. If we want to compare each edge of $P'$ to each edge in $P_H$ (that may have linear size), the total amount of information if too high. The goal of this section is to use monotonicity to break $P'$ to smaller parts, such that in average we compare each edge of $P'$ only to a constant number of highways in $P_H$, which is crucial for obtaining a small complexity. We next discuss the partitioning. 
We prove the following. See Figure \ref{partitioning_pic} for an illustration. 

\begin{lemma} \label{lemma_partitioning}
Let $P'$ be a non-highway, and let $P_H$ be a path of highways, such that $P_H$ is completely orthogonal to $P'$ or completely above $P'$. Denote by $P_1,...,P_k$ the different highways in $P_H$ going from the lowest to highest in the tree. Then, we can break the edges of $P'$ to (not necessarily disjoint) subsets $E'_1,...,E'_k$, such that the following holds.
\begin{enumerate}
\item $\sum_{i=1}^k |E'_i| = O(\dfrag + k).$
\item It is enough to solve the min 2-respecting cut problem on the pairs $\{P_i, E'_i\}_{i=1}^k$. More formally, if we denote by $\{e_i,e'_i\}$ the two edges $e_i \in P_i, e'_i \in E'_i$ such that $\cut(e_i,e'_i)$ is minimal, and denote by $j$ an index such that $\cut(e_j,e'_j) \leq \cut(e_i,e'_i)$ for all $1 \leq i \leq k$, then $\cut(e_j,e'_j) = \cut(e,e')$, where $e \in P_H, e' \in P'$ are edges such that $\cut(e,e')$ is minimal. 
\item Assume that all vertices know the values $\{e,\cov(e)\}$, for all edges $e$ that are highest or lowest in some highway. Then, we can compute the sets $E'_i$ in $O(\dfrag + k)$ time. At the end of the computation all the vertices in $T(P')$ know the identity of all edges in the set $E'_i$, for all $1 \leq i \leq k$. This can be done in different orthogonal non-highways simultaneously.
\end{enumerate}
\end{lemma}


The proof of Lemma \ref{lemma_partitioning} breaks down to three claims. We start by defining the sets $E'_i$, and then show they satisfy the required properties.
Recall that $P_1,...,P_k$ are the different highways in $P_H$ going from the lowest to highest in the tree. We denote by $e_{i,1},e_{i,2}$ the lowest and highest edges in the highway $P_i$, respectively. For $b \in \{1,2\}$, we denote by $e'_{i,b}$ the edge in $P'$ such that $\cut(e'_{i,b},e_{i,b})$ is minimal, taking the highest such edge if there is more than one option. We denote by $E'_i$ all the edges in $P'$ between $e'_{i,1}$ to $e'_{i,2}$. Using monotonicity, we have the following.

\begin{claim} \label{claim_xP'}
$\sum_{i=1}^k |E'_i| = O(\dfrag + k).$
\end{claim}

\begin{proof}
We next use Claim \ref{claim_monotonicity} to show that the different sets $E'_i$ are almost disjoint. The set $E'_i$ includes all edges between $e'_{i,1}$ and $e'_{i,2}$. As $P'$ is either orthogonal or below $P_H$, the highest vertex in $P'$ is the closest to $P_H$, we denote it by $t$.
From Claim \ref{claim_monotonicity}, we have that $e'_{i,2}$ is closer to $t$ compared to $e'_{i,1}$. Moreover, as all paths $P_j$ for $j>i$ are closer to $t$ compared to $P_i$, it follows from Claim \ref{claim_monotonicity} that all the edges $e'_{j,b}$ for $j > i, b \in \{1,2\}$ are either closer to $t$ than $e'_{i,2}$ or equal to $e'_{i,2}$. Similarly, as $P_i$ is closer to $t$ than the paths $P_j$ for $j < i$, we have that all the edges $e'_{j,b}$ for $j<i, b \in \{1,2\}$ are either equal to $e'_{i,1}$ or below it. To summarize, all edges $e'_{j,b}$ for $j \neq i$ are either equal to one of $e'_{i,1},e'_{i,2}$ or strictly above or below the whole set $E'_i$. It follows that all the edges in $E'_i$ except maybe two, are not contained in any of the sets $E'_j$ for $j \neq i$ (as both edges $e'_{j,1},e'_{j,2}$ are either strictly above or strictly below internal edges of $E'_i$). 

This gives $\sum_{i=1}^k |E'_i| = O(\dfrag + k)$. The $O(\dfrag)$ term counts internal edges in the sets $E'_i$ that are only contained in one of the sets $E'_i$, their number is bounded by the length of $P'$, which is $O(\dfrag)$. The second term $O(k)$ counts the edges $\{e'_{i,1}, e'_{i,2}\}_{1 \leq i \leq k}$. Note that such an edge may be included in more than one set (for example, we may have $e'_{i,1}=e'_{j,1}$), however per set $E'_i$ we only have two such edges, that sums to $2k$ in total.
\end{proof}

We next show that it is enough to focus on the sub-problems defined by the pairs $\{P_i, E'_i\}_{i=1}^k$.

\begin{claim}
Let $$(e',e) = \arg\min_{\{e' \in P', e \in P_H\}} \cut(e',e),$$ $$(e'_i,e_i)= \arg\min_{\{e' \in E'_i, e \in P_i\}} \cut(e',e).$$ Let $j$ be an index such that $\cut(e'_j,e_j) \leq \cut(e'_i,e_i)$ for all $1 \leq i \leq k$, then $\cut(e'_j,e_j) = \cut(e',e)$.
\end{claim}

\begin{proof}
Let $P_i \in P_H$ be the highway such that $e \in P_i$. If $e=e_{i,b}$ for $b \in \{1,2\}$, the edge $e'_{i,b} \in E'_i$ is an edge in $P'$ such that $\cut(e_{i,b},e'_{i,b})$ is minimal, hence $\cut(e'_i,e_i)=\cut(e',e)$, and we are done. Otherwise, since $e$ is in the tree path between $e_{i,1}$ to $e_{i,2}$, from Claim \ref{claim_monotonicity}, it follows that there is an edge $e' \in P'$ in the path between $e'_{i,1}$ to $e'_{i,2}$ such that $\cut(e,e')$ is minimal. In more detail, we again denote by $t$ the highest vertex in $P'$, which is the vertex in $P'$ closest to $P_H$. As $e_{i,2}$ is the highest edge in $P_i$, it is closer to $t$ compared to $e$, which means that the edge $e'_{i,2}$ is equal or closer to $t$ than the edge $e' \in P'$ such that $\cut(e',e)$ is minimal. Similarly, $e$ is closer to $t$ compared to $e_{i,1}$, which means that $e'$ is equal or closer to $t$ compared to $e'_{i,1}$. To sum up, $e'$ is between the edges $e'_{i,1}$ to $e'_{i,2}$, hence by definition $e' \in E'_i$, which gives $\cut(e'_i,e_i) = \cut(e',e)$, as needed.    
\end{proof}

We next explain how to compute the sets $E'_i$.

\begin{claim}
Assume that all vertices know the values $\{e, \cov(e)\}$, for all edges $e$ that are highest or lowest in some highway. Then, we can compute the sets $E'_i$ in $O(\dfrag + k)$ time. At the end of the computation all the vertices in $T(P')$ know the identity of all edges in the set $E'_i$, for all $1 \leq i \leq k$. This can be done in different non-highways that are not in the same root to leaf path simultaneously.
\end{claim}

\begin{proof}
First, we let all vertices in $T(P')$ learn the complete structure of $P'$, since $P'$ has length $O(\dfrag)$, this can be done in $O(\dfrag)$ time using upcast and broadcast in $T(P')$. Next, we compute the edges $e'_{i,b}$ for $1 \leq i \leq k, b \in \{1,2\}$. Recall that $e'_{i,b}$ is the edge in $P'$ such that $\cut(e'_{i,b},e_{i,b})$ is minimal. Since all vertices know the values $\{e_{i,b}, \cov(e_{i,b}) \}$ (as the edges $e_{i,b}$ are the highest or lowest in the highway $P_i$), we can use Claim \ref{claim_non_highway}, to compute these edges. In more detail, if we fix an edge $e=e_{i,b}$, using one aggregate computation all edges $e' \in P'$, learn the value $\cut(e',e)$. To let all vertices in $T(P')$ learn the highest edge $e' = e'_{i,b}$ such that $\cut(e',e)$ is minimal, we use convergecast and broadcast in $T(P')$. Using pipelining, all vertices in $T(P')$ can learn all the edges $e'_{i,b}$, which requires $O(k)$ aggregate and broadcast computations, this takes $O(\dfrag + k)$ time. Since all vertices in $T(P')$ know the complete structure of $P'$, they can deduce the sets $E'_i$, as $E'_i$ is the set of all edges in $P'$ between $e'_{i,1}$ to $e'_{i,2}$. As the whole computation was done inside $T(P')$, we can work simultaneously in different non-highways not in the same root to leaf path, as their trees $T(P')$ are edge-disjoint.
\end{proof}

\begin{figure}
\centering
\begin{minipage}{.45\textwidth}
  \centering
  \includegraphics[width=.9\linewidth]{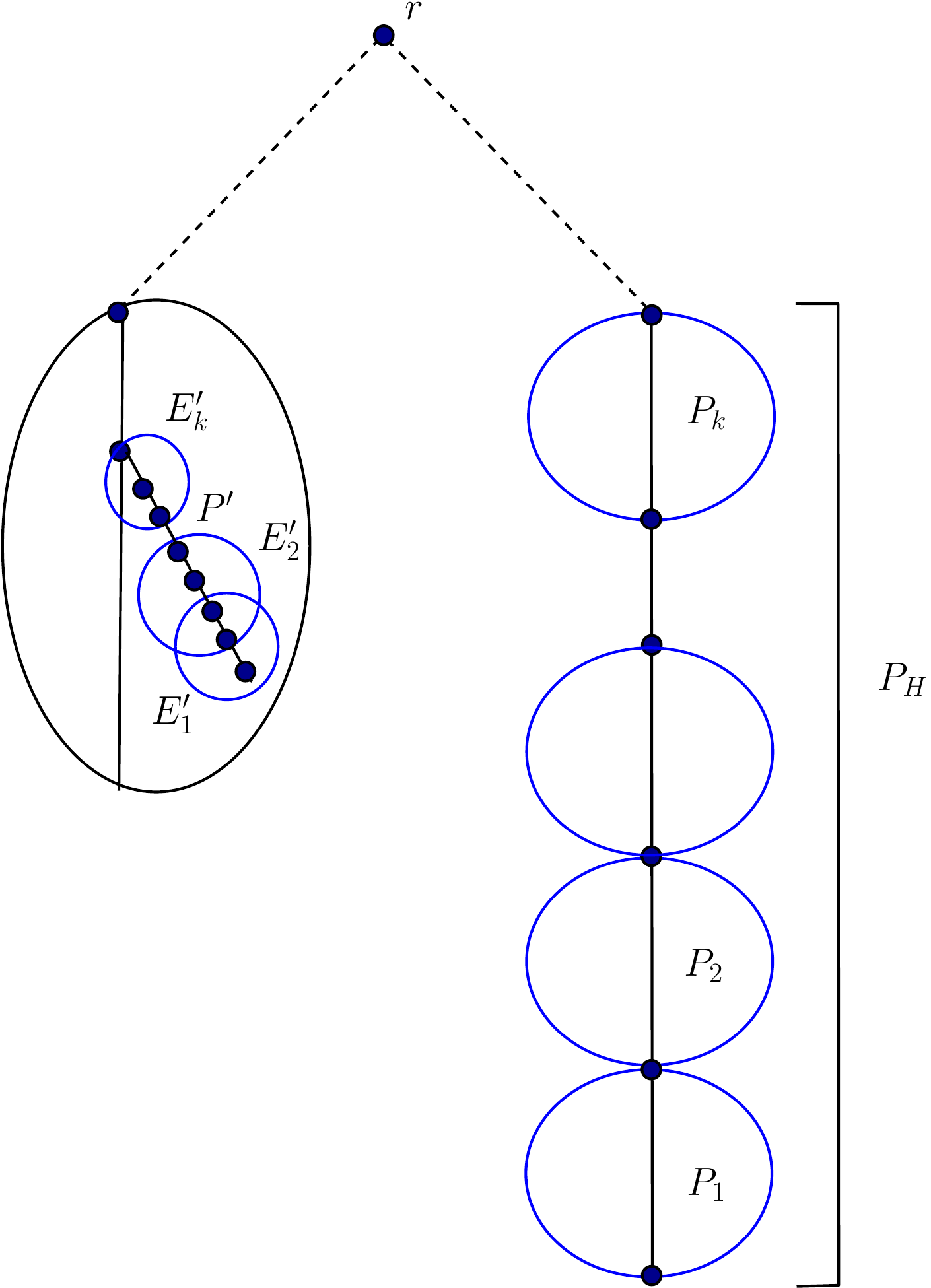}
  \captionof{figure}{\small An illustration of the partitioning. $P'$ is partitioned in almost disjoint $E'_1, \cdots, E'_k$ w.r.t. $P_1, \cdots, P_k$ which are highways of $P_H$.}
\label{partitioning_pic}
\end{minipage}%
\hspace{1 cm}
\begin{minipage}{.45\textwidth}
  \centering
  \includegraphics[width=.9\linewidth]{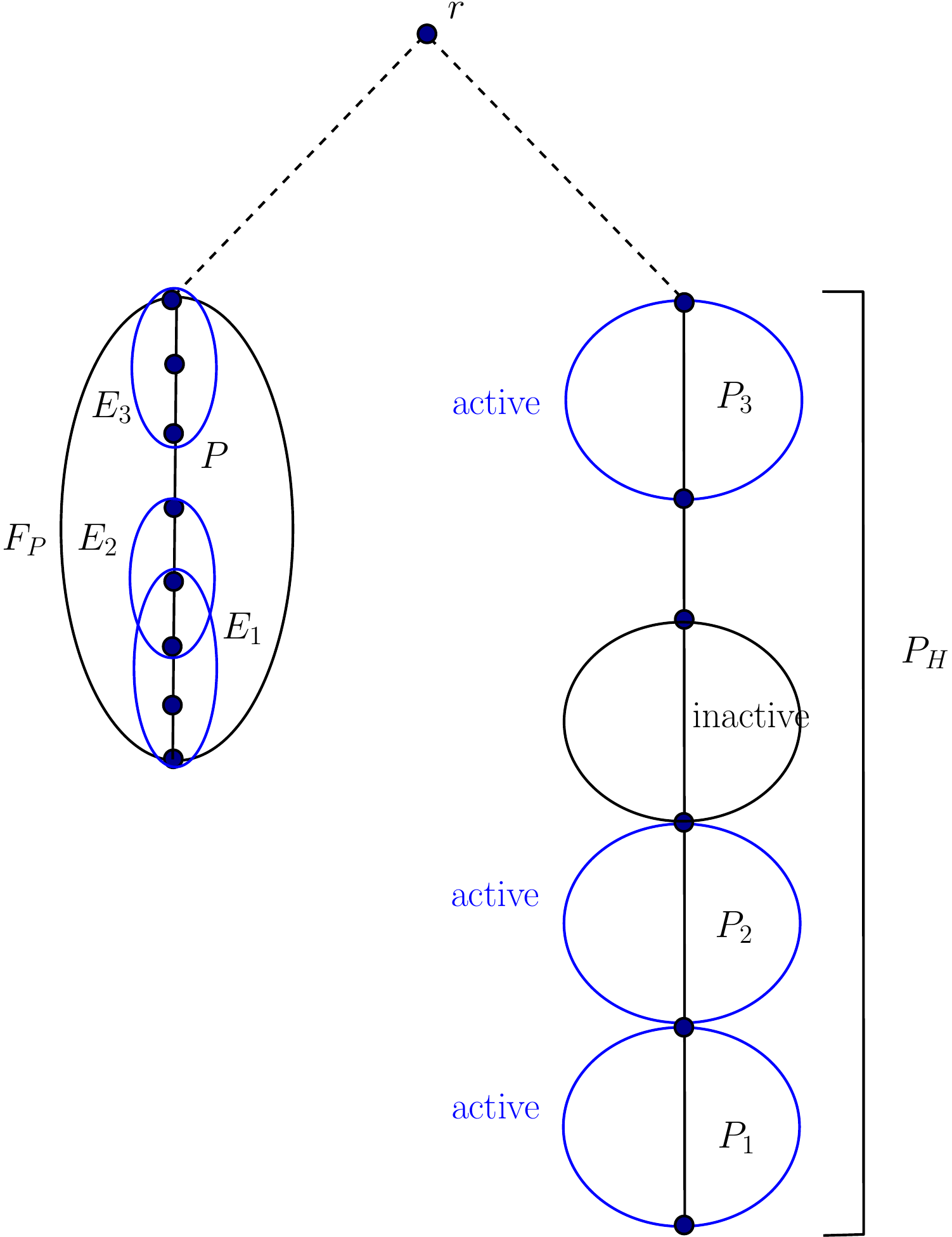}
  \captionof{figure}{\small Partitioning highway $P$ w.r.t. \textcolor{blue}{active highways} in $P_H$. $P_1, P_2, P_3$ are \textcolor{blue}{active highways} in $P_H$, and $P$ is partitioned in almost disjoint $E_1, E_2, E_3$ w.r.t. them.}
  \label{fig:hw-partition}
\end{minipage}
\end{figure}

We will also show a similar claim when $P$ is a highway inside a fragment. The only change here is the time required to compute the partition.

\begin{lemma} \label{lemma_partitioning-highway}
Let $P$ be a highway in a fragment, and let $P_H$ be a highway path such that $P$ is non-splittable w.r.t. $P$ (i.e., $P_H$ is either orthogonal to $P$, or an ancestor of a descendant of $P$). Among the highways in $P_H$, denote by $P_1,...,P_k$ the different highways in $P_H$ which are labeled active (i.e., each such $P_i$ is potentially interested in $P$).  Then, we can break the edges of $P$ to (not necessarily disjoint) subsets $E_1,...,E_k$, such that the following holds.
\begin{enumerate}
\item $\sum_{i=1}^k |E_i| = O(\dfrag + k).$ \label{clmitm:sum}

\item It is enough to solve the min 2-respecting cut problem on the pairs $\{P_i, E_i\}_{i=1}^k$. \label{clmitm:sufficient}

\item Assume that all vertices know the values $\{e',\cov(e')\}$, for all edges $e'$ that are highest or lowest in some highway. Then, we can compute the sets $E_i$ in $O(\dfrag + k + D)$ time. At the end of the computation all the vertices in $F_P$ know the identity of all edges in the set $E_i$, for all $1 \leq i \leq k$. \label{clmitm:compute}
\end{enumerate}
\end{lemma}


\begin{proof}
Property \ref{clmitm:sum} and \ref{clmitm:sufficient} can be proven exactly as in Lemma \ref{lemma_partitioning}. Note that here we are considering only the active highways in $P_H$, but that does not change much. As before, let $e_{i,1}$ and $e_{i,2}$ are the two extremal edges of $P_i$: If $P_H$ is orthogonal or ancestor of $P$, then they are the lowest and the highest edges of $P_i$, otherwise of $P_H$ is a descendant of $P$, then they are the highest and the lowest edges respectively. We find  $e'_{i,1}$ and $e'_{i,2}$ on $P$ as before, but this time w.r.t. only the active highways of $P_H$, and a similar argument shows Property \ref{clmitm:sum} and \ref{clmitm:sufficient} (See Figure \ref{fig:hw-partition} for a clearer idea). We will show how to achieve Property \ref{clmitm:compute} next. 

As in Lemma \ref{lemma_partitioning}, the computation boils down to finding the edges $e'_{i,b}$ for $1 \leq i \leq k$ and $b \in \set{1,2}$. Fix an edge $e = e_{i,b}$. We will show how every edge $e' \in P$ learns the value $\cut(e', e)$. We have assumed that every edge $e'$ knows the value $\set{e,\cov(e)}$ because $e$ is either the highest or the lowest in the highway $P_i$. Also, $e'$ knows the value $\cov(e')$. So it remains for $e'$ to know the value $\cov(e',e)$. Note that $\cov(e',e) = \cov_F(e',e) + \extcov(e, P) + \extcov(e', P_i) + \extcov(P, P_i)$ by Claim \ref{claim_cov_highways}. The value $\extcov(e, P)$ can be calculated inside $P_i$ by an aggregate computation and can be broadcast over a BFS tree of $G$. Similarly, the value $\extcov(P, P_i)$ can be calculated by one aggregate computation over a BFS tree of $G$ and can be broadcast. Each of these two operation requires $O(1)$ bits of aggregate and broadcast and takes time $O(D)$. 
The term $\cov_F(e',e) + \extcov(e', P_i)$ can be computed similar to that in Lemma \ref{lemma_partitioning}---We use one single aggregate computation inside $F_P$ for every edge $e' \in P$ learn the value of $\cov_F(e', e)$. At this point, each edge $e' \in P$ can compute $\cut(e',e)$. To let all the vertices of $F_P$ know the identity of $e'_{i,b}$ for which $\cut(e'_{i,b}, e)$ is the smallest, we just need to do convergecast and broadcast inside $F_P$. In total, this requires $O(1)$ bits of aggregate and broadcast inside $F_P$. Hence, this can be computed for every $e_{i,b}$ in a pipelined fashion which takes time $O(D + k)$ for aggregate and broadcast on a BFS tree of $G$ and $O(\dfrag + k)$ for aggregate and broadcast over $F_P$. Hence the total round complexity is $O(\dfrag + k + D)$.
\end{proof}

\begin{remark} \label{rem:partition}
A major difference between Lemma \ref{lemma_partitioning} and \ref{lemma_partitioning-highway} is the following: Computing the $E'_i$'s for different non-highway paths can be done simultaneously in Lemma \ref{lemma_partitioning} because the computation happens entirely within $E'_i$. In Lemma \ref{lemma_partitioning-highway}, however, $k$ many additional aggregate computation over the BFS tree is needed for each highway.  Later when we have to do it over many highways parallelly, we will see that we have to perform $\tO(\nfrag)$ many aggregate computations. When we pipeline them, the complexity of aggregate computation will be $\tO(D + \nfrag)$ and the total rounds complexity will be $\tO(\dfrag + \nfrag + D)$. 
\end{remark}
\section{Finding the min 2-respecting cut} \label{sec:2_respecting}

We will next explain how to find the 2 tree edges $e,e'$ that define the minimum cut. We divide to cases depending if the edges $e,e'$ are part of a highway or a non-highway. We start by explaining how we deal with the simple case that the cut is defined by one tree edge, and then focus on the case that the cut is defined by two edges.

\subsection{1-respecting cut} \label{sec:1-resp-cut}
 From Claim \ref{claim_1_respecting}, for each tree edge $e$, the value of the 1-respecting cut defined by $e$ is $\cov(e)$, and is known to $e$, by Claim \ref{claim_learn_cov}. Thus, in $O(D)$ rounds, all the network can know the value $\min\limits_{e\in E} \cov(e)$ of the min 1-respecting cut, as well as the edge $e$ minimizing this expression. From here on we assume that the the min cut is attained as the 2-respecting cut of some pair of tree edges. 

\subsection{Simple cases with non-highways} \label{ssec:finding_cut_non_highway}

We show how to compare all non-highway edges to all the edges in their fragment, as well as to all non-highways they are potentially interested in in other fragments.
The general idea is simple. From Claim \ref{claim_non_highway_fragment}, we can compare the edges of a non-highway path $P'$ to the edges of the fragment $F_{P'}$ in $O(\sfrag)$ time. Similarly, we can compare $P'$ to all edges of a different fragment in $O(\sfrag)$ time. From Corollary \ref{corol:-layering_interested}, we know that we only need to compare a non-highway $P'$ to non-highways in $O(\log{n})$ different fragments, and these fragments are known to all vertices in $T(P')$. Hence, overall, we can compare $P'$ to all these fragments in $\tilde{O}(\sfrag)$ time. Moreover, the computations can be done in parallel for different orthogonal non-highways $P'$. As non-highways in the same layer are orthogonal, we process the graph according to the $O(\log{n})$ layers, and in iteration $i$, in $\tilde{O}(\sfrag)$ time, take care of all non-highways in layer $i$. This gives the following, for a full proof see Appendix \ref{sec:app_8_nh}. 

\begin{restatable}{claim}{nhtwoedges} \label{claim_nh_2respecting}
Let $e,e'$ be a pair of edges that minimize $\cut(e,e')$ such that $e,e'$ are either two non-highway edges in different fragments, or two edges in the same fragment where at least one of them is a non-highway edge. In $\tilde{O}(D+\sfrag)$ time all the vertices in the graph learn the values $\{f,f',\cut(f,f')\}$ for a pair of edges such that $\cut(f,f') \leq \cut(e,e')$.
\end{restatable}

\subsection{Exactly one cut edge in a highway} \label{ssec:nh-high}
Here we discuss the case that the 2-respecting cut is defined by two edges $e',e$ such that $e'$ is in a non-highway $P'$, and $e$ is in a highway $P$ in different fragment. The case that $e$ is in a highway in the same fragment was already discussed in Section \ref{ssec:finding_cut_non_highway}.
We will deal separately with the case that there is an edge between $T(P')$ and the fragment $F_P$ of $P$, and the case there is no such edge. To do so, we first show the following.

\begin{claim} \label{claim_learn_edge}
In $\tilde{O}(\dfrag + \nfrag)$ time, for all non-highways $P'$, all vertices in $T(P')$ know for each fragment $F$ whether there is an edge between $T(P'^{\downarrow})$ and $F$, and the identity of an edge between $T(P'^{\downarrow})$ and $F$ if exists.
\end{claim}

\begin{proof}
We work in $O(\log n)$ iterations according to the layers. In iteration $i$, we take care of non-highways $P'$ in layer $i$. We work as follows. Given a fragment $F$, we run an aggregate computation in $T(P')$ to learn the identity of the first edge between $T(P'^{\downarrow})$ and $F$ if exists (for this, we use the fact that both endpoints of an edge can learn the fragments these endpoints belong to), then we broadcast the information to $T(P')$. To do so for all fragments, we run $O(\nfrag)$ aggregate and broadcast computations, which takes $O(\dfrag + \nfrag)$ time using pipelining. This can be done in parallel for all non-highways $P'$ in the same layer, as their trees $T(P')$ are edge-disjoint. Computing this for all layers, results in $\tilde{O}(\dfrag + \nfrag)$ time.  
\end{proof}

Next, we deal with the case that there is no edge between a non-highway and a highway. The main idea is that since there is no edge between the paths, one can employ Lemma \ref{lemma_nh_no_edge} in order to obtain the necessary information to compute the min 2-respecting cut between these paths.
\
\begin{claim} \label{claim_cut_nh_h_no_edge}
In $\tilde{O}(D + \dfrag + \nfrag)$ time, all vertices learn the values $\{e',e,\cut(e',e)\}$ for edges $e',e$ that minimize the expression $\cut(e',e)$, where $e'$ is in a non-highway $P'$, and $e$ is in a highway $P$, such that there is no edge between $T(P'^{\downarrow})$ and the fragment $F_P$ of $P$.
\end{claim}

\begin{proof}
We work in $O(\log n)$ iterations according to the layers. In iteration $i$, we take care of all non-highways $P'$ in layer $i$. 
We first let all edges $e' \in P'$ learn the values $\extcov(e',P)$ for all highways $P$, this takes $O(\dfrag + \nfrag)$ time using Claim \ref{claim_cov_eP_nh}, and can be done in all non-highways in the same layer simultaneously. Then, we let all vertices learn the values $\{e^P_{min}, \cov(e^P_{min})\}$ for all highways $P$, where $e^P_{min}$ is the edge $e \in P$ such that $\cov(e)$ is minimal. This takes $O(D+\nfrag + \dfrag)$ time by Claim \ref{claim_cov_min}. 

We next use this information to find the min 2-respecting cuts that have one edge in $P'$ and one edge in a highway $P$ such that there is no edge between $T(P'^{\downarrow})$ and the fragment $F_P$ of $P$. Note that all vertices in $T(P')$ know exactly the identity of all such highways from Claim \ref{claim_learn_edge}.
We next fix such highway $P$. We can use Lemma \ref{lemma_nh_no_edge}, to let all edges in $P'$, learn the values $\{e',e,\cut(e',e)\}$ for edges $e' \in P', e \in P$ such that $\cut(e',e)$ is minimal. This requires one aggregate and broadcast computations in $P'$. To do so for all such highways $P$, we do $O(\nfrag)$ computations, which takes $O(\dfrag + \nfrag)$ time, and can be done in parallel in different non-highways in layer $i$. To take care of non-highways in all layers, we have $O(\log n)$ iterations, which overall takes  $\tilde{O}(\dfrag + \nfrag)$ time.

After this, for each pair of a non-highway $P'$ and a highway $P$, where there is no edge between $T(P'^{\downarrow})$ and $F_P$, the vertices in $P'$ know the values $e',e,\cut(e',e)$ for edges $e' \in P', e \in P$ that minimize this expression. To learn the minimum such value over all pairs, we use convergecast and broadcast in a BFS tree, which takes $O(D)$ time.
\end{proof}

We next discuss the case there is an edge between a non-highway and a highway. Here, we use the partitioning described in Section \ref{sec:partitioning}, and bounds on the number of paths each path is potentially interested in from Section \ref{ssec:-interesting_lemma} to obtain a fast algorithm. Note that it is enough to compare a non-highway and a highway that are potentially interested in each other, as if $e' \in P',e \in P$ define the minimum 2-respecting cut, it holds that $P'$ and $P$ are potentially interested in each other.
The proof idea is as follows. First, we know that each non-highway is only potentially interested in $poly(\log{n})$ super-highways. To compare one non-highway $P'$ to a super-highway $P_H$ we use the path-partitioning lemma (Lemma \ref{lemma_partitioning}). After the partitioning, we route information from $P'$ to the fragment highways of $P_H$, which compute the relevant cut values. A fragment highway $P \in P_H$ only participates in the computation if it is potentially interested in a non-highway in the fragment $F_{P'}$, which bounds the total amount of computation. 
The full proof of the claim is deferred to Appendix \ref{sec:8_appendix}.

\begin{restatable}{claim}{clmshortnhlonghighedge} \label{clm:short-nh-long-high-edge}
In $\tilde{O}(D + \sfrag + \nfrag)$ time, all vertices learn the values $\{e',e,\cut(e',e)\}$ for edges $e',e$ that minimize the expression $\cut(e',e)$, where $e'$ is in a non-highway $P'$, and $e$ is in a highway $P$, such that there is an edge between $T(P'^{\downarrow})$ and the fragment $F_P$ of $P$, and such that $P'$ and $P$ are potentially interested in each other.
\end{restatable}

A schematic description of the algorithm for the non-highway-highway case appears in Algorithm \ref{alg:nh-h}.

\begin{center}
  \centering 
  \begin{minipage}[H]{0.8\textwidth}
\begin{algorithm}[H]
\caption{Schematic algorithm when exactly one edge is in a highway}\label{alg:nh-h}
\begin{algorithmic}[1]
\Require From Corollary \ref{corol:-_layer_highway}, for each non-highway bough $P'$ in layer $1 \leq i \leq L$, all vertices in $T(P')$ know a set of super-highways $P_H \in \intpot{P'}$. Each such super-highway is either completely above or completely orthogonal to the fragment of $P'$.
\Require From Corollary \ref{corol:-highway_non_highway}, for each fragment highway $P$ in a fragment $F_P$, all vertices in the fragment $F_P$ know the set of $O(\log n)$ fragments that contain non-highway paths that $P$ is potentially interested in, not including $F_P$. 
\Statex\hrulefill
\State For each non-highway $P'$ in layer $1\leq i \leq L$ and all fragments $F$, all vertices in $T(P')$ learn if there is an edge between $T(P'^{\downarrow})$ and the fragment $F$, and if so, the identity of such edge.
\Statex \Comment{See Claim \ref{claim_learn_edge}}

\For{every layer $1 \leq i \leq L$}
	\For{Every non-highway $P'$ in layer $i$ in parallel}
		\For{Every fragment highway $P$ where there is no edge between $T(P'^{\downarrow})$ and $F_P$}
		\State Compute $\{e',e,\cut(e',e)\}$ for $e' \in P', e \in P$ that minimize this expression.\label{line_no_edge}
		\Statex \Comment{See Claim \ref{claim_cut_nh_h_no_edge}}
		
		\EndFor
	\EndFor
\EndFor



\For{every layer $1 \leq i \leq L$}
	\For{Every non-highway $P'$ in layer $i$ in parallel}
		\For{Every super-highway $P_H \in \intpot{P'}$ (Represented by lowest fragment)}
			\State Let $P_1,...,P_k$ be the fragment highways of $P_H$.
			\State Partition the edges of $P'$ into sets $E'_1,...,E'_k$ such that we only need to compare $E'_i$ to $P_i$.
			\Statex \Comment{Use Lemma \ref{lemma_partitioning}}
			\For{Each fragment highway $P_j$ where there in edge $f$ between $T(P'^{\downarrow})$ and $F_{P_j}$ in parallel}
					\State Use $f$ to route the information $\{e',\cov(e'),\extcov(e',P_j)\}_{e' \in E'_j}$ from $P'$ to $P_j$.
					\State The cut values would be computed by the fragment highways $P_j$ that are potentially interested in $P'$, as described next. If $P_j$ is not potentially interested in $P'$ there is no need to compute the values.
					\Statex \Comment{See Claim \ref{clm:short-nh-long-high-edge}}
			\EndFor
		\EndFor
	\EndFor
	
	\For{every fragment highway $P$}
		\For{every fragment $F$ with set of edges $E_F$ that contains a non-highway path that $P$ is potentially interested in}
		\State Compute the values $\{\cut(e',e)\}_{e' \in E_F, e \in P}$ for all edges $e'$ where the values $\{e',\cov(e'),\extcov(e',P)\}$ were received from vertices in $F$, and specifically from the non-highway paths in $F$ that are potentially interested in $P$.
		\Statex \Comment{See Claim \ref{clm:short-nh-long-high-edge} and Lemma \ref{lemma_paths_nh_h}}
	\EndFor
\EndFor
\EndFor
 
\State Communicate over a BFS tree to let all vertices learn the values $\{e',e,\cut(e',e)\}$ for edges $e',e$ in the above cases that minimize $\cut(e',e).$
\end{algorithmic}
\end{algorithm}
\end{minipage}
\end{center}

\subsection{Both cut edges in highways}\label{ssec:highway_highway}

Now we turn to discussing the case of 2-respecting cuts when both cut edges $e$ and $e'$ are in different highway paths. The case when both edges are in the same highway will be discussed after that. We first show a claim analogous to Claim \ref{claim_learn_edge}.

\begin{claim} \label{claim_learn_edge_highway}
In $\tilde{O}(\dfrag + \nfrag)$ time, for all highways $P'$ (inside different fragments $F_{P'}$), all vertices in $F_{P'}$ know for each fragment $F$ whether there is an edge between $F_{P'}$ and $F$, and the identity of an edge between $F_{P'}$ and $F$ if exists.
\end{claim}

\begin{proof}
The proof is similar to that of Claim \ref{claim_learn_edge} as well. We run an aggregate computation inside $F_{P'}$ to learn the identity of the first edge between $F_{P'}$ and $F$, and if such an edge exist, we broadcast this information inside $F_{P'}$. This requires $O(\dfrag)$ rounds. As there are $\nfrag$ many fragments $F$, doing so for all fragments requires $\tO(\dfrag + \nfrag)$ rounds. We can do it for all short highways simultaneously as the fragments $F_{P'}$'s are disjoint for different short highways.
\end{proof}

The next will be a two step argument, each step requiring the monotonicity property as presented in Claim \ref{claim_monotonicity}: In the first step, we compute the complexity of comparing a highway path inside a fragment and a long highway path. This is where we use the path-partitioning trick on the short highway path to come up with an efficient algorithm---similar to what we have already seen in the previous section when we compared a non-highway path (which, by definition, is contained inside a fragment) and a long highway path. In the second step, we use this algorithm as a subroutine to come up with a divide and conquer technique for comparing a long highway to a long highway. This step also uses the monotonicity property of the minimum 2-respecting cut. Recall that, when we compare two long highway paths $P_{H_1}$ and $P_{H_2}$, a highway $P$ inside a fragment $F$ in either of the long highway paths is \textit{active} if $P$ is \textit{potentially interested} in the other long highway path. In both steps mentioned before, the complexity is in terms of the number of \textit{active} highways in the computation, and not in terms of the number of total highways in the computation. This is crucial because this will help us bound the complexity when we compare many pairs of long highway paths simultaneously---we will use Theorem \ref{thrm:pairing_up_highways} which bounds the number of pairs an \textit{active} highway takes part in is $O(\log^2 n)$. Next, we start with assuming the complexity of the first step, and show how to implement the second step. We then prove the complexity of the first step.

\begin{claim}\label{claim:blackbox_claim}
Let $P'$ be a highway completely contained in a fragment and $P_H$ be a highway path spread across many fragments $P_1, \cdots, P_k$ such that $P'$ is non-splittable w.r.t. $P_H$. Let $\ell$ many highways in $P_H$ are labelled as `active'. \begin{enumerate}
    \item In time $\tO(\ell + D + \dfrag)$ all vertices learn the value $\set{e, e',\cut(e,e')}$ for edges $e$ and $e'$ that minimize $\cut(e,e')$ where $e$ is in one of the `active' highways of $P_H$ and $e' \in P'$.
    
    \item The computation is done inside $P'$ and the active components of $P_H$ which is $\tO(\ell + \dfrag)$ bits of aggregate computation and, in addition, a broadcast of $\tO(\ell)$ bits of communication over the BFS tree of $G$ is performed.
    
\end{enumerate}
\end{claim}

\paragraph{An algorithm for two highway paths.} We will now show, assuming Claim \ref{claim:blackbox_claim}, how to compare different pairs of highways efficiently. To this end, we first focus on one pair of highway paths for now---this will showcase the divide and conquer technique that we want to employ. Later we show how to take care of all pairs of highways in parallel using Theorem \ref{thrm:pairing_up_highways}. Consider two highways $\high$ and $\way$ where $\high$ has $\ell_1$ many active short highways w.r.t. $\way$ and $\way$ has $\ell_2$ many active short highways w.r.t. $\high$. We make the following claim.

\begin{claim} \label{clm:two-long-highway}
Consider two highway paths $\high$ and $\way$ where $\high$ has $\ell_1$ many active highways w.r.t. $\way$ and $\way$ has $\ell_2$ many active highways w.r.t. $\high$. Also, assume that all vertices in $\high$ and $\way$ know the set of these active highways. 

Then there is an algorithm such that the minimum 2-respecting cut, where one edge from $\high$ and another tree edge from $\way$ is included, can be found in time $\tO(\ell_1 + \ell_2 + \dfrag +D)$. This computation requires $\tO(\ell_1 + \ell_2)$ bits of broadcast computation and at most $(\ell_1 + \ell_2 + \dfrag)$ bits of aggregate computation inside a fragment corresponds to an active highway. 
\end{claim}

\begin{proof}
Let us order the active components in $\high$ as $\high^1, \cdots, \high^{\ell_1}$, and similarly the active components of $\way$ are $\way^1, \cdots, \way^{\ell_2}$ such that $\high^1$ is the closest highway in $\high$ to $\way^1$ and vice versa (where the distance is measure via the unique path between $\high$ and $\way$, see Figure \ref{fig:two-highway} for reference). We first do the following two  comparisons, each between a short highway and a long highway:
\begin{enumerate}
    \item First $\high^{\ell_1/2}$ runs the algorithm from Claim \ref{claim:blackbox_claim} with $\way$ but only with the \textit{active} highways of $\way$. Let the tree edge from $\way$ taking part in this minimum 2-respecting cut is in component $\ell_i$ for $\way$. As we have learnt from Claim \ref{claim:blackbox_claim}, this requires $O(\ell_2 + \dfrag + D)$ rounds.
    
    \item Then $\way^{\ell_i}$ runs the algorithm from Claim \ref{claim:blackbox_claim} with $\high$. This requires $O(\ell_1 + \dfrag + D)$ rounds.
\end{enumerate}  
In total, these two comparisons require $O(\ell_1 + \ell_2 + \dfrag + D)$ rounds when run one after the other.

This gives rise to two disjoint subproblems (See Figure \ref{fig:two-highway} for reference): \begin{itemize}
    \item[(i)] Comparing the prefix of $\high$ starting from $\high^1$ till (but not including) highway $\ell_1/2$ (which we denote as $\high^t$) with the prefix of $\way$ starting from $\way^1$ till (but not including) highway $\ell_i$ (which we denote as $\way^t$); and,
    
    \item[(ii)] Comparing the suffix of $\high$ from (but not including) highway $\ell_1/2$ to $\high^{\ell_1}$ (which we denote as $\high^b$) with the suffix of $\way$ from (but not including) highway $\ell_i$ to $\way^{\ell_2}$ (which we denote as $\way^b)$.
\end{itemize}

\begin{figure}[h!]
    \centering
    \includegraphics[scale =0.6]{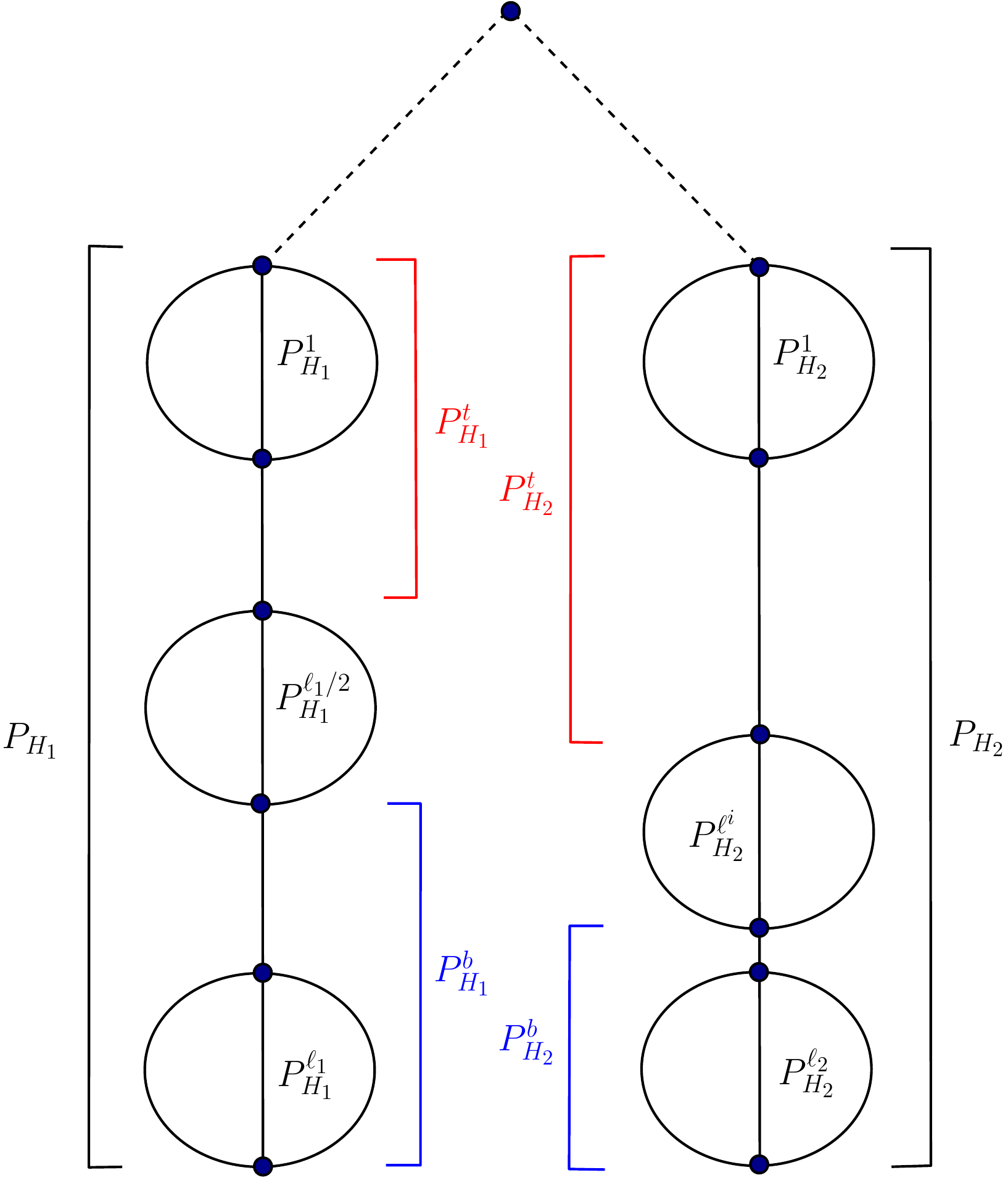}
    \caption{\small Recursion step of the algorithm for Claim \ref{clm:two-long-highway}}
    \label{fig:two-highway}
\end{figure}

Note that, because of the guarantee of Claim \ref{claim:blackbox_claim}, all edges of $\high$ and $\way$ know which subproblem they are included. Let the number of active components in $\high^t, \high^b, \way^t$ and $\way^b$ be $\ell_1^t, \ell_1^b, \ell_2^t, \ell_2^b$ respectively. Note that $\ell_1^t + \ell_1^b = \ell_1 - 1$, and $\ell_2^t + \ell_2^b = \ell_2 - 1$. We recurse parallelly in these two sub-problems by again choosing middle components in $\high^t$ and $\high^b$. The point to note over here is the following: Because comparing $\high^t$ with $\way^t$ and comparing $\high^b$ with $\way^b$ are disjoint subproblems, the computations within the fragments of the active highways for these two subproblems can be done parallelly. The total number of bits that need to be broadcast is $O(\ell_1^t + \ell_2^t)$ for the first subproblem, and $O(\ell_1^b + \ell_2^b)$ for the second subproblems, which can be pipelined: This requires $O(\ell_1^b + \ell_2^b + \ell_1^t + \ell_2^t + D)$ rounds, i.e, $O(\ell_1 + \ell_2 +D)$ rounds.  Hence, total number of rounds required in this iteration is $O(\ell_1 + \ell_2 + \dfrag + D)$. 

This argument extends in all levels of recursion, i.e., each level of recursion can be done in time $O(\ell_1 + \ell_2 + \dfrag + D)$. At the end, the algorithm compares many disjoint pairs of active highway and highway path in parallel such that the total number of active highways in all these pairs (including those in the highway paths) is $O(\ell_1 + \ell_2)$ (i.e., the leaves of the recursion tree correspond to comparing an active highway and a highway path). Using Claim \ref{claim:blackbox_claim}, this can be done in time $O(\ell_1 + \ell_2 +\dfrag + D)$ as well. Note that there are $\log \ell_1$ levels of recursion---this is because, in each level of recursion, the number of active highways in $P_H^t$ becomes half of what it was before. Hence the total time taken is $O(\ell_1 + \ell_2 + \dfrag + D)\log \ell_1 = \tO(\ell_1 + \ell_2 + \dfrag + D)$.

\paragraph{Correctness.} To argue the correctness, we need to argue that when we break a problem of comparing two highway paths into two disjoint subproblems (as is done in every recursion step), it is enough to solve these two subproblems to get the minimum 2-respecting cuts. We show that this holds for the first recursion call---similar argument extends to all recursion calls. We know that the minimum 2-respecting cut which has one tree-edge in $\high^{\ell_1/2}$ has another tree-edge in $\way^{\ell_i}$. We just need to show, at this point, that the minimum 2-respecting cut, which has one tree-edge in $\high^t$ must have another edge in $\way^t \circ \way^{\ell_i}$, and the minimum 2-respecting cut, which has one tree-edge in $\high^b$ must have another edge in $\way^{\ell_i} \circ \way^b$ ($\circ$ denotes concatenation). This follows immediately from the monotonicity of minimum 2-respecting cut (Claim \ref{claim_monotonicity}). Because we compare $\way^{\ell_i}$ with $\high$ to find out the minimum 2-respecting cut which has one edge in $\way^{\ell_i}$ and another edge in $\high$, it is now enough to consider the two disjoint subproblems. Hence the correctness follows.
\end{proof}


\paragraph{Dealing with many pairs of highway paths.} Now we consider the set $\cP$ as discussed in Theorem \ref{thrm:pairing_up_highways}. A high level schematic algorithm is given as Algorithm \ref{alg:highway_algo}. There are going to be $\log n$ many levels of recursion as before: It is instructive to keep in mind the algorithm from Claim \ref{clm:two-long-highway}. We will run this algorithm for each pair from $R \subseteq \cP \times \cP$ (as in Theorem \ref{thrm:pairing_up_highways}) in parallel, each of which will have $O(\log n)$ many levels of recursion. In level $j$ of recursion, we will complete the corresponding levels in all these instantiations of the algorithm before we move on to the next iteration. Let us denote set of highway pairs that we solve for in level $j$ as $R^{(j)}$, and we start with $R^{(1)} = R$. From Theorem \ref{thrm:pairing_up_highways}, we know that every highway component takes part in $B_{path} \cdot \log n$ many pairs in $\cP$ where it is active. We will maintain this invariant in all $R^{(j)}$. 

In the first iteration, we compare all pairs of $R^{(1)}$ simultaneously: The broadcast computation required for these computations are pipelined over a BFS tree of $G$, and the aggregate computations inside each fragment due to its participation in $B_{path} \cdot \log n$ many such pairs from $R$ are pipelines inside the fragment. Let us now try to compute the round complexity of each iteration: We compute how many bits are broadcast in total and how many bits are aggregated inside any fragment in total. By Claim \ref{clm:two-long-highway}, each pair $\high$ and $\way$ requires time $(\ell_1 + \ell_2 + \dfrag + D)$ time of which $\tO(\ell_1 + \ell_2)$ bits are broadcast over $G$. Note that each highway inside a fragment takes part in $B_{path} \cdot \log n = O(\log^2 n)$ many pairs as active highway. So the total number of bits that are broadcast can be upper bounded by \[\sum_{(P_i,P_j) \in R} (\ell_i + \ell_j) = \sum_{P:\text{component}} (\text{\# of pairs from } R \text{ where $P$ is active} )= O(\nfrag \cdot \log^2 n),
\]
which can be broadcast in time $O(\nfrag \cdot \log^2 n + D)$. For the internal computation within active components, we know that the total amount of bits aggregated inside an active fragment is at most $\tO(\ell_1 + \ell_2 + \dfrag)$-bits. Hence, by a similar calculation as above, the number of bits to be aggregated inside a fragment $F$ in total over all pairs of highway paths where $F$ appears as an active fragment is $\tO(\nfrag + \dfrag)$ bits which can be done in time $\tO(\nfrag + \dfrag)$ in pipelined fashion. Hence the total round complexity of the first iteration is $\tO(\nfrag + \dfrag + D)$.


In the second iteration, we have the following situation: $R^{(2)}$ is derived from $R^{(1)}$ in the following way. Each pair $(\high, \way) \in R^{(1)}$ now gives rise to at most two pairs as in Claim \ref{clm:two-long-highway}, namely $(\high^t, \way^t)$ and $(\high^b, \way^b)$. These pairs are included in $R^{(2)}$. Note that each edge of highway paths in $R^{(1)}$ knows which highway paths in $R^{(2)}$ it participates in. Also, because this decomposition is disjoint, each component takes part in $\log^3 n$ many pairs as active component as before---this is the invariant we wanted to maintain. The algorithm for this iteration is similar to that of the first iteration, except this time we perform on pairs coming from $R^{(2)}$. Hence, by a similar calculation as before we see that this iteration can be completed in time $\tO(\nfrag + \dfrag + D)$.

This concludes the following: As the invariant is maintained in each level of recursion, each level can be performed in time $\tO(\nfrag + \dfrag + D)$. The number of levels of recursion is $O(\log n)$, and hence the total time required is $\tO(\nfrag + \dfrag + D)$ as well.

\begin{center}
  \centering
  \begin{minipage}[H]{0.8\textwidth}
\begin{algorithm}[H]
\caption{Schematic algorithm for super highways}\label{alg:highway_algo}
\Require A set $R$ of non-splittable pairs of super highways $(\high, \way)$ that every vertex $v$ knows about.
\begin{algorithmic}[1]
\Procedure{SuperHighwayCompare}{$R$}
\For{all pairs $(\high,\way) \in R$}
\State Let $\high$ has $\ell_1$ many active fragments, and $\way$ has $\ell_2$ many active fragments.
    \If{$\ell_1 = 1$ or $\ell_2 = 1$}
        \State Compare $\high$ with $\way$. \Comment{Algorithm \ref{alg:highway_basic_algo}.}
    \EndIf
\State Compare $\high^{\ell_1/2}$ with active fragments of $\way$. Let the respected edge in $\way$ in $\way^{\ell_i}$.
\Comment{Algorithm \ref{alg:highway_basic_algo}.}
\State Compare $\way^{\ell_i}$ with active fragments of $\high$.
\Comment{Algorithm \ref{alg:highway_basic_algo}.}
\State Remove $(\high, \way)$ from $R$ and add $(\high^t, \way^t)$ and $(\high^b,\way^b)$.
\EndFor
\State \textbf{Run} {\sc SuperHighwayCompare} on $R$.
\EndProcedure

\Procedure{SuperHighwaySelfCompare}{$R$}
\For{every $P_H \in R$}
    \If{$\ell = 1$}
        \State Find min 2-respecting cut when both tree edges are from $P_H$. 
        \Statex \Comment{Run algorithm from \Cref{lem:both-edge-short-high} on $P_H$.}
    \EndIf
    \State Initialize $R' = \emptyset$.
    \State $P_H$ has $\ell$ many fragments. Break $P_H^t = P_1, \cdots, P_{\floor{\ell/2}}$, and $P_H^b = P_{\floor{\ell/2}+1}, \cdots, P_\ell$.
    \State Include $P_H^t$ and $P_H^b$ in $R'$.
    \State \textbf{Run} \Call{SuperHighwayCompare}{$R'$}.
    \State Remove $P_H$ from $R$ and include $P_H^t$ and $P_H^b$.
\EndFor
\State \textbf{Run} {\sc SuperHighwaySelfCompare} on $R$.
\EndProcedure
\end{algorithmic}
\end{algorithm}
\end{minipage}
\end{center}

\subsubsection{Proof of Claim \ref{claim:blackbox_claim}}

The idea is to use Lemma 6.19 and 6.20 in parallel with highway partitioning. A high level schematic algorithm is provided in Algorithm \ref{alg:highway_basic_algo}. The readers are encouraged to notice the similarity of this proof to that of Section \ref{ssec:nh-high}. Unfortunately, we cannot show an analogous claim as that of Claim \ref{claim_cut_nh_h_no_edge} for highways---this will invariably increase the round complexity. Instead, we focus on only the highways inside a fragment and the highway paths in consideration. Let the \textit{active} highways in $P_H$ be $P_1, \cdots, P_\ell$. We first look at the case when there is no edge between $F_{P'}$ and $F_{P_i}$ for any $i \in [\ell]$.  The idea is to use Lemma \ref{lem:paths_h_h-noedge} instead of Lemma \ref{lemma_nh_no_edge}. We need to make sure we can compute all the necessary information needed to apply Lemma \ref{lem:paths_h_h-noedge}. Clearly, we can use Claim \ref{clm:two-short-highway} to know the value of $\extcov(P', P_i)$ which is a broadcast of $O(1)$ bits. If we do it for all pairs $(P',P_i)$ where $P_i$ is an active component of $P_H$, then the number of bits to be broadcast is $O(\ell)$ where $\ell$ is the number of active components of $P_H$, and hence requires $O(D + \ell)$ rounds. At this point, we will use Lemma \ref{lem:paths_h_h-noedge} for all pairs $(P',P_i)$. Again, each of these instantiation of Lemma \ref{lem:paths_h_h-noedge} requires $O(1)$ many aggregate computations inside $F_{P'}$ and $F_{P_i}$ and $O(1)$ bits of broadcast communication. By pipelining these computations for different $(P',P_i)$ pairs, we get round complexity of $O(D + \dfrag + \ell)$.



Once we are done with these computations, we turn to the case when there is an edge between $F_{P'}$ and $F_{P_i}$. Wlog assume all active highways in $P_H$ are such that there is an edge between $F_{P'}$ and $F_{P_i}$ (the case when there is no such edge has already been dealt with). Note that, by Claim \ref{claim_learn_edge_highway}, vertices in $F_P$ also know of one edge between $F_P$ and each component $F_{P_i}$. First we invoke \Cref{lemma_partitioning-highway} to partition $P$ in $E_1, \cdots, E_\ell$ corresponding to $P_1, \cdots, P_\ell$ which are different active highways of $P_H$. This can be done in time $O(\dfrag + \ell)$ where, at the end, all vertices in $T(P')$ know the sets $E_1, \cdots, E_\ell$---the computation is entirely inside $T(P')$ and requires $O(\ell)$ bits of aggregate and broadcast computation. At this point, it is sufficient to compare the pairs $(E_i, P_i)$ for all $i \in [\ell]$.

The idea is similar to that of Claim \ref{clm:short-nh-long-high-edge}, but we would like to replace the algorithm of Lemma \ref{lemma_paths_nh_h} with Lemma \ref{lemma_paths}. For this, we need to check whether we can satisfy the premise of Lemma \ref{lemma_paths}. This is also almost identical to that of Lemma \ref{lemma_paths_nh_h}---we simply have to use the claims for highway instead of non-highway. For completeness, we provide this argument here. All vertices in $T(P')$ should learn the values $\{\cov(e),\extcov(e,P_i)\}_{e \in E_i}$ for all $1 \leq i \leq \ell$. To do so, we first let all edges $e \in P'$ learn the values $\extcov(e,P)$ for all highways $P$, this takes $O(\dfrag + \nfrag)$ time using Claim \ref{claim_cov_eP}, and can be done in all highways simultaneously. Note that, after this step, the edges $e \in P_i$ also know the value $\extcov(e, P')$. Then, the information $\{\cov(e),\extcov(e,P_i)\}_{e \in E_i}$ is known to the edge $e$.  To let all vertices in $T(P')$ learn it we use pipelined upcast and broadcast (similar to Claim \ref{claim_broadcast}) within $T(P')$. As we have $\sum_{i=1}^k |E_i| = O(\dfrag + \ell)$---which is the number of pairs $\{\cov(e),\extcov(e,P_i)\}$, and hence the number of bits, that needs to be distributed inside $T_F(P)$---this takes $O(\dfrag + \ell)$ time. From Lemma \ref{lemma_partitioning-highway}, we also have that all vertices in $T(P')$ know the identity of all edges in the sets $E_i$. The only additional information that Lemma \ref{lemma_paths} requires in its premise is for the vertices of $T(P')$ and $F_{P_i}$ to know the value $\extcov(P', P_i)$. We use Claim \ref{clm:two-short-highway} where $O(
\ell)$ many aggregate and broadcast computations over a BFS tree on $G$ are required. This takes time $O(D + \ell)$. This information is known to every vertex of the graph. We next discuss the information known in $P_i$. First, using upcast and broadcast in the fragment $F_{P_i}$ of $P_i$, we can make sure that all vertices in the fragment know all the values $\{\cov(e)\}_{e \in P_i}$, they can also learn the identity of the edge $f$ between $T(P')$ and $F_{P_i}$, as follows. As vertices in $T(P')$ know the identity of $f$, then $f$ has an endpoint that knows about it, and can inform the second endpoint in $F_{P_i}$. Then, the information can be broadcast in $F_{P_i}$. This is only done if $P_i$ is potentially interested in $P'$, hence only $B_{path}$ times for $P'$ (from Theorem \ref{thrm:pairing_up_highways}; as we have seen in Claim \ref{clm:two-long-highway}, $P'$ is an active highway inside a highway path orthogonal to $P_H$). This shows that vertices in $T_F(P)$ and $F_{P_i}$ have all the information needed for applying Lemma \ref{lemma_paths}.

\begin{center}
  \centering
  \begin{minipage}[H]{0.8\textwidth}
\begin{algorithm}[H]
\caption{Schematic algorithm for a fragment highway and a super highway}\label{alg:highway_basic_algo}
\hspace*{\algorithmicindent} \textbf{Input:} A fragment highway $P$ and a super highway $P_H$ with active fragments $P_1, \cdots, P_\ell$.\\
 \hspace*{\algorithmicindent} \textbf{Output:} A minimum 2-respecting cut $C_P=\cut(e,e')$ where $e \in P$ and $e' \in P_H$.
 \algrenewcommand\algorithmicdo{\textbf{do parallelly}}
\begin{algorithmic}[1]
\For{every $i \in [\ell]$}
    \State Compute $\extcov(P, P_i)$. \Comment{Claim \ref{clm:two-short-highway}.}
\EndFor
\For{every $i \in [\ell]$ such that there is no edge between $F_P$ and $F_{P_i}$}
    \State Compare $(P, P_i)$. \Comment{Lemma \ref{lem:paths_h_h-noedge}.}
\EndFor
\State Let active fragments $P_1, \cdots, P_k$ are such that there is an edge between $F_P$ and $F_{P_i}$ for all such $P_i$.
\State Partition $P$ into $E_1, \cdots, E_k$ w.r.t.~$P_1, \cdots, P_k$.
\Statex \Comment{\Cref{lemma_partitioning-highway}.}
\For{every pair $(E_i, P_i)$}
    \State Let the edge between $F_P$ and $F_{P_i}$ be $f$.
    \State Use $f$ to route the information $\{e,\cov(e),\extcov(e,P_i)\}_{e \in E_i}$ from $P$ to $P_i$.
    \State The cut values is computed by the fragment highways $P_i$ (which is an active fragment). Compute the values $\{\cut(e,e')\}_{e \in P, e' \in P_i}$ for all edges $e \in E_i$ where the values $\{e,\cov(e),\extcov(e,P_i)\}_{e \in E_i}$ were received from $P$.
    \Statex\Comment{Lemma \ref{lemma_paths}.}
\EndFor
\State Communicate over a BFS tree to let all vertices learn the values $\{e,e',\cut(e,e')\}$ for edges $e,e'$ in the above cases that minimize $\cut(e,e').$
\end{algorithmic}
\end{algorithm}
\end{minipage}
\end{center}

\subsection{Both cut edges in same highway}
Note the if both cut edges are in the same highway and \textit{in the same fragment}, then it is already dealt with in \ref{lem:both-edge-short-high}. So we are only interested in the case when the cut edges are in different fragments of the same highway path. Recall that these highway paths are actually maximal highway paths in some layer. We will use Claim \ref{clm:two-long-highway} in a divide and conquer fashion---a similar technique was used in Algorithm 3.3 in \cite{mukhopadhyay2019weighted}. The idea is simple: We know how to efficiently compare two disjoint highway paths. Given a highway path $P_H$ consisting of (not necessarily active\footnote{We do not need to consider active highways here, because by construction every highway takes part in exactly one computation in each iteration of the algorithm contrary to $B_{path}$ many computations as is the case in comparing two highway paths.}) highways $P_1, \cdots, P_\ell$, we will employ a divide and conquer technique which will run for $O(\log \ell)$ rounds. In round $i$, we will work on the set of highways $\cP^{(i)}$ which we will define below. Initially, in the first round, $\cP^{(1)} = \cP$ (See The set discussed in Theorem \ref{thrm:pairing_up_highways}). We will also maintain the invariant that the highways in $\cP^{(i)}$ are disjoint. To start with, by construction, the highways of $\cP$ are disjoint.
\begin{description}
\item[--] In the first rounds, we will compare the highway composed of $P_1, \cdots, P_{\floor{\ell/2}}$ (denote it as $P_H^1$) with the highway composed of $P_{\floor{\ell/2}+1}, \cdots, P_\ell$ (denote it as $P_H^2$) for every $P_H \in \cP^{(1)}$. Using Claim \ref{clm:two-long-highway}, we can do it in $O(\ell + \dfrag + D)$ rounds (out of which $O(\ell + D)$ rounds are needed for broadcasting $O(\ell)$ bits of information and the rest of the computation is local). We will do it for all highways in the set $\cP$ in parallel. As the highways in $\cP$ are disjoint and the number of fragments is $\nfrag$, it is easy to see that this can be done in $O(\nfrag + \dfrag + D)$ rounds. At the end, all vertices know the tuples $\{e,e',\cut(e,e')\}$ for each pair $(P_H^1, P_H^2)$ where $e \in P_H^1$ and $e' \in P_H^2$ and $(e,e')$ minimizes such $\cut(e,e')$.

\item[--] In the second round, we construct the set $\cP^{(2)}$ by putting $P_H^1$ and $P_H^2$ of all highways $P_H \in \cP^{(1)}$. Note that, because the vertices know the set $\cP^{(1)}$, they can locally compute the set $\cP^{(2)}$. We follow the same procedure as in the first round, i.e., we divide the highway paths in $\cP^{(2)}$ and compare them. Note that the highways of $\cP^{(2)}$ are disjoint as well---this is the invariant we wanted to maintain. Hence this step can also be performed in $O(\nfrag + \dfrag + D)$ rounds.

\item[--] We continue this divide and conquer procedure for $O(\log n)$ steps until each $P_H \in \cP^{(i)}$ are left with only 2 components. This case can be solved in $O(\dfrag + D)$ time. As before, note that every such component takes part in exactly one comparison. Hence, pipelining the broadcast computation, this round can be completed in $O(\nfrag  + \dfrag + D)$ time.
\end{description}

Hence total round complexity of $\tO(\nfrag + \dfrag + D)$. At the end, all vertices know a pair $(e,e')$ for each $\cP^{(i)}$ and for each super highway $P_H \in \cP^{(i)}$ which minimized $\cut(e,e')$ for that super highway $P_H$ where $e \in P_H^t$ and $e' \in P_H^b$. The vertices can choose the minimum among them by local comparison.

\paragraph{Correctness.} We need to argue that if the minimum 2-respecting cut include edge $e$ from highway $P_i$ and edge $e'$ in highway $P_j$, then $P_i$ and $P_j$ are compared (possibly as a highway in a highway path) in one of the iterations of the divide and conquer algorithm. This follows from the following observation: At the $k$-th iteration, if the $k$-th significant bit of the binary representation of $i$ and $j$ are different, then they are compared. Hence the correctness follows.

\section{The min-cut algorithm}\label{sec:schematic_description}

This section provides the schematic algorithm for finding min-cut of a weighted graph in $\tO(\sqrt n + D)$ time which proves Theorem \ref{thm:main}. We first start with the schematic algorithm for minimum 2-respecting cut, proving Theorem \ref{thm:intro:2-respect}. In every sense, Algorithm \ref{alg:2-res-min-cut} is the heart of this work.

\subsection{A schematic algorithm for minimum 2-respecting cut}
\label{sec:algo-2-res}

In this section, we give a schematic algorithm for finding minimum 2-respecting cut in CONGEST model where, at the end of the algorithm, every vertex $v$ knows the following information:
\begin{enumerate}
    \item The value of the cut, \label{info:value-cut}
    \item The tree-edges (at most two) which the cut respects, and \label{info:tree-cut-edge}
    \item The edges incident to it that cross the cut. \label{info:cut-edge}
\end{enumerate}

 We next provide the schematic algorithm for minimum 2-respecting cut. Note that this is a high-level overview---the details of each step can be found in corresponding section mentioned in the comment. Also, at the end of the algorithm, every vertex knows (\ref{info:value-cut}) and (\ref{info:tree-cut-edge}). By applying Observation \ref{obs:observation_learn_cut}, it is immediate that the vertices will also know (\ref{info:cut-edge}).

\remove{
We start with the following observation.

\begin{observation}
If, at the end of the algorithm, every vertex $v$ knows (\ref{info:tree-cut-edge}) from above, then by local computation $v$ can find out (\ref{info:cut-edge}).
\end{observation}

This observation follows from LCA checks that $v$ can do (Claim \ref{claim_LCA_labels}). We show it when there are two tree edges $(e, e')$ that the cut respects (i.e., when the 2-respecting cut is an \textit{exact 2-respecting cut}). When the cut respects only one edge of the spanning tree, the computation is even simpler.

Consider any edge $f= (u,v)$ which is incident on $v$. Note that, given $(e,e')$, $v$ can do an LCA check to find out which edges among $e$ and $e'$ is covered by $f$. The edge $f$ takes part in the cut if $f$ covers \textit{exactly} one edge among $e$ and $e'$---this can be computed inside $v$ without any communication. Hence $v$ can infer (\ref{info:cut-edge}) from (\ref{info:tree-cut-edge}) itself by local computation.

Equipped with this observation, now we provide the schematic algorithm for minimum 2-respecting cut. Note that this is a high-level overview---the details of each step can be found in corresponding section mentioned in the comment. Also, at the end of the algorithm, every vertex knows (\ref{info:value-cut}) and (\ref{info:tree-cut-edge}). By applying the previous observation, it is immediate that the vertices will also know (\ref{info:cut-edge}).
}

\paragraph{Very high-level description of Algorithm \ref{alg:2-res-min-cut}.} The algorithm is divided into mainly four parts, each is labeled properly in the schematic description for easy reference. These are as follows:
\begin{description}

\item[Tree decompositions.] The first step is to perform fragment decomposition on the spanning tree $T$ followed by layering decompositions. See Section \ref{ssec:-decompFragment} and \ref{ssec:-layerDecomp} for details. The vertices also assign LCA labels to edges such that it is easy to find out whether a non-tree edge covers a tree edge. See Section \ref{sec:lca}.

\item[Minimum 1-respecting cut.] In the next step, the algorithm computes the minimum 1-respecting cut using Claim \ref{claim_learn_cov}.  
Also see Section \ref{sec:1-resp-cut}.

\item[Sampling.] The next step is do the sampling procedure to compute the set $\intpot{e}$, the set of paths which $e$ is \textit{potentially interested} in, for each tree edge $e$, where each such path is represented by the id of the lowest fragment that intersects the path (See Section \ref{ssec:-samplingProcedure} for a definition). This lets the vertices know the necessary information to perform the algorithm for the non-highway non-highway case discussed in Section \ref{ssec:finding_cut_non_highway}, and Algorithms \ref{alg:nh-h} and \ref{alg:highway_algo}. See Section \ref{sec:lemma_and_samp} for details regarding the implementation of this sampling procedure and the routing of necessary information across the graph.

\item[Minimum exact 2-respecting cut.] Finally, the vertices run the algorithm from section \ref{ssec:finding_cut_non_highway}, and  Algorithms \ref{alg:nh-h} and \ref{alg:highway_algo} one after the other. The vertices compare the minimum 2-respecting cut found in each algorithm. The vertices output the minimum among the minimum exact 2-respecting cut and the minimum 1-respecting cut.
\end{description}

Because each step here can be performed in time $\tO(\sqrt n + D)$ (See relevant theorem mentioned in the schematic description), the total time complexity of $\tO(\sqrt n +D)$. This proves Theorem \ref{thm:intro:2-respect}.

\begin{center}
  \centering
  \begin{minipage}[H]{0.9\textwidth}
\begin{algorithm}[H]
\caption{Schematic algorithm for distributed minimum 2-respecting cut}\label{alg:2-res-min-cut}
\hspace*{\algorithmicindent} \textbf{Input:} \begin{enumerate}
    \item  Weighted graph $G=(V,E,w)$, where every vertex $v \in V$ knows the set of incident edges on it along with their weights, \item  A spanning tree $T$ of $G$ where every vertex $v \in V$ knows the set of incident edges of $T$ on it.
\end{enumerate}
 \hspace*{\algorithmicindent} \textbf{Output:} Every vertex $v\in V$ knows the edges incident on it which take part in a minimum 2-respecting cut $C_T$ w.r.t. $T$, and the value of the cut.
\begin{algorithmic}[1]
\Statex \hrulefill
\begin{mdframed}[backgroundcolor=blue!20,skipabove=5pt,skipbelow=5pt]
      \centering
      {\em Tree decompositions. (Section \ref{sec:building_blocks})}
    \end{mdframed}
\State Perform a \textit{fragment decomposition} with parameters $\nfrag=\dfrag=\sfrag=O(\sqrt{n})$. At the end, each vertex $v$ knows the information detailed in Lemma \ref{lemma:strong_decomp_construct}.
\Statex \Comment{\textcolor{blue}{See Section \ref{ssec:-decompFragment}.}}
\State Perform a \textit{layering decomposition} on the highways as in Lemma \ref{lemma:layer_decomposition_highway} and on the non-highways as in \Cref{lemma:layer_decomposition_inside_fragment_nonhighway}.
\Comment{\textcolor{blue}{See Sections \ref{ssec:-layerDecomp} and \ref{ssec:way_todo_layering}.}} 

\begin{mdframed}[backgroundcolor=blue!20,skipabove=5pt,skipbelow=5pt]
      \centering
      {\em Computing the minimum 1-respecting cut.}
    \end{mdframed}
\State Find out the minimum 1-respecting cut. \label{algline:1-res}
\Comment{\textcolor{blue}{See Claim \ref{claim_learn_cov} and Section \ref{sec:1-resp-cut}}} 

\begin{mdframed}[backgroundcolor=blue!20,skipabove=5pt,skipbelow=5pt]
      \centering
      {\em Sampling and routing. (Section \ref{sec:lemma_and_samp})}
    \end{mdframed}
\State Each edge $e\in T$ finds a set of potentially interesting paths $\intpot{e}$.
\Statex \Comment{\textcolor{blue}{See Lemma \ref{lemma:find-int-samp}, and Lemma \ref{lemma:parse_paths}.}}
\State Each non-highway bough $P'$ in layer $i$ routes information to $T(P')$ which consists of all ids of fragments $F$ that contain a non-highway path which is in $\intpot{P'}$.
\Statex \Comment{\textcolor{blue}{See Lemma \ref{lemma:parse_paths} and Corollary \ref{corol:-layering_interested}.}}
\State Each non highway bough $P$ routes the relevant information about super highways in $\intpot{P}$ which are completely above or completely orthogonal to the fragment of $P$.
\Statex\Comment{\textcolor{blue}{See Lemma \ref{lemma:parse_paths}, and Corollary \ref{corol:-_layer_highway}}.}
\State Each fragment highway $P$, routes the relevant information to $F_P$ about ids of fragments $F$ that contain non-highway  paths in $\intpot{P}$, not including $F_P$. 
\Statex\Comment{\textcolor{blue}{See Lemma \ref{lemma:parse_paths} and Corollary \ref{corol:-highway_non_highway}}}
\State Each vertex $v$ learns the set $R$ of pairs of super highways potentially interested in one another. 
\Statex\Comment{\textcolor{blue}{See Theorem \ref{thrm:pairing_up_highways}.}} 

\begin{mdframed}[backgroundcolor=blue!20,skipabove=5pt,skipbelow=5pt]
      \centering
      {\em Computing the minimum exact 2-respecting cut. (Section \ref{sec:2_respecting})}
    \end{mdframed}
\State Run the algorithm from Claim \ref{claim_nh_2respecting}.  \hfill \textbf{Record} the minimum cut. \label{algline:nh-nh}
\State Run Algorithm \ref{alg:nh-h}. \hfill\textbf{Record} the minimum cut. \label{algline:nh-h}
\State Run Algorithm \ref{alg:highway_algo} on $R$. \hfill\textbf{Record} the minimum cut. \label{algline:h-h} 

\State Output the minimum cut among what is recorded in Line \ref{algline:1-res}, \ref{algline:nh-nh}, \ref{algline:nh-h} and \ref{algline:h-h}.
\end{algorithmic}
\end{algorithm}
\end{minipage}
\end{center}

\subsection{The min-cut algorithm for weighted graphs}

Now we are ready to give a schematic description of minimum cut on a weighted graph. Note that this algorithm calls Algorithm \ref{alg:2-res-min-cut} as a subroutine. As discussed before, at the end of Algorithm \ref{alg:2-res-min-cut}, every vertex knows the value of a minimum 2-respecting cut w.r.t. a spanning tree $T$ along with the edges participating in the cut which are incident on it. 
\begin{center}
  \centering
  \begin{minipage}[H]{0.8\textwidth}
\begin{algorithm}[H]
\caption{Schematic algorithm for distributed min-cut}\label{alg:min-cut}
\hspace*{\algorithmicindent} \textbf{Input:} Weighted graph $G=(V,E,w)$, where every vertex $v \in V$ knows the set of incident edges on it along with their weights. \\
 \hspace*{\algorithmicindent} \textbf{Output:} Every vertex $v\in V$ knows the edges incident on it which take part in the minimum cut along with the value of that min-cut. \\
\begin{algorithmic}[1]
\State Perform a greedy tree packing on $G$ to obtain $\cT=\set{T_1,...,T_k}$, where $k = O(\log^{2.1} n)$.
\Statex \Comment{See Theorem \ref{thm:reduction_mincut_to_respect}}
\For{each $T \in \cT$} 
    \State Perform minimum 2-respecting cut algorithm w.r.t. $T$ (Algorithm \ref{alg:2-res-min-cut}). Let the cut obtained be $C_T$. \Comment{See Theorem \ref{thm:intro:2-respect}}
    \State Every vertex $v \in V$ knows the value of $C_T$ and the edges incident on $v$ which take part in $C_T$.
\EndFor \label{algline:1-resp-e}
\State Every vertex $v \in V$ chooses the $C_T$ which has minimum total weight.
\end{algorithmic}
\end{algorithm}
\end{minipage}
\end{center}

The greedy tree-packing can be performed in time $\tO(\sqrt n +D)$ as mentioned in Theorem \ref{thm:reduction_mincut_to_respect}. We also know that Algorithm \ref{alg:2-res-min-cut} takes time $\tO(\sqrt n + D)$. Hence, computing Algorithm \ref{alg:2-res-min-cut} for each $T \in \cal T$ takes time $|\cT| \times \tO(\sqrt n + D) = \tO(\sqrt n + D)$. Hence the total time complexity of Algorithm \ref{alg:min-cut} is $\tO(\sqrt n + D)$. This proves Theorem \ref{thm:main}.


\section*{Acknowledgment}

We would like to thank Keren Censor-Hillel for many valuable discussions.
This project has received funding from the European Research Council (ERC) under the European
Unions Horizon 2020 research and innovation programme under grant agreement No 715672 and 755839. Danupon
Nanongkai and Sagnik Mukhopadhyay are also partially supported by the Swedish Research Council (Reg.~No.~2015-04659 and 2019-05622). Michal Dory and Yuval Efron are supported in part by the Israel Science Foundation (grant no.~1696/14).

\bibliography{biblio,references}

\newcommand{\etalchar}[1]{$^{#1}$}
\begin{thebibliography}{GKK{\etalchar{+}}15}

\bibitem[AGKR04]{alstrup2004nearest}
Stephen Alstrup, Cyril Gavoille, Haim Kaplan, and Theis Rauhe.
\newblock Nearest common ancestors: A survey and a new algorithm for a
  distributed environment.
\newblock {\em Theory of Computing Systems}, 37(3):441--456, 2004.

\bibitem[BKKL17]{BeckerKKL16}
Ruben Becker, Andreas Karrenbauer, Sebastian Krinninger, and Christoph Lenzen.
\newblock Near-optimal approximate shortest paths and transshipment in
  distributed and streaming models.
\newblock In {\em {DISC}}, volume~91, pages 7:1--7:16, 2017.

\bibitem[CHD19]{censor2019fast}
Keren Censor-Hillel and Michal Dory.
\newblock Fast distributed approximation for tap and 2-edge-connectivity.
\newblock {\em Distributed Computing}, pages 1--24, 2019.

\bibitem[DG19]{DBLP:conf/podc/DoryG19}
Michal Dory and Mohsen Ghaffari.
\newblock Improved distributed approximations for minimum-weight
  two-edge-connected spanning subgraph.
\newblock In {\em {PODC}}, pages 521--530, 2019.

\bibitem[DHK{\etalchar{+}}11]{Das-Sharma}
Atish {Das Sarma}, Stephan Holzer, Liah Kor, Amos Korman, Danupon Nanongkai,
  Gopal Pandurangan, David Peleg, and Roger Wattenhofer.
\newblock Distributed verification and hardness of distributed approximation.
\newblock In {\em {STOC}}, pages 363--372, 2011.

\bibitem[DHNS19]{DagaHNS19}
Mohit Daga, Monika Henzinger, Danupon Nanongkai, and Thatchaphol Saranurak.
\newblock Distributed edge connectivity in sublinear time.
\newblock In {\em {STOC}}, pages 343--354. {ACM}, 2019.

\bibitem[Dor18]{Dory18}
Michal Dory.
\newblock Distributed approximation of minimum k-edge-connected spanning
  subgraphs.
\newblock In {\em {PODC}}, pages 149--158. {ACM}, 2018.

\bibitem[Dor20]{dory2020}
Michal Dory.
\newblock {\em Distributed Network Design}.
\newblock PhD thesis, Technion, 2020.
\newblock
  \url{http://www.cs.technion.ac.il/users/wwwb/cgi-bin/tr-info.cgi/2020/PHD/PHD-2020-07}.

\bibitem[EFS56]{EliasFS56}
Peter Elias, Amiel Feinstein, and Claude~E. Shannon.
\newblock A note on the maximum flow through a network.
\newblock {\em IRE Trans. Information Theory}, 2(4):117--119, 1956.

\bibitem[EKNP14]{ElkinKNP14}
Michael Elkin, Hartmut Klauck, Danupon Nanongkai, and Gopal Pandurangan.
\newblock Can quantum communication speed up distributed computation?
\newblock In {\em {PODC}}, pages 166--175. {ACM}, 2014.

\bibitem[Elk06]{Elkin06}
Michael Elkin.
\newblock An unconditional lower bound on the time-approximation trade-off for
  the distributed minimum spanning tree problem.
\newblock {\em {SIAM} J. Comput.}, 36(2):433--456, 2006.

\bibitem[FF87]{FordF87}
L.~R. Ford and D.~R. Fulkerson.
\newblock {\em Maximal Flow Through a Network}, pages 243--248.
\newblock Birkh{\"a}user Boston, Boston, MA, 1987.

\bibitem[GG18]{geissmann2018parallel}
Barbara Geissmann and Lukas Gianinazzi.
\newblock Parallel minimum cuts in near-linear work and low depth.
\newblock In {\em SPAA}, pages 1--11, 2018.

\bibitem[GH16]{GhaffariH16}
Mohsen Ghaffari and Bernhard Haeupler.
\newblock Distributed algorithms for planar networks {II:} low-congestion
  shortcuts, mst, and min-cut.
\newblock In {\em {SODA}}, pages 202--219. {SIAM}, 2016.

\bibitem[GK13]{Ghaffari-Kuhn}
Mohsen Ghaffari and Fabian Kuhn.
\newblock Distributed minimum cut approximation.
\newblock In {\em Proceedings of the 27th {DISC}}, pages 1--15, 2013.

\bibitem[GKK{\etalchar{+}}15]{GhaffariKKLP15}
Mohsen Ghaffari, Andreas Karrenbauer, Fabian Kuhn, Christoph Lenzen, and Boaz
  Patt{-}Shamir.
\newblock Near-optimal distributed maximum flow: Extended abstract.
\newblock In {\em {PODC}}, pages 81--90, 2015.

\bibitem[GMW20]{GawrychowskiMW19}
Pawel Gawrychowski, Shay Mozes, and Oren Weimann.
\newblock Minimum cut in o(m log{\({^2}\)} n) time.
\newblock In {\em {ICALP} 2020}, pages 57:1--57:15, 2020.

\bibitem[GN18]{GhaffariN18}
Mohsen Ghaffari and Krzysztof Nowicki.
\newblock Congested clique algorithms for the minimum cut problem.
\newblock In {\em {PODC}}, pages 357--366. {ACM}, 2018.

\bibitem[GNT20]{Ghaffari0T20}
Mohsen Ghaffari, Krzysztof Nowicki, and Mikkel Thorup.
\newblock Faster algorithms for edge connectivity via random 2-out
  contractions.
\newblock In {\em {SODA}}, pages 1260--1279. {SIAM}, 2020.

\bibitem[GP16]{ghaffari2016near}
Mohsen Ghaffari and Merav Parter.
\newblock Near-optimal distributed algorithms for fault-tolerant tree
  structures.
\newblock In {\em SPAA}, pages 387--396, 2016.

\bibitem[HKN16]{HenzingerKN-STOC16}
Monika Henzinger, Sebastian Krinninger, and Danupon Nanongkai.
\newblock A deterministic almost-tight distributed algorithm for approximating
  single-source shortest paths.
\newblock In {\em {STOC}}, pages 489--498, 2016.

\bibitem[Kar99]{karger1999random}
David~R Karger.
\newblock Random sampling in cut, flow, and network design problems.
\newblock {\em Mathematics of Operations Research}, 24(2):383--413, 1999.

\bibitem[Kar00]{Karger00}
David~R. Karger.
\newblock Minimum cuts in near-linear time.
\newblock {\em J. {ACM}}, 47(1):46--76, 2000.

\bibitem[KKP13]{KorKP13}
Liah Kor, Amos Korman, and David Peleg.
\newblock Tight bounds for distributed minimum-weight spanning tree
  verification.
\newblock {\em Theory Comput. Syst.}, 53(2):318--340, 2013.

\bibitem[KP98]{KuttenP98}
Shay Kutten and David Peleg.
\newblock Fast distributed construction of small $k$-dominating sets and
  applications.
\newblock {\em Journal of Algorithms}, 28(1):40--66, 1998.
\newblock Announced at PODC'95.

\bibitem[MN20]{mukhopadhyay2019weighted}
Sagnik Mukhopadhyay and Danupon Nanongkai.
\newblock Weighted min-cut: Sequential, cut-query and streaming algorithms.
\newblock In {\em STOC}, 2020.

\bibitem[Nan14]{Nanongkai-STOC14}
Danupon Nanongkai.
\newblock Distributed approximation algorithms for weighted shortest paths.
\newblock In {\em Symposium on Theory of Computing (STOC)}, pages 565--573,
  2014.

\bibitem[NS14]{Nanongkai-Su}
Danupon Nanongkai and Hsin{-}Hao Su.
\newblock Almost-tight distributed minimum cut algorithms.
\newblock In {\em {DISC}}, pages 439--453, 2014.

\bibitem[Par19]{Parter19-smallcut}
Merav Parter.
\newblock Small cuts and connectivity certificates: A fault tolerant approach.
\newblock 2019.

\bibitem[PR00]{PelegR00}
David Peleg and Vitaly Rubinovich.
\newblock A near-tight lower bound on the time complexity of distributed
  minimum-weight spanning tree construction.
\newblock {\em {SIAM} J. Comput.}, 30(5):1427--1442, 2000.

\bibitem[PT11]{PritchardT11}
David Pritchard and Ramakrishna Thurimella.
\newblock Fast computation of small cuts via cycle space sampling.
\newblock {\em {ACM} Trans. Algorithms}, 7(4):46:1--46:30, 2011.

\bibitem[Tho07]{Thorup07}
Mikkel Thorup.
\newblock Fully-dynamic min-cut.
\newblock {\em Combinatorica}, 27(1):91--127, 2007.
\newblock Announced at STOC'01.

\end{thebibliography}
\appendix
\section{Reduction to 2-respecting cut}\label{appen:redutction_to_respect}
In this section, we highlight key ideas from the proof of Theorem \ref{thm:reduction_mincut_to_respect}, which is proved in \cite{DagaHNS19}. In this section, we treat edges of weight $w$ in the graph $G$ as $w$ multi-edges. We begin some definitions and notations.

\begin{definition}\label{def:greedy_tree_packing}
Let $\cT$ be some set of spanning trees of a given multi-graph $G$. Denote the load of an edge $e$ by $L_{\cT}(e)=|\set{T\in \cT\mid e\in T}|$. Furthermore, a set $\cT=\set{T_1,...,,T_k}$ of spanning trees is called a greedy tree packing if for all $1\leq i\leq k$ it holds that $T_i$ is a minimum spanning tree with the respect to the load given by $L_{\cT_{i-1}}(e)$. Here $\cT_{i-1}=\set{T_1,...,T_{i-1}}$. 
\end{definition}

Next, we state the following known results. 
\begin{lemma}[\cite{Thorup07}]\label{lemma:lemma_tree_packing1}
Let $C$ be any cut of a multi-graph $G$ with at most $1.1 \cdot OPT$ many edges and $\mathcal T$ be a greedy tree packing with $OPT \cdot \ln m$ many trees. Then $C$ 2-respects at least $1/3$ fraction of trees in $\mathcal T$.
\end{lemma}

Note that thus if one implements a greedy tree packing for sufficiently many iterations, one obtains a greedy tree packing $\cT$ such that $\frac{1}{3}$ fraction of the trees in $\cT$ 2-respect the minimum cut in $G$. However, one needs to obtain a greedy tree packing with $\omega(OPT\log n)$ trees. This can be problematic since it might be the case that $OPT=\Omega(n)$. 
To circumvent this obstacle, we employ a sampling idea from \cite{karger1999random} which will reduce the value of $OPT$ in the graph we consider to be $O(\log n)$, while preserving the minimum cut.

\begin{lemma}[\cite{karger1999random}]\label{lemma:reduction_sampling_appendix}
Let $0<p<1$ and let $H=G_p$ be a random subgraph of $G$ resulting from keeping each edge of $G$ with probability $p$, and removing it with probability $1-p$. Let $OPT_H$ be the value of the min cut in $H$. If $p\cdot OPT=\omega(\log n)$, then it holds w.h.p. that $OPT_H=(1\pm o(1))p\cdot OPT$. Moreover, w.h.p., min cuts of $G$ are near-minimal in $H$ and vice versa, in the following sense. A min cut $C$ of $G$ has $(1\pm o(1))OPT_H$ edges crossing it in $H$, and a min cut $C_H$ in $H$ has $(1\pm o(1))OPT$ edges crossing it in $G$.
\end{lemma}

We restate Theorem \ref{thm:reduction_mincut_to_respect} here for quick reference,  a proof of which appears in \cite{DagaHNS19}. We include it here for completeness.

\begin{theorem*}[Theorem \ref{thm:reduction_mincut_to_respect} restated]
Given a weighted graph $G$, in $\tilde{O}(\sqrt{n}+D)$ rounds, we can find a set of spanning trees ${\cal T}=\set{T_1,\cdots,T_k}$ for some $k=\Theta(\log ^{2.2} n)$ such that w.h.p. there exists a min-cut of $G$ which 2-respects at least one spanning tree $T\in \cal T$. Also, each node $v$ knows which edges incident to it are part of the spanning tree $T_i$, for $1\leq i\leq k$.
\end{theorem*}

\begin{proof}[Proof of Theorem \ref{thm:reduction_mincut_to_respect}]
To prove Theorem \ref{thm:reduction_mincut_to_respect} using the Lemma \ref{lemma:reduction_sampling_appendix}, we would like to have $p\cdot OPT=\Theta(\log ^{1.1} n)$, so that we can construct a greedy tree packing in $G$ with $\Theta(\log ^{2.1} n)$ trees. A-priori we do not know the value of $OPT$, hence we do not know the sampling probability $p$ as well. We solve it by doing the following: We run a $(1 + \eps)$-algorithm for min-cut by Nanongkai and Su \cite{Nanongkai-Su} in $\tO(D + \sqrt n)$ time at the end of which every vertex knows a $(1 + \eps)$ approximation of the min-cut. Let us call this value as $\widetilde{OPT}$. Next we set $\tilde p = 2 \cdot O(\log^{1.1} n/ \widetilde{OPT})$. It is easy to see from the approximation guarantee of $\widetilde{OPT}$ that $p \leq \tilde p \leq 2p$. We use $\tilde p$ as the sampling probability.

As a side-note, notice that the quantity $\tilde p = 2 \cdot O(\log^{1.1} n/ \widetilde{OPT})$ needs to be at most 1 which implies that the following sampling process works only when $\widetilde{OPT} = \Omega(\log^{1.1} n)$. When $\widetilde{OPT} = o(\log^{1.1} n)$, we skip the sampling altogether and proceed to tree packing.

The sampling procedure itself is fairly simple. As mentioned before, the weighted edges are treated as multi-edges and the idea is to sample the unweighted edges uniformly and independently with probability $\tilde p$. To this end, given a weighted edge $e = \{u, v\}$ with weight $w(e)$, where $u \prec v$ in some global ordering of the vertices that they agree up on, $u$ samples unweighted edges uniformly and independently from $w(e)$ many unweighted edges corresponding to $e$, and sends across this number of sampled edges to $v$. This does not cause congestion as the number of bits needed to be transferred to $v$ is $O(\log n)$. Once this is done, the pair $(u,v)$ can order the sampled edges by themselves using some global ordering.

 Once the sampling is done, all that is left is constructing a greedy tree packing in $G$ (viewed as a multi-graph) with $\Theta(\log ^{2.1} n)$ trees. This can be done in $\tilde{O}(\sqrt{n}+D)$ rounds since constructing an MST in a given graph has complexity $\tilde{O}(\sqrt{n}+D)$ \cite{KuttenP98}. To start with, each vertex assumes all edges incident on it has load 0. For each MST computation, each vertex pair $(u,v)$ consider only the edge between them which has the smallest load---if there only one edge between them, then they consider that edge. This does not affect the computation because an MST will always include the edge with the smallest load from a set of parallel edges. If there are multiple edges between $u$ and $v$ with smallest load, they break tie arbitrarily. After each MST computation, $u$ and $v$ increase the load of this edge depending on whether this edge in included in the MST or not. An important point to note here is that the sampled edges between $u$ and $v$ can be distinguished by $u$ and $v$ by their load; if two such edges have same load, then they can be used interchangeably and it does not cause a problem in the execution of this algorithm. Hence, the vertices $u$ and $v$ can be in sync in all executions of the MST algorithm.
 \end{proof}

\section{Missing proofs from Section \ref{sec:building_blocks}}\label{sec:appen_building_blocks}

\subsection{Fragment decomposition}

\FragDecomp*

\begin{proof}

The proof is divided into two main steps: Decomposing $T$ into $O(\sqrt{n})$ components of size $O(\sqrt{n})$, and modifying these components to satisfy properties 2-4 depicted in the beginning of Section \ref{ssec:-decompFragment}.

\paragraph{Breaking down $T$.} Our goal here is breaking down $T$ into connected components $S_1,...,S_k$ with the following properties.
\begin{enumerate}
    \item $k=O(\sqrt{n})$.
    \item For all $i\in [k]$, $|S_i|=O(\sqrt{n})$.
    \item For all $i,j\in [k], i\neq j$, it holds that $|S_i\cap S_j|\leq 1$. Furthermore, if $|S_i\cap S_j|=1$, then $v\in S_i\cap S_j$ is the root of both of these connected components. In particular, all components are edge disjoint.
\end{enumerate}

We start with running the MST algorithm of Kutten and Peleg \cite{KuttenP98}. This algorithm decomposes $T$ into $O(\sqrt{n})$ vertex disjoint components of diameter $O(\sqrt{n})$ in $\tO(\sqrt{n}+D)$ rounds. Denote these components by $C_1,...,C_m$. 

We now describe a procedure that breaks down these components further into components of \emph{size} at most $O(\sqrt{n})$, while maintaining the number of components at most $O(\sqrt{n})$ as well.

Let $C_i,i\in [m]$ be some connected component, and denote the highest (closest to the root) vertex in $C_i$ by $r$, and denote by $L_i$ the set of lowest vertices in $C_i$, i.e. all vertices in $C_i$ whose all descendants are not in $C_i$.

Using standard aggregation on $T\cap C_i$, the vertices in $C_i$ can learn $|C_i|$. If $C_i<10\sqrt{n}$, then this component is marked \emph{good} and does not perform the rest of algorithm.

From here on in we assume that $|C_i|>10\sqrt{n}$.
Denote the parent vertex in $T$ of a given vertex $v$ by $p(v)$.
Denote by $D_i(v)$ the set of descendants of $v$ in $T$ which are in $C_i$.
Denote by $T_i(v)$ the tree $T(v)\cap C_i$ where $T(v)$ is the subtree of $T$ rooted at $v$.
We now perform an aggregate computation that starts at each vertex of $L_i$ and ends in $r$. Each vertex $v\in C_i$ holds a variable $count(v)$ initialized to zero and each vertex $v$ performs the following. 
Conceptually, we do a \emph{single} scan of the tree from the leaves to the root and a \emph{single} scan of the tree from the root to the leaves. The first involves an aggregate computation of the number of descendants in the subtree rooted at each given vertex, the latter assigns vertices to their respective new components.

\begin{enumerate}
    \item If $v\in L_i$, $v$ sends $1$ to $p(v)$
    \item $v$ waits to receive a message $m_u$ from each immediate descendant $u\in D_i(v)$, then $v$ sets $counter(v)=\sum_{u\in D_i} m_u+1$.
    \item If $\sqrt{n}\leq counter(v)\leq 5\sqrt{n}$, $v$ creates a new component 
    composing of vertices in $T_i(v)$ who were not yet assigned to any other new component in the following way. $v$ broadcasts its own id as the id of the new component down the tree $T_i(v)$ and tells $p(v)$ that $v$ started a new component as well and $0$ as the message $m_v$. Furthermore in this case, if $v$ receives from $p(v)$ a message that it was added to a new component (i.e., the id of some ancestor vertex of $v$), it ignores it and does not forward it down the tree. Any vertex in $T_i(v)$ that received the id of $v$ and wasn't yet assigned to any other component joins the component of $v$ whose id is the id of $v$.
    \item If $counter(v)>5\sqrt{n}$,  $v$ does the following. $v$ goes over the messages $m_u$ it received from each of its immediate descendants $u$ in some order and starts summing up their messages $m_u$ up to the point where $z$ satisfies $\sqrt{n}\leq z\leq 5\sqrt{n}$. Denote the set of descendants whose message were summed up to this point by $A$.
     $A$ now defines a new component (which is created in the same manner as the previous bullet point) which contains the trees $T_i(u)$ for each $u\in A$, except the vertices who were already assigned to a new component previously, $v$ is included in the new component as well. Then $v$ initializes $z$ to $0$ and resumes this component from the next vertex in the initial order. $v$ continues creating components in such a manner, giving the id $(v,i)$ to the $i$-th component created, until all of its descendants are in new components or until the last iteration, in which the summation of messages $m_u$ of the remaining descendants did not exceed $\sqrt{n}$. If the latter is the case, denote these remaining descendants by $A^*$, denote by $j$ the current iteration of the described procedure in this bullet. All vertices in $T_i(u),u\in A^*$ who were not yet assigned to a component are assigned to the component whose id is $(v,j-1)$. 
    
Note that $v$ is  being included in all of these components created in the procedure above, and is the only vertex in the intersection of any two of these new components..
    \item If $counter(v)<\sqrt{n}$, $v$ sends $counter(v)$ as the message $m_v$ to its parent $p(v)$.
    \item If $v=r$ and satisfies the previous bullet point, then $r$ starts a new component in the same manner as described in bullet point 3,  with all the vertices that were not assigned yet a component. This is done by broadcasting the id of that component, which is $r$ down the tree until this forwarding halts at vertices that were already assigned to other new components.
\end{enumerate}

Note that this procedure an be done in parallel in each of the original connected components since they are vertex disjoint.
 
 Since the diameter of each component in $O(\sqrt{n})$, and an MST can be computed in $\Tilde{\Theta}(\sqrt{n}+D)$ rounds, the round complexity of this procedure is $\Tilde{\Theta}(\sqrt{n}+D)$.
 
 For correctness, note that each newly created component, apart from the one rooted at the root of the component $r$, has size at least $\sqrt{n}$ and at most $O(\sqrt{n})$. Thus there are at most $2\sqrt{n}$ new components (at most an additional $\sqrt{n}$ smaller new components rooted at the roots of components). Furthermore, there are at most $O(\sqrt{n})$ good components, thus in total there at most $O(\sqrt{n})$ components.
 
The property that for all $i\in [k]$, $|S_i|=O(\sqrt{n})$ follows from the behavior of the algorithm, since every new component has size at most $O(\sqrt{n})$.

The property that for all $i,j\in [k],i\neq$ it holds that $|S_i\cap S_j|\leq 1$, and if $|S_i\cap S_j|=1$ then $v\in S_i\cap S_j$ is the root of both components holds by the explanations in bullet point 4 above.
 
 Denote the newly created components by $S_1,...,S_k$.
 Note that each node can identify its component by the ID that was broadcast to it from the root of the component, using an additional $O(\sqrt{n})$ rounds for broadcasting that ID down the tree. Note that this part of the procedure in implemented in parallel as all components are edge disjoint
 
 \paragraph{Modifying and creating structure.} Now, our goal is to modify $S_1,...,S_k$ to have properties 1-4 in the statement of the lemma depicted in the start of this appendix section, while maintaining the fact that the amount of components and the size of each component is $O(\sqrt{n})$.
 
 This part of the algorithm follows the lines of \cite{ghaffari2016near}, in which a similar decomposition is constructed, but with diameter instead of size guarantees, as explained in the beginning of the section. Later works such as \cite{dory2020} employed this procedure as well and proved additional claims regarding its final properties, which we use later.
 
 \paragraph{Marking vertices.}
Now, with $S_1,...,S_k$, we do the following. Let $A\subseteq V$ be the set of vertices which are in more than a single component. Clearly $|A|=O(\sqrt{n})$ since each two components can intersect at at most one vertex, and there are $O(\sqrt{n})$ components. Furthermore, let $B\subseteq E(T)$ be the set of tree edges connecting vertices which are in different components and neither of them are on more than one component. Again, $|B|=O(\sqrt{n})$ since $k=O(\sqrt{n})$. Thus, using a BFS tree, all the vertices in the graph can learn the sets $A$ and $B$.

We now proceed to describe the marking procedure. First of all, all vertices in $A$ are marked, all vertices that have an incident edge in $B$ are marked, and the root $r$ of the tree $T$ is marked. Furthermore, for each two marked vertices $v,u$ we mark $\lca{v,u}$ in the following way. We scan each component from the leaves to the root of the component, a leaf $v$ sends to his parent its ID if it's marked or $\emptyset$ otherwise. An internal vertex $v$ waits to receive messages from all its descendants, if it received exactly one ID, it forwards it to $p(v)$. If it received 2 IDs, it marks itself and forwards one of the IDs arbitrarily to $p(v)$. Otherwise, it sends $\emptyset$ to $p(v)$.

Clearly this procedure takes $O(\sqrt{n})$ rounds since all communication is done internally in the components and there is no congestion since components are edge disjoint.

In \cite[Lemma 2.7]{dory2020} the following claim is proven about this marking procedure. The properties of the components are a bit different, but the proof is exactly the same.
\begin{claim}\label{claim_appen:marked_vertices}
The set of marked vertices satisfies the following properties.
\begin{enumerate}
    \item The root $r$ is marked, and  each other vertex $v$ has a marked ancestor of distance at most $O(\sqrt{n})$.
    \item For each two marked vertices $v,u$, $\lca{v,u}$ is also marked.
    \item There are $O(\sqrt{n})$ marked vertices.
\end{enumerate}
\end{claim}

Lastly, using an aggregate computation over the tree inside of each component, each component makes sure that there is at least a single leaf of that component that is marked. If there was previously a marked leaf, the component does nothing. Otherwise, the leaf with the lowest id is marked.
Note that this still maintains the number of marked vertices at $O(\sqrt{n})$ since there are $O(\sqrt{n})$ components.

\paragraph{Defining fragments.} Lastly, we define the final fragments and highway paths in our final decomposition. For each marked vertex $d_F\neq r$, the path from $d_F$ to its closest marked ancestor ($r_F$) defines the highway path of the fragment $F$. Note that by definition no vertex on the path between $r_F$ and $d_F$ is marked, and no other descendant (except $d_F$) of any internal vertex on that path is marked as well, since otherwise, the vertex itself is marked as well. The fragment $F$ thus includes the path between $r_F$ and $d_F$, and all the descendants (not in the direction of $d_F$) of the internal vertices on that path. Now, if $v$ is a marked vertex, denote by $NM(v)$ all of $v$'s descendants that are not included yet in any fragment. In other words, these are exactly the descendants of $v$ that have no marked descendants themselves. If $v$ is already the root of some fragment $F$, and $v$ is \emph{not} the root of any component, then all of $NM(v)$ are added to $F$ as non highway vertices and edges. If $v$ is a root of both a fragment and a component, remember that $v$ can be the root of several components. Note that for each such component, $S$, there is already a fragment $F_S$ rooted at $v$ which is contained in $S$. This is true since each component, including $S$, has at least one marked leaf. Thus we can take the set $NM(v)$ and partition its elements (vertices) between the fragments rooted at $v$ such that each fragment is completely contained inside some component.

\paragraph{Proof of properties.}
We now turn to proving the listed properties in the statement of the lemma about the resulting decomposition. We restate these properties for the sake of clarity.

\begin{enumerate}
    \item Each tuple $t_F=(r_F, d_F)$ represents an edge-disjoint fragment (subtree) $F$ of $T$ rooted at $r_F$ with diameter $\dfrag = O(\sqrt{n})$ and size  $\sfrag = O(\sqrt{n})$. The vertex $r_F$ is an ancestor of all vertices in the fragment $F$ in $T$. It is important to note that no fragment has an empty highway that consists of a single vertex, thus each fragment in our decomposition has a unique identifier.
    
    \item Each fragment $F$ has a special vertex $d_F$ which is called the unique descendant of the fragment. The unique path between $r_F$ and $d_F$ is called the \textit{highway} of the fragment. Each fragment has a single highway path. The vertices $r_F$ and $d_F$ are the only two vertices of the fragment $F$ which can occur in other fragments.
    
    \item All edges that are not part of the highway, are called \emph{non-highway} edges. Each non-highway path is completely contained inside a single fragment.  
    
    \item Each edge of $T$ takes part in exactly one fragment $F$.
\end{enumerate}

The properties regarding the id, the highway path and non-highway edges of each fragment are satisfied by construction.

\textbf{Number of fragments.} Since there are $O(\sqrt{n})$ marked vertices, and each of them is the most bottom highway vertex of exactly one fragment, there are $O(\sqrt{n})$ edge disjoint final fragments in total.

\textbf{Size bound.} A key observation is that each fragment is completely contained inside the original component $S$ whose vertices belong to, and since the size of each original component is $O(\sqrt{n})$, so is the size of each fragment.

\textbf{Learning the skeleton tree.} There are $O(\sqrt{n})$ vertices that are the start and end of the highways of fragments ($r_F,d_F$), thus all the vertices in the graph can learn the IDs of these vertices along with the ID of the fragment each pair belongs to in $O(\sqrt{n}+D)$ rounds using a BFS tree. Using this information each vertex can immediately learn the topology of the skeleton tree.

Lastly, we prove the 4 conditions stated in the lemma regarding the information each vertex $v$ holds about the fragment $F$ it belongs to it holds that each vertex $v$ knows the following information about the fragment $F$ it belongs to:
\begin{enumerate}
    \item The identity $(r_F,d_F)$ of the fragment $F$.
    \item The complete structure of the skeleton tree $T_S$.
    \item All the edges of the highway of the fragment $F$.
    \item All the edges of the unique path connecting $v$ and $r_F$, and also the edges of the unique path connecting $v$ and $d_F$.
\end{enumerate}

Each node can know the id of its fragment by two standard broadcasts in the fragment, of the highest and lowest vertices on the highway.

Each vertex can also know the complete topology of the skeleton tree by using aggregate computations over a BFS tree of the graph and pipe-lining, and thus learning in $O(\sqrt{n}+D)$ (Since there are $O(\sqrt{n})$ components) round the ids of all fragments in the graph, from which he can deduce the topology of the skeleton tree.
The two last properties follow wince $\sfrag=O(\sqrt{n})$ and thus each node can learn the complete topology of its fragment in $O(\sqrt{n})$ rounds.

This completes the proof of the lemma.

\end{proof}

\subsection{Layering decomposition}

\LayerDecompNonHighway*
\begin{proof}
We fix one fragment $F$ and describe the computation in the fragment. As the computation is completely inside the fragment, we can work simultaneously in all different fragments. In each fragment $F$, we work on the subtree $T_F$ of the fragment, and run an aggregate computation from the leaves to the root. As we only consider non-highways, we stop the computation each time we reach a vertex in the highway. By the end of the computation, all non-highway edges have a layer number.

We next show that we can perform the layering on $T_F$ in $O(\dfrag)$ rounds, using only communication on the edges of $T_F$. Each edge $e$ in $T_F$ holds a number $\ell_e$ initialized to $0$. Now, each edge $e$ does the following. If $e$ is connected to a leaf, $e$ sets $\ell_e=1$, and sends its layer to its parent edge in $T_F$. Otherwise, denote by $e_1,...,e_m$ the descendants of $e$ in $T_F$. $e$ waits to receive $\ell_{e_1},...,\ell_{e_m}$. Then, denote by $\ell_{max}=\max\limits_{i=1,...,m} \ell_{e_i}$. If there are $i\neq j$ such that $\ell_{e_i}=\ell_{e_j}=\ell_{max}$, then $e$ sets $\ell_e=\ell_{max}+1$, otherwise, $e$ sets $\ell_e=\ell_{max}$. In both cases, $e$ sends to its parent edge the value $\ell_e$. 
This process terminates when we reach the edges adjacent to the highway after $O(\dfrag)$ rounds since the diameter of $T_F$ is $O(\dfrag)$, and we are performing a single aggregate computation over the tree $T_F$, and communicating only over the edges of $T_F$.
Correctness is implied from the definition of the layer decomposition and Observation \ref{obs:-no_down_layer}.
\end{proof}

\subsection{Information of edges}

\LearnInfoEdge*

\begin{proof}
Note that  for each tree edge $e$, it holds that $|\info(e)|=O(\log n)$. This is true since:
\begin{enumerate}
    \item Whether $e$ is a highway edge edge or a non-highway edge is $O(1)$ information. Also the information regarding the fragment of $e$ and the values $\cov(e),|\ECov(e)|$ are $O(\log n)$ bits.

\end{enumerate}

Since $\info(e')$ for a non-tree edge $e'$ is comprised of $\info(e)$ of exactly two tree edges, one can deduce that $|\info(e')|=O(\log n)$ as well. Note that, it suffices to prove the theorem for tree edges alone, since given a non-tree edge $e'=\set{v,u}$, if $v$ holds $\info(e'_1),e'_1=\set{p(v),v}$ and $u$ holds $\info(e'_2),e'_2=\set{p(u),u}$, then in $O(\log n)$ rounds, both $u,v$ can learn $\info(e')$. 

The information about whether an edge $e$ is a highway, or a non-highway edge, and the id of the fragment of $e$ is known by $e$ due to the fragment decomposition properties, proved in Lemma \ref{lemma:strong_decomp_construct}, since an edge $e$ in fragment $F$ can check whether the path from $e$ to $r_F$ and the path from $e$ to $d_F$ are vertex disjoint, and if so, $e$ marks itself as a highway edge, and otherwise, as a non-highway edge.

Now for the second bullet 
of the theorem, we employ Lemma \ref{lemma:routing_lemma_aggregate}, in the following way. 
     To learn $|\ECov(e)|$, $f$ is the addition function and the information of each edge in the graph is simply the integer $1$. After the computation is done $e$ adds 1 to count itself as well.


Similarly, as explained in Claim \ref{claim:learn_cov_value}, $e$ knows $\cov(e)$.
\end{proof}

\section{The sampling procedure: Lemma \ref{lemma:find-int-samp}}\label{sec:appen-samp}
Here we prove Lemma \ref{lemma:find-int-samp} which we restate for ease of reference. 

\FindIntSample*



Before proving Lemma \ref{lemma:find-int-samp} in its full generality, we first argue the same for \textit{unweighted graphs} for ease of presentation. Readers may chose to skip the following `warm-up' section.

\bsni
\paragraph{Warm-up: Unweighted graph.} We first prove Lemma \ref{lemma:find-int-samp} for the case of unweighted graph. The idea is simple: Each edge $e$ should sample $O(\log n)$ edges from the set $\ECov(e)$, and make decisions about other interesting edges based on the sampled edges. 

We assume Claim \ref{claim_learn_cov} such that every tree edge $e$ knows the value of $\cov(e)$. The sampling procedure runs in $O(\log n)$ iterations. We divide the tree edges into classes according to the values of $\cov(e)$, as follows. We put a tree edge $e$ in class $C_j$ if $\cov(e) \in [2^{j-1}+1,2^j]$. Note that the number of classes is bounded by $O(\log n)$.  Now, there are $O(\log n)$ many iterations, one for each class $C_j$. We describe iteration $j$ next:

\begin{description}
\item[Iteration $j$.] Each edge $e'$ first generates a unique identifier of length $O(\log n)$, denoted by $Id_e$, for this iteration. This can be done by sampling identifiers randomly from a range of size $n^5$ which each edge can do locally, and with high probability, the identifiers will be unique. Then, each edge $e'$ samples itself independently with probability $\frac{1}{2^j}$. Now, each tree edge $e\in C_j$ uses Claim \ref{claim:route_info} to learn about $\info(e^*)$ (which has a $\poly \log n$ bit representation; see Theorem \ref{thm:-info_theorem_tree_edge} and Claim \ref{claim:route_info})  of the minimal $Id$ edge sampled $e^*$ that covers $e$. This process is repeated $\beta=\log ^2 n$ times. Denote by $\ECovs(e)$ the set of edges from $\ECov(e)$ that were learned by $e$.\footnote{Note that, for making the elements sampled in $\ECovs(e)$ i.i.d., we needed the random $Id_e$ generated in the beginning. Each edge in the set $\ECovs(e)$ is a random edge from a set of i.i.d. sampled edges because each such edge has the minimum $Id_e$ w.r.t. a random ordering of those i.i.d. sampled edges. It is to be noted that similar argument cannot be made if we fixed a global ordering of the edges apriori, such as, fixing an ordering of $V$ and using the induced ordering on the edge set.} A simple concentration bound shows that the number of distinct edges sampled from the set $\ECov(e)$ is $\covs(e) = |\ECovs(e)|=\Omega(\log n)$.\footnote{This is only true when $\cov(e) = \omega(\log n)$. For the case when $\cov(e) = O(\log n)$, we do the same sampling procedure---in this case the set $\ECovs(e)$ is a multiset which w.h.p. will cover the set $\ECov(e)$ completely. We can do exact calculation for such cases, but this does not affect the probabilistic calculation presented next.} We keep $O(\log n)$-many elements in $\ECovs(e)$ and discard the rest.


We denote by $\ECovs(e,P) \subseteq \ECov(e,P)$ for some tree edge $e$ and some ancestor to descendant path $P$ as follows: $\ECovs(e,P)= \ECovs(e) \cap \ECov(e,P) = \set{e''\in \ECovs(e) \mid e''\in \ECov(e, P)}$, and $\covs(e,P) = | \ECovs(e,P)|$. Now we define for each tree edge $e$ the set of potentially interested paths as follows. Here, $P_e$ denotes the path from the root to $e$.
\[
\intpot{e}=\left\{P \mid \covs(e,P)\geq \frac{\covs(e)}{3}\right\}\cup P_e.
\]
\end{description}



We need to verify that property \ref{item:samp-1} and \ref{item:samp-2} of Lemma \ref{lemma:find-int-samp} hold for such a sampling procedure which is done next in Claim \ref{claim:sampling_new_correct_uw} and \ref{claim:samp_correct} respectively. But, first, we quickly check the round complexity of the sampling algorithm.
\begin{claim}\label{clm:round-comp-samp-uw}
The sampling algorithm for unweighted graph can be performed in $\tO(D + \sqrt n)$ rounds at the end of which each tree edge $e$ (and the vertices in $e$) knows the set $\set{\info(e')\mid e'\in \ECovs(e)}$.
\end{claim}

\begin{proof}
The communication happens only in two parts of the algorithm:
\begin{enumerate}
    \item[(i)] The unique identifier of each edge is sent across from one end-point of the edge to the other end-point simultaneously on all edges. As mentioned before, this can be done in one round as there is no congestion on any edge. As there are $O(\log n)$ iterations, the total number of rounds for this operation is also $O(\log n)$.
    
    \item[(ii)] In each iteration, each tree edge $e$ active in that iteration (i.e., $\cov(e)$ falls in the range corresponding to that iteration) uses Claim \ref{claim:route_info} to learn about of $\info(e')$ of $O(\log n)$ many (possibly distinct) edges $e'$ which covers $e$. Each invocation of Claim \ref{claim:route_info} requires $O(D + \sqrt n)$ rounds. As there are $O(\log n)$ many such invocations in each iteration and as there are $O(\log n)$ many iterations, the total round complexity is $\tO(D + \sqrt n)$.
\end{enumerate}
Hence the round complexity of the sampling algorithm is $\tO(D + \sqrt n)$.
\end{proof}


Now we turn to verify Property \ref{item:samp-1} and \ref{item:samp-2} of Lemma \ref{lemma:find-int-samp}.
\begin{claim}\label{claim:sampling_new_correct_uw}
W.h.p, given a tree edge $e$ and an ancestor to descendant path $P^\ast$ such that $P^\ast \in \interesting(e)$, we have $\covs(e,P^\ast)\geq \frac{\covs(e)}{3}$.
\end{claim}
\begin{proof}
By definition of interesting paths, we know that $\cov(e,e')>\frac{\cov(e)}{2}$. Let $j$ be the index such that $e\in C_j$ (i.e., $\cov(e) \in [2^{j-1}+1, 2^j]$ and let us observe sub iteration $j$. It is easy to see that, on expectation, the value of $\covs(e,P^\ast) > \covs(e)/2$. Hence, by an application of Chernoff bound, we see
\begin{align*}
    \Pr[\covs(e,P^\ast) > \covs(e)/3] \geq 1 - o(1).
\end{align*}
This means that, w.h.p., $P^\ast \in \intpot{e}$ which verifies Property \ref{item:samp-1} of Lemma \ref{lemma:find-int-samp}.
\end{proof}
Using union bound, we see that such inequality holds for every tree-edge w.h.p. We actually achieve something stronger by this random process, as mentioned in the following claim.

\begin{claim}\label{claim:samp_correct}
 For every path $P \in \intpot{e}$, $\cov(e, P) \geq \cov(e)/6$ w.h.p.
\end{claim}

\begin{proof}
Suppose not. This means that there is a path $P$ such that $\covs(e, P) \geq \covs(e)/3$ (i.e., $P \in \intpot{e}$) but $\cov(e, P) = \eps \cdot \cov(e)$ where $\eps < 1/6$. As before, we see that $\expt{}{\covs(e,P)} = \eps \cdot \covs(e)$, and by an application of Chernoff bound, we have $\Pr[\covs(e, P) \leq 2\eps \cdot \covs(e) < \covs(e)/3] \geq 1 - o(1)$ which is a contradiction. This verifies Property \ref{item:samp-2} of Lemma \ref{lemma:find-int-samp}.
\end{proof}

\bsni
\begin{proof}[Proof of Lemma \ref{lemma:find-int-samp}]
 To prove Lemma \ref{lemma:find-int-samp}, we have to do similar argument as above for \textit{weighted graphs}. The argument is actually very similar to the unweighted case---we treat every weighted edge $e$ with weight $w(e)$ to be $w(e)$ many parallel unweighted edges. As before, the sampling procedure runs in $O(\log n)$ iterations, one of each class $C_j$. We next point out the difference in iteration $j$ from the case of unweighted graphs. In the rest of the proof, an unweighted edge represents an edge from the purported multi-edge description of a weighted edge. So it is instructive to view the graph as an unweighted graph where multi-edges are allowed.

\begin{description}
\item[Iteration $j$.] We first mention one of the main differences from the unweighted sampling procedure. In the case of unweighted sampling, we assumed that for an edge $e = \{u,v\}$ one of the vertices (say $u$ assuming $u \prec v$ in the total ordering of $V$) will sample the unique id $Id_e$ for $e$ and send it across to $v$. This can happen for all edges simultaneously with no congestion, and requires one additional round of communication. We cannot afford to do exactly that when the edges are weighted, because in this case $u$ will generate $w(e)$ many unique identifiers, one for each of its unweighted edges in its representation of unweighted multi-edge, for the edge $e$. Sending this across to $v$ using the edge $e$ can potentially incur considerable congestion which we cannot afford. We avoid congestion in the following way. The following sampling is done entirely by $u$. For each weighted edge $e = \{u,v\}$, $u$ samples each of the $w(e)$ many identifiers corresponding to $e$ independently with probability $1/2^j$. At this point, for each edge $e$, one of its end-points holds a set of sampled identifiers. Now, for every edge $e$, the end-point which holds this set of identifier will send across the minimum identifier to the other end-point using the edge $e$. This does not cause any congestion. From this point onward, we identify every weighted edge $e$ by the minimum identifier $Id'(e)$ that is sampled by the sampling procedure just mentioned (If there is no such minimum identifier for an edge $e$, we ignore that edge). The rest of the sampling procedure remains similar to that of the unweighted case: Each tree $e$ learns about $\info(e^\ast)$ of the edge $e^\ast$ with minimum $Id'(e^\ast)$ that covers $e'$ using Claim \ref{claim:route_info}. This process is repeated $\beta = \log^2 n$ times and thus each tree edge $e'$ samples $O(\log n)$ many distinct identifiers uniformly from the set of identifiers corresponding to the set of weighted edges $\ECov(e)$. If we imagine each weighted edge as a set of parallel unweighted edges, each with unique identifier, this set of sampled identifiers corresponds to the set of such unweighted edges which are sampled---we denote this set as $\ECovs(e)$. We also denote by $\covs(e,P)$ the number of such unweighted edges that covers both $e$ and $P$. Again, as before, each tree edge $e$ defines $\intpot{e} = \{ P \mid \covs(e,P) \geq \covs(e)/3\}$. It is not hard to see the following: (a) Each identifier in the set $\ECovs(e)$ is included in the set independently and with identical probability (i.e., with probability $1/2^j$) from the set of identifiers corresponding to the set of weighted edges $\ECov(e)$ and (ii) the number of distinct identifiers in the set $\ECovs(e)$ is at least $\log n$. Hence \ref{claim:sampling_new_correct_uw} and \ref{claim:samp_correct} hold for this case as well by similar calculation. 
\end{description}

We now analyze the round complexity of this sampling algorithm. We make a similar claim as Claim \ref{clm:round-comp-samp-uw}. The proof is also similar and hence skipped.

\begin{claim}\label{clm:round-comp-samp-w}
The sampling algorithm for weighted graph can be performed in $O(D + \sqrt n)$ rounds at the end of which each tree edge $e$ (and the vertices in $e$) knows the set $\set{\info(e')\mid e'\in \ECovs(e)}$.
\end{claim}
\end{proof}

\section{Missing proofs from Section \ref{sec:lemma_and_samp}}\label{sec:appen_five}

Below we prove Lemma \ref{lemma:parse_paths} in Section \ref{appn:learn_int_path}, and Theorem \ref{thrm:pairing_up_highways} in Section \ref{appn:hw-pairing}.

\subsection{Learning interesting paths} \label{appn:learn_int_path}

\PathParsing*

For proving the above lemma, we first require some useful claims.
\begin{claim}\label{claim:parsing_stuff}
For each tree edge $e$, given $\info(e')$ for each $e'\in \ECovs(e)$, $e$ can construct internally the set $\intpot{e}$.
\end{claim}

Before diving deeper into the proof, we prove the following useful claim. Intuitively, in graph theoretic terms, we prove here that given any set $A$ of vertices in a tree, the set of LCA's of any subset $S\subseteq A$ of these vertices is contained in the set of LCA's of \emph{pairs} of vertices from $A$ (See Figure \ref{fig:LCA_figure_appen} for an illustration). The notation in the formal statement of the claim adds information that is useful for the way we employ the claim. 

\begin{figure}[h!]
    \centering
    \includegraphics[width=\textwidth, height=9cm]{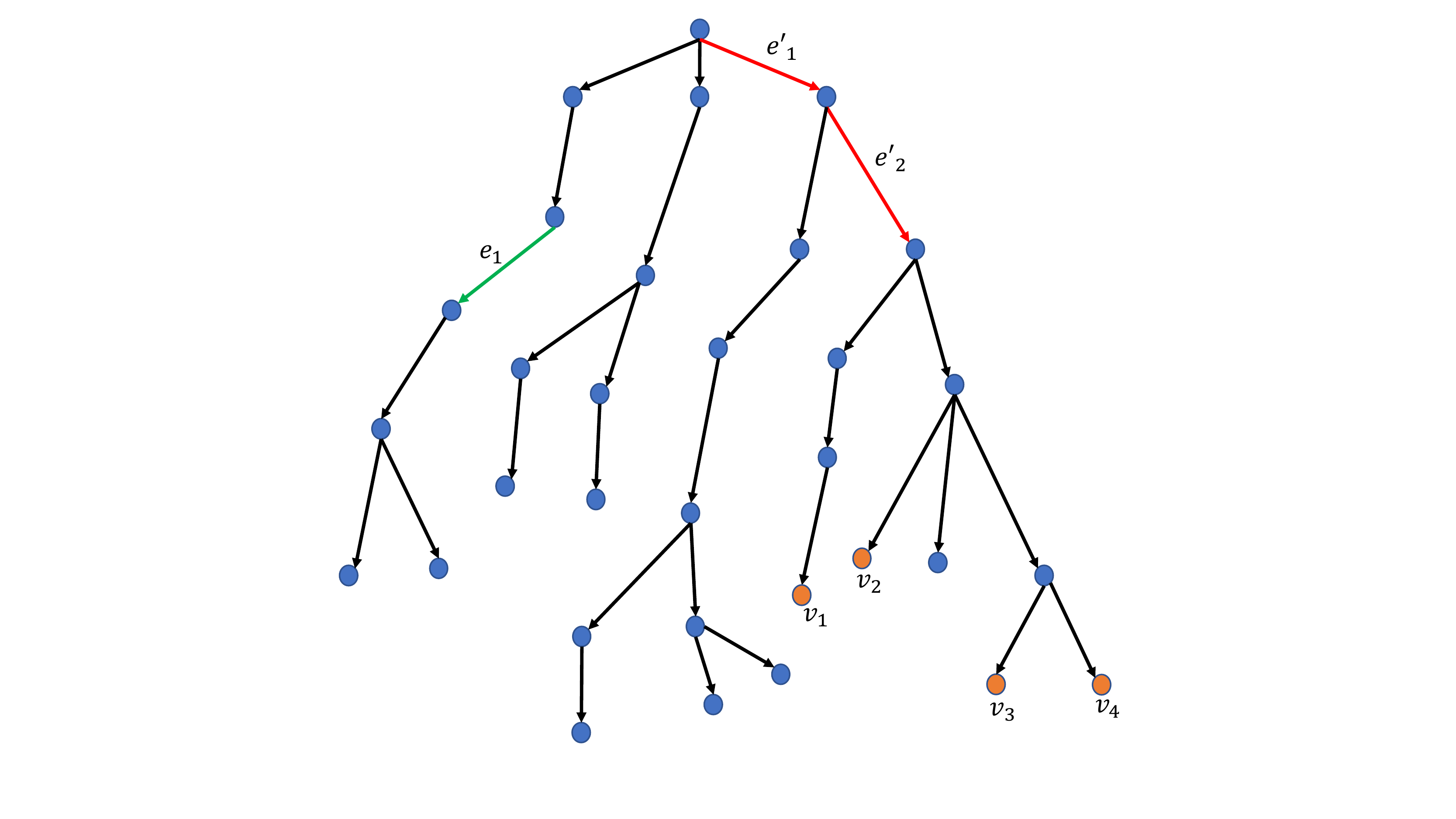}
    \caption{\small  In this figure we have the green edge $e_1$, and the path $P$ which consists of $e'_1,e'_2$, and satisfies $P\in \intpot{e}$. $v_1,v_2,v_3,v_4$ denote the vertices in $T(P)$ that have incident edges in $\ECovs(e_1)$. Note that there exists a pair of vertices (e.g. $v_1,v_3$) such that $\lca{v_1,v_3}$ is exactly the lowest vertex of $P$.}
    \label{fig:LCA_figure_appen}
\end{figure}

\begin{claim}[LCA claim]\label{claim_lca_computation_of_pairs_appen}
Let $e$ be some tree edge, and let $P$ be some root to descendant path such that $e$ is potentially interested in $P$, i.e. $P\in \intpot{e}$. Denote by $v$ the lowest vertex of $P$, and by $e_v$ the lowest edge of $P$ . Furthermore, $P$ is maximal in the sense that $e$ is not potentially interested in any ancestor to descendant path that strictly contains $P$.  For each $e'\in \ECovs(e)$, denote by $e'_1,e'_2$ its respective tree edges. Denote by $D =     \set{e'_i \mid e'\in \ECovs(e,P),i=1,2}$. 

Then, there exists $e^1,e^2\in D$ such that $\lca{e^1,e^2}=v$.
\end{claim}

\begin{proof}
Since $e$ is potentially interested in $P$, there is some edge $e'\in D$ such that $e'$ has an endpoint in $ \dw{e_v}$, this is since $\covs(e,P)\geq \frac{\covs(e)}{3}$, and every edge that covers $P$ has an endpoint in $\dw{e_v}$. Now that we know that $D$ is not empty, let $e'\in D$ be some edge, and denote by $v_*$ the highest vertex in the set $\set{\lca{e',e^*}\mid e^*\in D}$. Note that such a vertex exists since $D$ is not empty and each vertex in the aforementioned set is an ancestor of $e'$, so there is a well defined order on these vertices, in terms of closeness to the root.

Our goal will be to prove that $v_*=v$, and this will conclude the proof. For that, we first show that $\set{\lca{e_1,e_2}\mid e_1,e_2\in D}\subseteq \dw{v_*}$, i.e., $v^*$ is the highest vertex among all local common ancestors of all pairs of edges in $D$. Assume towards a contradiction that this does not hold and denote by $e_1,e_2\in D$ edges that satisfy $\lca{e_1,e_2}\not\in \dw{v_*}$. Now, since both $\lca{e',e_1},\lca{e',e_2}$ are ancestors of $e'$, one of these vertices is higher than the other, assume w.l.o.g. it is $\lca{e',e_1}$, it must hold also that $\lca{e',e_1}\not\in \dw{v_*}$, since otherwise we get that both $e_1,e_2\in \dw{v_*}$, then so is their LCA, which is a contradiction. But since $\lca{e',e_1}$ is an ancestor of $e'$, it is either a descendant of $v_*$, or an ancestor of it. the former contradicts  $\lca{e',e_1}\not\in \dw{v_*}$, and the latter contradicts the fact that $v^*$ the highest vertex in the set $\set{\lca{e',e^*}\mid e^*\in D}$. Either way, we arrive at a contradiction. Thus $\set{\lca{e_1,e_2}\mid e_1,e_2\in D}\subseteq \dw{v_*}$, as desired.

Now in order to show that $v_*=v$, we note that basically all edges in $D$ are by definition exactly all the edges adjacent to edges in $\ECovs(e,P)$. Thus we proved that $\covs(e,P_{v_*})\geq \frac{\covs(e)}{3}$, where $P_{v_*}$ is the path from $\lca{e,v_*}$  to $v_*$. Thus $e$ is potentially interested in $P_{v_*}$, and since $P$ is maximal with respect to that, we deduce that $v_*=v$ as required.
\end{proof}

\begin{proof}[Proof of Claim \ref{claim:parsing_stuff}]

Let $e'=\{u,v\}\in \ECovs(e)$ where $v \in \dw{e}$. Denote by $e_1,e_2$ the edges $e_1=\{p(u),u\},e_2=\{p(v),v\}$. $\info(e')$ contains $\info(e_1),\info(e_2)$, i.e.

\begin{itemize}

\item Whether $e_1,e_2$ are highway/non-highway edges. 

\item The id of the fragments of $e_1,e_2$. From this, using Observation \ref{obs:knowing_allfragments_from_one}, $e$ can deduce the set of fragments $F$ that satisfy $F\cap P_{e_1}\neq \emptyset$ or $F\cap P_{e_2}\neq \emptyset$.

\end{itemize}

Now, with the LCA claim proven, in order to figure out the root to descendant paths that are in $\intpot{e}$, it suffices for $e$ does the following internally. Denote by $S$ the set $\set{e'_i\mid e'\in \ECovs(e),i=1,2}$. For each pair $e_1,e_2\in S$, $e$ computes $v=\lca{e_1,e_2}$, and then computes $z=\lca{u,v}$ ($u$ represents the lower vertex of $e$) and then computes $\covs(e,P_v^z)$ in the following way ($P_v^z$ is the ancestor to descendant path from $z$ to $v$).  $e$ computed $z$, and based on that divides into two cases:
\begin{enumerate}
    \item If $z=u$, then for each edge $e'=\set{v',u'}\in \ECovs(e)$, $e$ computes $\lca{v,v'},\lca{v,u'}$, if neither of these equal $v$, $e$ moves on to the next edge. Otherwise, assume w.l.o.g $\lca{v,v'}=v$. $e$ checks whether $\lca{u,u'}\neq u$, and if so concludes that $e'\in \ECovs(e,P^z_v)$ and adds $w(e')$ to $\covs(e,P_v)$. Otherwise, $e$ moves on to the next edge.
    \item If $z\neq u$, then for each edge $e'=\set{v',u'}\in \ECovs(e)$, $e$ computes $\lca{v,v'},\lca{v,u'}$, if neither of these equal $v$, $e$ moves on to the next edge. Otherwise, assume w.l.o.g $\lca{v,v'}=v$. $e$ checks whether $\lca{u,u'}=u$, and if so concludes that $e'\in \ECovs(e,P_v)$ and adds $w(e')$ to $\covs(e,P^z_v)$. Otherwise, $e$ moves on to the next edge.
\end{enumerate}

$e$ then keeps all paths that satisfy $\covs(e,P)\geq \frac{\covs(e)}{3}$. Note that for all such paths $e$ knows the path $P$ (See definition \ref{def:edge_knows_path}). This is true since for each $e'\in \ECovs(e,P)$, $\info(e)$ contains all the relevant information about $P$ (See \ref{thm:-info_theorem_tree_edge} and discussion afterwards about non-tree edges), and $e$ knows the ending vertex $v$ of $P$, and the starting vertex $z$ of $P$.
Note that $e$ can also know whether $v$  is the bottom end of a non-highway edge using LCA computations, i.e. if $P$ contains non-highway edges. Furthermore note that $e$ can know the fragments of $v,z$ by the fact that $e$ knows both the skeleton tree (See Lemma \ref{lemma:strong_decomp_construct}) and the fragments of all nodes with incident edges in $\ECovs(e)$ (See definition \ref{sec:-info_tree_dge}). Thus $e$ knows all the necessary information about all paths in $\intpot{e}$.  The resulting set of paths is $\intpot{e}$, up to two modifications. $e$ always includes the path from the root to $e$, $P_e$ in $\intpot{e}$, $e$ knows this path by the information gathered by $e$ in Theorem \ref{thm:-info_theorem_tree_edge}.


Lastly, $e$ discards from $\intpot{e}$ any path $P$ that is not maximal, i.e. such that there exists a path $P'\in \intpot{e}$ that satisfies $P\subsetneq P'$. The resulting set is precisely $\intpot{e}$ as constructed in Lemma \ref{lemma:find-int-samp}.  
\end{proof}

\begin{proof}[Proof of Lemma \ref{lemma:parse_paths}]

Now that we showed that an edge $e$ can internally construct its own set potentially interested paths $\intpot{e}$, the lemma follows by using standard routing techniques.
The length of any fragment highway and any non-highway path is $\Tilde{O}(\dfrag)$. For each edge $e$ it holds that $|\intpot{e}|=O(\log n)$ , and the representation of each path is $O(\log n)$ bits, thus each such path $P$ can upcast $\intpot{e}$ for each of its edges up to the highest vertex, and then broadcast it down the path in $\tO(\dfrag)$ rounds. This can be done in parallel since all non-highway boughs and all fragment highways are edge disjoint. Now, the last part of the lemma follows since any non-highway ancestor to descendant path is of length at most $\tO(\dfrag)$, thus each non-highway edge $e$ can broadcast down its tree the set $\intpot{e}$ along with its layer in the layer decomposition to allow edges below to distinguish between different paths, and using pipe-lining the desired statement in achieved in $\tO(\dfrag)$ rounds. Furthermore, given a highway path $P$ of a fragment $F$, The vertices of $P$ can in $O(\dfrag)$ send $\intpot{P}$ to the root of $F$. Then, again in $O(\dfrag)$ rounds, the root of $F$ can send $\intpot{P}$ to all the vertices of $F$, as required.

Note that since each path in $\intpot{e}$ for every $e$ is represented by some fragment $F$ (the lowest fragment of the path), each edge now holds a set of fragments that exactly represent the set of paths that $e$'s  respective path is potentially interested in. Each edge $e$ in a path $P$ also knows for each of these fragments whether $P$ is potentially interested in a non-highway path in that fragment or not (see Definition \ref{def:edge_knows_path}). Thus, each edge $e$ can discard from the resulting set of fragments those fragments that are both not the lowest and $P$ is not potentially interested in a non-highway path in $F$ (i.e. fragments $F$ in the set $\intpot{P}$ such that there is a descendant fragment $F'$ of $F$ in $\intpot{P}$ and $P$ is not potentially interested in a non-highway path in $F$).
Note that this can be done internally in each tree edge $e$, since each tree edge knows the skeleton tree (See Lemma \ref{lemma:strong_decomp_construct}).

The resulting set accurately represents $\intpot{P}$ for a given bough non-highway or fragment highway $P$.

\end{proof}

\subsection{Proof of highway pairing theorem.} \label{appn:hw-pairing}

\HighwayPairingTheorem*

\begin{proof}

Define the following initial pairing
\begin{gather*}
    R^*=\set{(P_1,P_2)\in \cP \times \cP \mid\\ P_1\ \text{is potentially interested in } P_2\ \text{and vice versa}}
\end{gather*}

We first prove that $R^*$ satisfies all requirements of the theorem, except for maybe the last one, and then we modify $R^*$ to meet all requirements. 
Clearly $R^*$ satisfies condition 1 of the theorem. Now, for the other requirements.

\paragraph{Requirement 2.} Consider some fragment $F$ in some bough highway path $P$, and the set $R_F$. Now, By Observation \ref{obs:orthogonal_maximal_fragment_highways} and  by Corollary \ref{corol:-_fragment_highway}, we know that $P_F$ is potentially interested in at most $B_{path}$  highway paths of any given layer.
Thus, since the number of layers is $O(\log n)$, $P_F$ is potentially interested in at most $B_{path}\cdot \log n$ highway paths of a given layer. Thus there can be at most $O(B_{path}\cdot \log n)$ pairs $(P,P')\in R^*$ s.t. $P_F$ is active with respect to that pair. This concludes requirement 2 of the theorem.

\paragraph{Requirement 3.} For the third requirement, assume $e_1,e_2$ are such edges and $P_1,P_2$ their respective bough highway paths (Each path is represented by its lowest and highest fragment). Now, by Lemma \ref{lemma:find-int-samp}, and Lemma \ref{lemma:parse_paths} and discussion right afterwards, we know that w.h.p $P'_1\in \intpot{e_2}$ and $P'_2\in \intpot{e_1}$ for some paths $P'_1\subseteq P_1, P'_2\subseteq P_2$. Thus by definition of $R^*$, we can deduce that $(P,P')\in R^*$. 

\paragraph{Requirement 4.} Now, for the fourth requirement. Consider a fragment $F$, each bough highway path that $P_F$  is potentially interested in is known to some vertex in $P_F$. This is true by definition \ref{def:-interest_path_cover} interesting paths and Claim \ref{lemma:parse_paths}. 
Furthermore, since there are $N_{frag}=O(\sqrt{n})$ fragments, and each highway path contained in a single fragment is active in at most $O(B_{path}\cdot \log n)$ pairs $(P,P')\in R$, one can deduce that $|R^*| = O(\sqrt{n}\cdot B_{path}\cdot \log n)$. Thus, using a BFS tree of the network, the network can elect some leader vertex $v_0$, and upcast 
towards it the elements of $R$, where each vertex $v$ 
sends the pairs of which it is participating in and its fragment is active along with the Id of the fragment $F$ that $v$ is contained in (and thus $F$ is active in all these pairs). Again, this can be done since each bough highway path that $P_F$ is potentially is known to some vertex in $P_F$. Thus all pairs in $R^*$ in which $P_F$ is active are known to vertices in $P_F$. Then, $v_0$ broadcasts these pairs along with the active fragments for each pair back to all other nodes. Using pipe-lining, both of these procedures can be completed in $\Tilde{O}(D+\sqrt{n})$ rounds, as required. This is true since the total amount of Ids of active fragments that $v_0$ needs to send is $O(\nfrag \cdot B_{path}\cdot \log n)$.

\paragraph{Requirement 5.}
Now, each vertex $v$ does the following internally. $v$ goes over all pairs $(P_1,P_2)\in R^*$, if $P_1,P_2$ don't split one another, $v$ does nothing and continues to the next pair. Otherwise, $v$ removes the pair $(P_1,P_2)$ from $R^*$ and does the following. Assume w.l.o.g that $P_1$ splits $P_2$, $v$ knows both paths $P_1,P_2$, and therefore $v$ knows the fragment in  which $P_1,P_2$ intersect, denote it by $F$ (It is a part of $P_2$, but not $P_1$). $v$ splits the path $P_2$ into 2 paths, $P'_2,P''_2$. Here, $P'_2$ is the sub-path of $P_2$ which includes all edges of all fragments higher than $F$, and $P''_2$ includes all the rest of the edges. Note that Both $P'_2,P''_2$ are super highways. Now $v$ considers the pairs $(P_1,P'_2),(P_1,P''_2)$, note that both of these pairs do not split each other. Now, since $v$ knows the active fragments of $P_2$ with respect to the pair $(P_1,P_2)$, $v$ knows whether $P'_2,P''_2$ are potentially interested in $P_1$, and if so, knows the active fragments in the pairs $(P_1,P'_2),(P_1,P''_2)$. If one of $P'_2,P''_2$ is not potentially interested in $P_1$, $v$ discards the appropriate pair. Then, $v$ adds to $R^*$ the pair $(P_1,P'_2)$ if it wasn't discarded, and the pair $(P_1,P''_2)$ if it wasn't discarded.

Denote the resulting set by $R$.

The first criteria of the theorem holds.
for the second criteria, since the number of pairs each fragment is active in in the transition from $R^*$ to $R$ doubled at most, this criteria holds as well.

Since each vertex knows $R^*$, and each vertex constructs $R$ without any further communication, the third criteria also holds.

The fifth criteria holds since  by the modification we made to $R^*$, every pair of paths in $R$ don't split one another. This concludes the proof.
\end{proof}
\section{Missing proofs from Section \ref{sec:alg_short_path}}\label{sec:app_short_paths}
\subsection{Basic subroutines on a tree}\label{sec:app_short_paths_prelim}

\pipeline*

\begin{proof}
As explained above, each broadcast or aggregate computation requires $O(D_{T'})$ time, and the computation requires only communication on edges of the tree. To complete the proof, we explain how to pipeline the computations efficiently. By the description, each broadcast requires sending only one message from each vertex to each one of its children, when we go over the tree from the root to the leaves, hence clearly we can pipeline such computations. Similarly, an aggregate computation requires sending one message from each vertex to its parent, when we go over the tree from the leaves to the root, which can also be pipelined easily. Since the two types of computations send messages in opposite directions, they do not interfere with one another, and we can run broadcast and aggregate computations in the same time, which results in a complexity of $O(D_{T'} + c_1 + c_2)$ time.
\end{proof}

\pipelineFragment*

\begin{proof}
Pipelining of broadcast and aggregate computations is already discussed in the proof of Claim \ref{claim_pipeline}. We next explain how we pipeline also aggregate computations where $d_P$ is the root. To pipeline aggregate computations in two different directions (where $r_P$ and $d_P$ are the roots in the two computations), we work as follows. Note that orienting edges towards $d_P$ only changes the orientation of highway edges (see Figure \ref{orientation_pic} for an illustration), hence in all subtrees attached to the highway, the communication pattern in the two computations is identical, and we can just pipeline them as before. On highway edges, we send messages in opposite directions in both computations, hence they do not interfere and we can run them on the same time. To pipeline a broadcast (from $r_P$) and an aggregate computation with root $d_P$, we work as follows. On the subtrees attached to the highway, we send messages on opposite directions in these computations, hence they do not interfere. On highway edges, we send messages on the same direction, from $r_P$ towards $d_P$, since both computations send one message per edge in the same direction, we can just pipeline them, as we would pipeline two broadcast computations. This completes the proof.
\end{proof}

\subsection{Simple cases where $P'$ is a non-highway}\label{sec:app_short_paths_nh}

\CompareFragments*

\begin{proof}
As all vertices in $F$ know the values $\{e,\cov(e)\}_{e \in F}$, and the edge $f$ has an endpoint in $F$ and an endpoint in $T(P')$, it knows this information, and can pass it to all vertices in $T(P')$ using $O(\sfrag)$ aggregate and broadcast computations in $T(P').$ Note that since the subtrees $T(P'^{\downarrow})$ are disjoint for non-highways not in the same root to leaf path, the edge $f=\{u,v\}$ where $u \in P', v \in F$ cannot be used by other paths in the fragment of $P'$, which allows working in parallel as needed. As $F$ and $P'$ are in different fragments, and $P'$ is a non-highway, we have that all edges $e \in F$ are not in $T(P')$, hence we can use Claim \ref{claim_non_highway} to let all edges $e' \in P'$ learn the values $\{e,\cut(e',e)\}$, this takes $O(\sfrag)$ time using pipelining, and can be done in parallel for different non-highways $P'$ not in the same root to leaf path. 
\end{proof}

\NhFragment*

\begin{proof}
First, in $O(\sfrag)$ time, all edges $e'$ in $F_{P'}$ can learn the values $\{e,\cov(e)\}$ for all edges $e \in F_{P'}$ by aggregate and broadcast computations in the fragment. This can be done for all fragments simultaneously.

We can now use Claim \ref{claim_non_highway} to compare the edges of $P'$ to all edges outside $T(P')$ in the fragment $F_{P'}$. This takes $O(\sfrag)$ time using pipelining, and can be done in parallel for different paths not in the same root to leaf path. After this, each edge $e' \in P'$, knows the values $\cut(e',e)$ for all $e \in F_{P'} \setminus T(P').$

To complete the proof, we show how to compare pairs of edges in $T(P')$. We show that any edge $e' \in T(P')$ can compute the values $\cut(e',e)$ for edges $e \in T(P')$ that are above or orthogonal to $e'$. To compute it, we fix an edge $e=\{v,p(v)\} \in T(P')$ and run an aggregate computation inside $T(P')$ as described in the proof of Claim \ref{claim_non_highway}, with the difference that we stop it when we reach any vertex that is equal to $v$ or an ancestor of $v$. This computation only reaches edges that are below or orthogonal to $e$, and hence $e$ is not in the subtree below them which is enough for the correctness of the computation. By the end, all edges $e'$ that are below or orthogonal to $e$ in $T(P')$, know the value $\cut(e',e)$, as needed. Using pipelining we can run such computations for all edges $e \in T(P')$ in $O(\sfrag)$ time.
As the computations are inside $T(P')$, it can be done in parallel for other paths not in the same root to leaf path.
\end{proof}

\subsection{$P'$ is a non-highway and $P$ is a highway}\label{sec:app_short_paths_nhh}

\coverNhH*

\begin{proof}
First note that by Claim \ref{claim_subtree}, any edge $x$ that covers $e'$ has one endpoint in the subtree $T(P'^\downarrow)$ which is inside the fragment $F_{P'}$ of $P'$. 
Now the edges that cover $e'$ and $e$ are all the edges with one endpoint in $T(P')$ that cover $e'$ and $e$. If the second endpoint of the edge is in the fragment $F_P$, these edges are counted in $\cov_{F}(e',e)$. To complete the proof, we need to show that any edge $x$ that covers $e'$ and $e$ and has both endpoints outside $F_P$, must cover the whole highway $P$. If $x = \{u,v\}$ covers both $e'$ and $e$ with both endpoints outside $F_P$ and $u \in T(P')$, the unique $u-v$ path in the tree starts in the fragment $F_{P'}$, and then enters and leaves the fragment $F_{P}$. From the structure of the decomposition, the only two vertices in $F_P$ that are connected by an edge to vertices outside the fragment are the ancestor $r_P$ and descendant $d_P$ of the fragment, hence any path that enters and leaves $F_P$ must include $r_P$ and $d_P$ and the whole path between them, which is the highway $P$. Hence, $x$ must cover the whole highway $P$ as needed. 
\end{proof}

\paragraph{Computing $\cov_{F}(e',e)$.}

\CovFnh*

\begin{proof}
The proof is similar to the proof of Claim \ref{claim_non_highway}. When we fix an edge $e \in P$, computing $\cov_{F}(e',e)$ for all edges $e' \in P'$ is an aggregate computation in $T(P').$ This follows as the edges counted in $\cov_{F}(e',e)$ for $e' = \{v',p(v')\}$ are all the edges with one endpoint in the subtree of $v'$ (which is contained in $T(P')$) and one endpoint in $F_P$ that cover $e'$ and $e$. To compute the cost of those we just run an aggregate computation in $T(P')$. Note that for each edge $x$ with an endpoint $v$ in $T(P')$, $v$ knows if the second endpoint is in $F_P$, and also can deduce from the LCA labels using Claim \ref{claim_LCA_labels} whether $x$ covers $e$ and $e'$, which allows computing the aggregate function. At the end of the computation, each vertex $v' \in T(P')$ knows exactly the cost of edges in its subtree that cover $e'$ and $e$ with the second endpoint in $F_P$, which is exactly $\cov_{F}(e',e)$ for $e' = \{v',p(v')\}.$ 
\end{proof}

\CovFh*

\begin{proof}
Here we break into cases according to the connection between $P'$ and $P$.
One case is that they are orthogonal, and one case is that $P$ is a highway above $P'$.
Note that from the structure of the decomposition, there are no other cases, as any path in the tree between a descendant and an ancestor starts with a non-highway and highways above it, so we cannot have $P'$ above $P$. Note also that all vertices know the complete structure of the skeleton tree, and can deduce accordingly in which one of the two cases we are.

\emph{Case 1: $P'$ and $P$ are orthogonal.} Here we know that $T(P')$ and $T(P)$ are disjoint, hence from Claim \ref{claim_subtree} it follows that all edges that cover $e' \in P$ and $e \in P$ have one endpoint in $T(P')$ and one endpoint in $T(P).$ Note that $T(P)$ is not necessarily contained in $F_P$, but since we are interested only in computing $\cov_{F}(e',e)$, we are only interested in edges with one endpoint in $T(P')$ and one endpoint in $T(P) \cap F_P$ that cover $e'$ and $e$. Computing the total cost of these edges requires one aggregate computation in $T(P) \cap F_P$, and follows exactly the computation described in the proof of Claim \ref{claim_covF_nh}.

\emph{Case 2: $P$ is above $P'$.} Here we need to do an aggregate computation in the reverse direction, to explain this we first take a closer look on edges that cover $e'$ and $e$ in this case. Since $e'$ is below $e$ in the tree, any tree path that contains both of them must include the whole tree path between them, which in particular includes a part between the descendent $d_P$ of the fragment $F_P$ to $e \in P.$ Hence, we need to sum the cost of edges that cover $e'$ and $e$ and also the whole path between $e$ and $d_P$. To compute it, it would be helpful to reverse the orientation in the fragment $F_P$, such that now $d_P$ is the \emph{root} of the fragment. Now if $e = \{v,d(v)\}$ where $d(v)$ is the vertex closer to $d_P$, all edges that cover $e'$ and the path between $e$ and $d_P$ and have an endpoint in $F_P$, must have an endpoint in the subtree of $v$ according to the new orientation (the subtree that includes everything below $v$ when we think about $d_P$ as the root). Hence, to compute the total cost of edges that cover $e'$ and $e$ and have one endpoint in $F_P$, we just do one aggregate computation in $F_P$ in the reverse direction to sum the costs of these edges. At the end of the computation, the vertex $v$ such that $e = \{v,d(v)\}$ knows exactly $\cov_F(e',e).$ Similarly, for each edge $e \in P$, one of its endpoints knows $\cov_F(e',e)$ as needed.
\end{proof}

\paragraph{Computing $\extcov(e',P)$.}

\ClaimExtnh*

\begin{proof}
Let $e' = \{v',p(v')\} \in P'$, and fix a highway $P$. The edges that cover $e'$ have one endpoint in $T(P') \subseteq F_{P'}$ by Claim \ref{claim_subtree}. Hence, to compute the total cost of edges that cover $e'$ and the whole highway $P$ we just need to do one aggregate computation in $T(P').$ To implement the computation, we need to explain how given an edge $x$ we know if it covers the whole highway $P$ and if both its endpoints are outside $F_P.$ The second is immediate. For the first, note that $x$ covers the whole highway $P$ iff it covers the highest and lowest edges in the highway, since in this case it must cover the whole path between them which is the highway. As the highest and lowest edges in the highway are known to all vertices by Claim \ref{claim_high_low}, and since we can use LCA labels of edges to learn if $x$ covers some edge by Claim \ref{claim_LCA_labels}, we can compute the aggregate function. At the end of the computation, each vertex $v' \in P'$ knows the total cost of edges in its subtree that cover the whole highway $P$ and have both endpoints outside $F_P$, this is exactly $\extcov(e',P)$ for the edge $\{v',p(v')\}.$
This requires one aggregate computation. To compute the values $\extcov(e',P)$ for all highways we run $O(\nfrag)$ such computations, which results in $O(\dfrag + \nfrag)$ time as we pipeline the computations. As the whole computation was inside $T(P')$, we can run in parallel in disjoint subtrees, and in particular we can run the computations in parallel for all non-highways $P'$ in the same layer.
\end{proof}

\subsection{$P'$ and $P$ are highways}\label{sec:app_short_paths_hh}

\covHighway*

\begin{proof}
First, if $x = \{u,v\}$ is an edge that covers $e \in P$, then the unique tree path between $u$ and $v$ contains $e$. As the the only vertices in $F_{P}$ that are connected by an edge to a vertex outside $F_{P}$ are the root $r_P$ and unique descendant $d_P$ of the fragment, any path in the tree that contains $e$ and have both endpoints outside the fragment, must also contain the whole highway between $r_P$ and $d_P$ (as otherwise, at least one of the endpoints would be inside the fragment). Hence, if both endpoints of $x$ are outside $F_P$, it covers the whole highway $P$. 

Now any edge $x$ that covers $e'$ and $e$, has the following options.
\begin{enumerate}
\item $x$ has both endpoints outside $F_{P'}$ and $F_P$, in this case it follows that $x$ must cover the whole highways $P'$ and $P$, which is exactly counted by $\extcov(P',P).$
\item $x$ has one endpoint in $F_{P'}$ but two endpoints outside $F_P$, in this case $x$ must cover the whole highway $P$, which is exactly counted in $\extcov(e', P).$
\item $x$ has one endpoint in $F_{P}$ but two endpoints outside $F_{P'}$, in this case $x$ must cover the whole highway $P'$, which is exactly counted in $\extcov(e, P').$
\item $x$ has two endpoints in $F_{P'} \cup F_{P}$, in this case it must be the case that $x$ has one endpoint in $F_{P'}$ and one endpoint in $F_{P}$, as otherwise the unique tree path defined by $x$ is contained entirely in one of the fragments and cannot cover both edges $e' \in P', e \in P$. This is counted by $\cov_{F}(e',e)$.
\end{enumerate} 
\end{proof}

\paragraph{Computing $\extcov(P,P')$.}

\claimCovH*

\begin{proof}
Given some edge $x = \{u,v\}$, it knows if its endpoints are outside $F_P \cup F_{P'}$. In addition, it can learn if it covers the highway $P$ as follows. An edge $x$ covers the whole highway $P$ iff it covers both the highest and lowest edges in $P$: $e_1, e_2$, because in this case it follows that the unique tree path between $u$ and $v$ contains the unique tree path between $e_1$ and $e_2$ which is the highway $P$. Using the LCA labels of edges and Claim \ref{claim_LCA_labels}, $x$ can learn if it covers $e_1,e_2$. Similarly, $x$ can learn if it covers the whole highway $P'$, and then deduce if it covers both $P$ and $P'$. Now to compute $\extcov(P,P')$ we sum the cost of all edges that cover both $P$ and $P'$ and have both endpoints outside $F_P \cup F_{P'}$, this is a sum of values that can be computed using an aggregate computation in a BFS tree. For each edge $x$ outside $F_P \cup F_{P'}$ that covers both $P$ and $P'$, one of its endpoints represents $x$ in the computation, and adds its cost to the sum computed. To let all vertices learn the value $\extcov(P,P')$ we use a broadcast in the BFS tree. If we have $k$ different pairs, we pipeline the computations to get a complexity of $O(D+k)$.    
\end{proof}

\paragraph{Computing $\cov_F(e',e)$.}

\ClaimCovHF*

\begin{proof}
If we look at the highways $P$ and $P'$, there are several cases:
\begin{enumerate}
\item $P$ and $P'$ are orthogonal. \label{case_or}
\item $P$ and $P'$ are in the same root to leaf path with $P'$ below $P$. \label{case_below}
\item $P$ and $P'$ are in the same root to leaf path with $P$ below $P'$. \label{case_above}
\end{enumerate}
Note that all vertices know the complete structure of the skeleton tree, hence they can distinguish between the cases.
We start with the first two cases.

\emph{Cases \ref{case_or} and \ref{case_below}}. In these cases, we show that any edge $x$ that covers $e \in P$ and some edge $e' = \{v',p(v')\} \in P'$ and has endpoint in $F_{P'}$ has an endpoint in the subtree of $v'$ in $F_{P'}$. 
This is justified as follows. We know that any edge $x$ that covers $e'$ has an endpoint $u$ in the subtree rooted at $v'$ by Claim \ref{claim_subtree}. Assume to the contrary that $u \not \in F_{P'}$, but rather the second endpoint of $x$, $u'$, is in $F_{P'}$. Now the path between $u$ and $u'$ has a part in $F_{P'}$, and a part below it, as $u$ is in the subtree rooted at $v' \in P'$. This path can only cover edges in $F_{P'}$ or below it, but we are in the case that $P$ is orthogonal to $P'$ or above it, hence it does not have any edge in $F_{P'}$ or below it, which means that $x$ cannot cover $e \in P$, a contradiction.

Hence, we know that any edge that covers $e \in P$ and $e' \in P'$ and also has an endpoint in $F_{P'}$ has an endpoint in $T(P') \cap F_{P'}$.
Hence, given an edge $e \in P$, computing the values $\cov_{F}(e',e)$ for all edges $e' \in P'$ requires one aggregate computation in $F_{P'}$, in which every vertex $v' \in F_{P'}$ learns the total cost of edges adjacent to its subtree in $F_{P'}$ that cover also $e$, and have the second endpoint in $F_P$. As explained in the proof of Claim \ref{claim_non_highway} this is an aggregate computation in the subtree, and it can be computed as for each non-tree edge $x$ we can deduce it it covers $e$ using its LCA labels (see Claim \ref{claim_LCA_labels}). At the end of the computation, the vertex $v'$ such that $e' = \{v',p(v')\}$ knows exactly the value $\cov_{F} (e',e).$

\emph{Case \ref{case_above}.} In this case $P$ and $P'$ are in the same root to leaf path with $P$ below $P'$. 
Let $e' \in P', e \in P$, we start by analysing the structure of edges that cover $e'$ and $e$ in this case. Since $P$ is below $P'$, any edge $x$ that covers $e'$ and $e$ covers the whole path between $e'$ and $e$. In particular, it covers the whole path between $e' \in P'$ to the unique descendant $d_{P'}$ of the fragment $F_{P'}$. Hence, we need to sum the costs of edges that cover $e'$ and $e$, have endpoints in $F_{P'}$ and $F_P$, and also cover the whole path between $e'$ and $d_P$.
To compute the cost of these edges, it would be helpful to change the orientation in the fragment $F_{P'}$, such that now $d_{P'}$ is the \emph{root}. We write $e' = \{v,d(v)\}$ where $d(v)$ is the vertex closer to $d_{P'}$ (this is the reverse orientation compared to the previous cases). Now, given an edge $e \in P$, we want to compute the total cost of edges adjacent to the subtree of $v$ in $F_{P'}$ that cover $e,e'$ and have the second endpoint in $F_P$, where the subtree is with respect to the new orientation (having $d_{P'}$ as the new root). This gives exactly $\cov_{F} (e',e).$ Again, this is an aggregate computation in the fragment $F_{P'}.$

Note that the aggregate computation in Case \ref{case_above} is in a different direction than the aggregate computations in Cases \ref{case_or} and \ref{case_below}. However, since all vertices know the structure of the skeleton tree, they all know in which case we are, and can change the orientation accordingly, which only requires changing the orientation for highway edges, as discussed in Section \ref{sec:preliminaries_paths}.
\end{proof}

\paragraph{Computing $\extcov(e,P)$.}

\ClaimCoveP*

\begin{proof}
First, we fix two highways $P',P$, and show how all edges $e' \in P'$ compute $\extcov(e',P).$ 
The edges that cover $e'$ and the whole highway $P$ are exactly the edges that cover $e'$ and the highest and lowest edges in $P$: $e_1,e_2.$
To compute $\extcov(e',P)$, which is the sum of all such edges that have one endpoint in $F_{P'}$ and one endpoint outside $F_{P'} \cup F_{P}$ we use an aggregate computation in $F_{P'}$. This aggregate computation is very similar to the computation done in the proof of Claim \ref{claim_covl}, with the difference that now instead of summing the cost of edges that cover $e'$ and a specific edge $e$, we sum the cost of edges that cover $e'$ and $e_1,e_2$, and also make sure that the second endpoint is not in $F_{P'} \cup F_{P}$, the computation can be done in the same manner. For each highway we have one aggregate computation to compute $\extcov(e',P)$ for all edges $e' \in P'.$ To do so for all highways $P$, we need $\nfrag$ such computations. Since the whole computation was done inside $F_{P'}$, we can do the same computation in all fragments simultaneously, the overall complexity is $O(\dfrag + \nfrag)$.
\end{proof}

\subsection{Both edges in the same highway}\label{sec:both-edge-highway}


Here we also cover the case when both edges defining the cut are in the same highway $P$ inside a fragment. Let $e_1$ and $e_2$ be the two edges in a highway $P$ inside a fragment $F$ where $e_1$ is closer to the root than $e_2$. Let us also denote the top edge of $P$ (closest to the root of $G$) inside $F$ to be $e_r$ and the bottom edge of $P$ (farthest from the root of $G$) inside $F$ to be $e_d$. As before, we look at the $\cov(e_1, e_2)$. We see that $\cov(e_1,e_2)$ can be broken up into the following parts (See Figure \ref{fig:both-edge-high})\footnote{This is a slight abuse of notation of $\cov(\cdot)$ but is clear from the context.}:
\begin{itemize}
    \item The cost of edges that cover $e_1, e_2$ and have both end-points inside $F_P$: $\cov_F(e_1,e_2)$,
    
    \item The cost of edges with one end point in $F_P$ and the other end-point is a descendant of $e_d$ outside $F_P$ (i.e., occurs below $e_d$): $\extcov(e_1, e_d)$,
    
    \item The cost of edges with one end point in $F_P$ and the other end-point is an ancestor of $e_r$ outside $F_P$ (i.e., occurs above $e_r$): $\extcov(e_2, e_r)$, and
    
    \item The cost of edges that cover $P$ and both end-points outside $F_P$: $\extcov(P)$.
\end{itemize}

\begin{figure}[h!]
    \centering
    \includegraphics[scale =0.6]{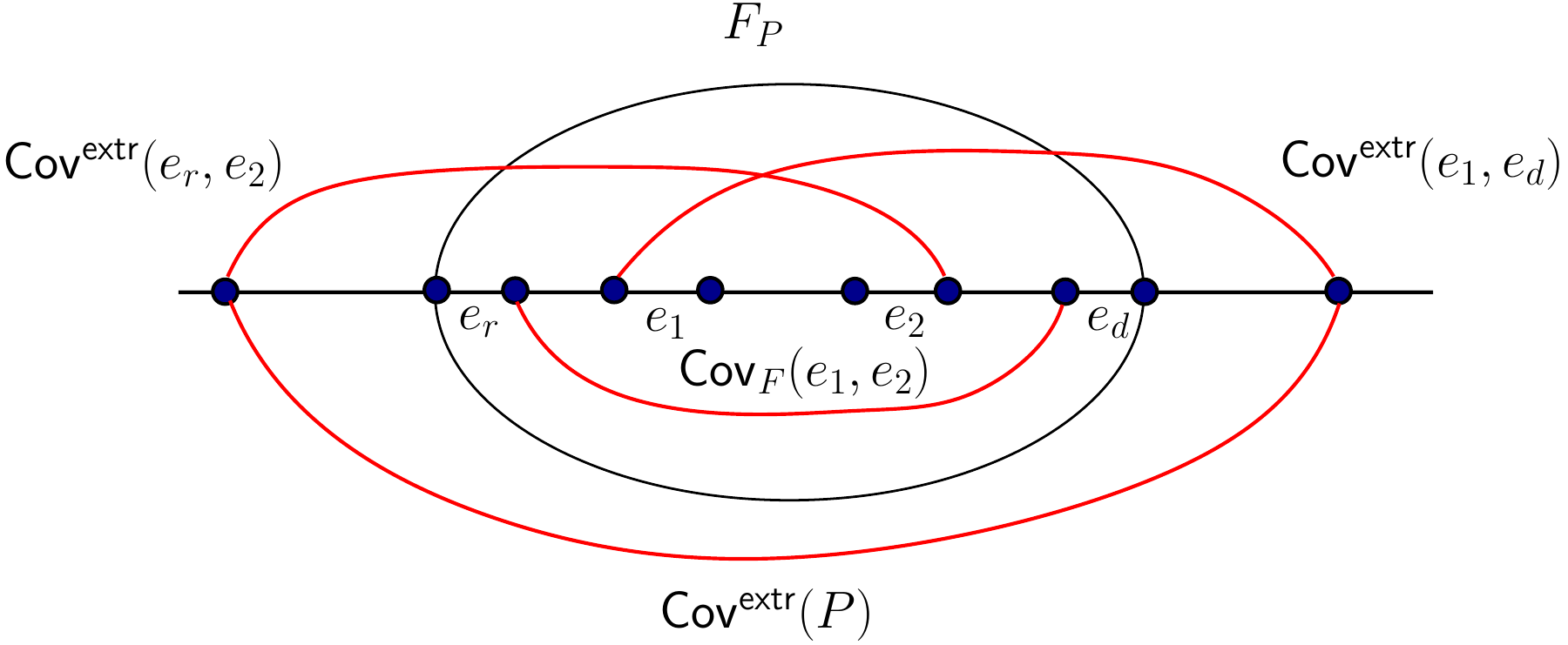}
    \caption{\small Example of both cut edges $e_1$ and $e_2$ in the same fragment highway in the fragment $F_P$.}
    \label{fig:both-edge-high}
\end{figure}


We make the following observation.

\begin{observation}
Let $e_1$ and $e_2$ be two edges on a fragment highway $P$ inside a fragment $F$. Then,
\begin{align*}
    \cov(e_1,e_2) = \cov_F(e_1, e_2) + \extcov(e_1, e_d) + \extcov(e_r, e_2) + \extcov(P).
\end{align*}
\end{observation}


\paragraph{Computing $\cov_F(e_1,e_2)$.} 
\begin{claim}
Let $P$ be the fragment highway of a fragment $F_P$. Using two aggregate computation in time $O(\dfrag)$ all edges $e \in P$ can compute the value $\cov_F(e, e')$ where $e'$ is any other edge in $P$. Moreover, this can be done in all fragments parallely.
\end{claim}

\begin{proof}
This is similar to the proof of Claim \ref{claim_non_highway_fragment}, but we have to perform both forward and reverse aggregate computation inside the fragment $F$ (as in Claim \ref{claim_pipeline_fragment}). Fix $e_1 \in P$. For any $e \in P$ below $e_1$, $e$ can find out the total weight of the edges that covers $e_1$ and $e$ and have both end-points in $F$ by one forward aggregate computation as follows: Any edge $x$ which has one end point in $F$ on a descendant vertex of $e$ can find out whether it covers $e_1$ and have other end-point inside $F$ by an LCA comparison (as in Claim \ref{claim_LCA_labels}). Hence, by a forward aggregate computation inside $F$, starting from the leaf nodes $T(F_P)$ (i.e., $T$ restricted to $F_P$) and aggregating the weight of such non-tree edges, every $e$ below $e_1$ can learn $\cov_F(e_1, e)$. The aggregate computation takes time $O(\dfrag)$. Moreover, this can be done for all $e_1$ in a pipe-lined fashion. As there are $O(\dfrag)$ many edges in $P$, this takes $O(\dfrag)$-rounds in total. For edges $e \in P$ above $e_1$, we do the exact same aggregate computation but in reverse direction on $T(F_P)$ (See Claim \ref{claim_pipeline_fragment}). This also takes $O(\dfrag)$ times. Moreover, as these are aggregate computations inside $F$, this can be done simultaneously in all fragments.
\end{proof}

\paragraph{Computing $\extcov(e, e_d)$ and $ \extcov(e_r, e)$.}

\begin{claim}
Let $P$ be the fragment highway of a fragment $F_P$. Also, let all vertices in $F_P$ know the identity of $e_r$ and $e_d$. Then, by two aggregate computation in time $O(\dfrag)$, each edge $e\in P$ know the value of $\extcov(e, e_d)$ and $\extcov(e, e_r)$. Moreover, this can be done in different fragments parallelly.
\end{claim}

\begin{proof}
 As in Claim \ref{claim_cov_eP_nh}, any edge $x$ which has one end point inside $F$ and covers $e \in P$ will know whether it covers the last edge of $P$. Similarly any edge $y$ which has one end point inside $F$ and covers $e$ will know whether it covers the first edge of $P$. Hence, by an aggregate computation inside fragment $F$, each edge $e$ will know the value of $\extcov(e, e_r)$ by aggregating all edges which covers $e$ and have one end-point below $e$ and the other end-point outside $F$ covering $e_r$ (can be found using LCA labels, see Claim \ref{claim_LCA_labels}). Similarly, using another aggregate computation inside $F$ in reverse direction (see Claim \ref{claim_pipeline_fragment}), each edge $e$ will know the value of $\extcov(e, e_d)$. By Claim \ref{claim_pipeline_fragment}, this can be done in time $O(\dfrag)$. As these are aggregate computations inside $F$, this can be done simultaneously in all fragments.
 \end{proof}

\paragraph{Computing $\extcov(P)$.}
\begin{claim}
Let $P$ be the fragment highway of $F_P$. Using one aggregate and one broadcast computation over a BFS tree of $G$, every vertex will know the value of $\extcov(P)$. For $k$ many fragments, this can be done in time $O(k + D)$.
\end{claim}

\begin{proof}
This can be done similar to Claim \ref{clm:two-short-highway}. Every edge knows whether it covers both $e_r$ and $e_d$ or not, and also whether both of its endpoints outside $F_P$ or not. Now, by using one aggregate computation over a BFS tree of $G$, we can compute the total weight of such edges, and then by using the same BFS tree, this information can be broadcasted to all vertices. This takes $O(D)$ rounds of communication. For different fragments, this computation can be pipelined, and hence can be performed in $O(k + D)$ time.
\end{proof}

\paragraph{An algorithm for a highway path.} We now explain how we can compute $\cut(e, e')$ when both $e$ and $e'$ belong to a  fragment highway path $P$ of a fragment $F$.

\begin{lemma} \label{lem:both-edge-short-high}
Let $P$ be a fragment highway of a fragment $F$. Also, let every edge $e$ know the value of (i) $\cov(e)$, (ii) $\cov_F(e, e')$ for all other edges $e' \in P$, (iii) $\extcov(e, e_r)$ and $\extcov(e, e_d)$, and (iv) $\extcov(P)$. Then, in $O(\dfrag)$ rounds, each pair of edges $(e, e')$ in $P$ can compute the value of $\cut(e, e')$. Moreover, this can be done parallelly in all fragments $F$.
\end{lemma}

\begin{proof}
Fix a pair of such edges $(e, e')$ where $e$ is above $e'$. For $e$ and $e'$ to know the value of $\cut(e, e')$, $e$ needs to know the value of $\extcov(e', e_r)$ which is with $e'$ and $e'$ needs to know the value of $\extcov(e, e_d)$ which is with $e$. To this end, each edge $e$ broadcasts the values $\extcov(e, e_r)$ and $\extcov(e, e_d)$ inside the fragment $F$. This is, in total, $\tO(\dfrag)$ bits of information broadcasted over a BFS tree of depth $O(\dfrag)$ (because of the guarantee that the diameter of $F$ is $O(\dfrag)$), hence can be done in time $\tO(\dfrag)$ in pipelined fashion. As the broadcast is happening inside $F$, this can be done parallelly for all fragments $F$.
\end{proof}

\section{Missing proofs from Section \ref{sec:2_respecting}} \label{sec:8_appendix}

\subsection{Simple cases with non-highways} \label{sec:app_8_nh}

In this section, we prove the following.

\nhtwoedges*

We first show that it is efficient to broadcast information inside all non-highway paths in the same layer. 

\begin{claim} \label{claim_broadcast}
Fix a layer $j$, and assume that for each non-highway path $P$ in layer $j$ there are $k$ pieces of information of size $O(\log{n})$ where initially each one of them is known by some vertex in $T(P)$. In $O(\dfrag + k)$ rounds, all the vertices in $T(P)$ can learn all the $k$ pieces of information. In addition, this computation can be done in all non-highway paths $P$ in layer $j$ in parallel.
\end{claim}

\begin{proof}
Note that since $P$ is a non-highway path, the entire subtree $T(P)$ is in the fragment of $P$ and has diameter $\dfrag$.
To solve the task, we use pipelined upcast and broadcast in the subtree $T(P)$. First, we collect all $k$ pieces of information in the root $r_P$, and then broadcast them to the whole tree $T(P)$, using pipelining this takes $O(\dfrag + k)$ rounds.  Additionally, as the different trees $T(P)$ of paths $P$ in layer $j$ are edge disjoint by Observation \ref{obv:path_disjoint}, the computation can be done in parallel for all non-highway paths in layer $j$.  
\end{proof}

Additionally we make sure that all vertices in a fragment $F$ know all values $\{e,\cov(e)\}_{e \in F}$, this can be done in $O(\sfrag)$ time by broadcast in the fragment $F$. Note that the value $\cov(e)$ is initially known by $e$ from Claim \ref{claim_learn_cov}.

\begin{claim} \label{claim_Tp_info}
In $O(\sfrag)$ time, for all fragments $F$, all vertices in $F$ learn the values $\{e,\cov(e)\}_{e \in F}$.
\end{claim}

\remove{
We next make sure that all vertices in $T(P')$, learn the values $\{e,\cov(e)\}_{e \in P'}$.

\begin{claim} \label{claim_Tp_info}
In $\tilde{O}(\dfrag)$ time, for all non-highways $P'$, all vertices in $T(P')$ learn the values $\{e,\cov(e)\}_{e \in P'}$.
\end{claim}

\begin{proof}
We work in $O(\log n)$ iterations corresponding to the layers. In iteration $i$, we take care of non-highways $P'$ in layer $i$.
We use Claim \ref{claim_broadcast} to let all vertices in $T(P')$ learn the values $\{e', \cov(e')\}_{e' \in P'}$. As the diameter of $P'$ is $O(\dfrag)$, this is $O(\dfrag)$ information, hence this computation takes $O(\dfrag)$ time by Claim \ref{claim_broadcast}. In addition, this can be done for all paths $P'$ in layer $i$ simultaneously. Note that all edges $e$ know the values $\cov(e)$ and their layer due to Claim \ref{claim_learn_cov} and Lemma \ref{lemma:layer_decomposition_inside_fragment_nonhighway}, hence the information $\{e', \cov(e')\}$ is initially known by the edge $e' \in P'$, as required for using the claim. The time complexity for all iterations is $\tilde{O}(\dfrag)$.
\end{proof}
}

Claim \ref{claim_Tp_info} in particular makes sure that all vertices in $T(P')$ for a non-highway $P'$ in layer $i$ know the complete structure of $P'$, they just look at the unique tree path of layer $i$ with edges above them in the fragment, if exists.
We now prove Claim \ref{claim_nh_2respecting}.

\begin{proof}[Proof of Claim \ref{claim_nh_2respecting}]
We work in $O(\log{n})$ iterations according to the layers. We next fix an iteration $i$ and a non-highway path $P'$ in layer $i$, all the computations we discuss can be done in parallel for all non-highways in the same layer.

We first use Claim \ref{claim_non_highway_fragment} to compare all edges of $P'$ to all edges above them or orthogonal to them in the same fragment in $O(\sfrag)$ time. Note that eventually we apply this computation for all non-highways, which guarantees that for each pair of edges $e,e'$ in the same fragment where at least one of them is a non-highway, one of the edges learns $\cut(e,e')$. If they are orthogonal both learn the cut value, and if one is above the other, the lower edge learns it.

Next we use Claim \ref{claim_orthogonal_nh} to compare $P'$ to non-highways in other fragments that $P'$ is potentially interested in. First, we know that there are only $O(\log{n})$ fragments $F$ such that $P'$ is potentially interested in the fragment $F$, and all vertices in $T(P')$ know these fragments by Corollary \ref{corol:-layering_interested}. Additionally, if $P'$ is potentially interested in a non-highway $P$ in the fragment $F$ there must be an edge between $T^{\downarrow}(P')$ and $F$ from Claim \ref{claim_edge}. Using aggregate computations in $T(P')$ we can find such edges $f$ to all relevant fragments $F$ in $O(\dfrag + \log{n})$ time. 
We also know that all vertices in the fragment $F$ know all the values $\{e,\cov(e)\}_{e \in F}$. Hence, all the requirements of Claim \ref{claim_orthogonal_nh} are satisfied, and we can use it to let all edges $e' \in P'$ learn all the values $\cut(e',e)$ for all $e \in F$ in $O(\sfrag)$ time, for a fragment $F$ that has a non-highway that $P'$ is potentially interested in. Doing so for all fragments $F$ that $P'$ is potentially interested in takes $\tilde{O}(\sfrag)$ time.

All the claims we use work also if we work in all the paths of the same layer simultaneously. Since we have $O(\log{n})$ layers, we compute all cut values in $\tilde{O}(\sfrag)$ time. 
Finally, we can let all vertices in the graph learn the values $\{e,e',\cut(e,e')\}$ for the min 2-respecting cut we found in $O(D)$ time using minimum computation on a BFS tree. Since we compared all non-highways that are potentially interested in each other, as well as all pairs of edges in the same fragment where at least one edge is a non-highway, the claim follows. Note that in our computation we also computed some additional cuts (since we compare $P'$ to the complete fragment $F$ and not just to a specific non-highway), which can only decrease the value of the min 2-respecting cut we find.
\end{proof}

\subsection{Proofs for non-highway highway case}

\clmshortnhlonghighedge*

\begin{proof}
We start by fixing a non-highway path $P'$ of layer $j$ and a long path composed of highways $P_H$ that $P'$ is potentially interested in, and describe the computation needed to compare $P'$ and $P_H$. Later, we explain how to do many such computations in parallel. Note that from Corollary \ref{corol:-_layer_highway}, and by iterating over all layers of the skeleton tree, one can deduce that each non-highway path $P'$ is only potentially interested in $O(\log^2{n})$ 
such paths $P_H$, that is because each edge keeps bough highway paths in its set of potentially interested paths (Namely, $P'$ knows the lowest fragment of each of these paths). Also, from Lemma \ref{lemma:parse_paths}, all vertices in $T(P')$ know exactly the identity of all such paths $P_H$. Also, from the structure of the decomposition, any path $P_H$ composed of highways, that does not contain the highway in the fragment of $P'$ is either completely orthogonal to $P'$ or completely above $P'$.

Let $P_1,...,P_k$ be the different highways in $P_H$ going from the lowest to highest in the tree.
We use Lemma \ref{lemma_partitioning} to break the edges of $P'$ to subsets $E'_1,...,E'_k$, such that it is enough to solve the problems defined by the pairs $(P_i, E'_i)$, i.e., compare only the edges $E'_i$ to $P_i$ (see Lemma \ref{lemma_partitioning} for the exact statement). To do so, all vertices should know the values $\{e,\cov(e)\}$, for all edges $e$ that are highest or lowest in some highway, this can be obtained in $O(D + \nfrag)$ time using Claim \ref{claim_high_low}. Applying Lemma \ref{lemma_partitioning} takes $O(\dfrag + k)$ time, and the computation can be done in parallel for different non-highways in layer $i$. The Lemma guarantees that $\sum_{i=1}^k |E'_i| = O(\dfrag + k).$

Now we would like to compare $E'_i$ to $P_i$ for all $1 \leq i \leq k$. If there is no edge between $T(P'^{\downarrow})$ and $P_i$, the minimum 2-respecting cut with one edge in $P'$ and one edge in $P_i$ was already computed in Claim \ref{claim_cut_nh_h_no_edge}, so we only need to take care of highways $P_i$ such that there is an edge $f$ between $T(P'^{\downarrow})$ and $F_{P_i}$. Also, all vertices in $T(P')$ already know about the edge $f$ from Claim \ref{claim_learn_edge}. Moreover, we only need to consider highways $P_i$ that are potentially interested in the path $P'$. From Corollary \ref{corol:-highway_non_highway} and Lemma \ref{lemma:parse_paths}, each highway $P_i$ is only potentially interested in $O(\log n)$ non-highways $P'$ which are in different fragments, and all the vertices in the fragment of $P_i$ know the list of fragments that contain non-highway paths that $P_i$ is potentially interested in.  Hence, by communicating on the edge $f$ we can learn whether $P_i$ is potentially interested in the fragment of $P'$. Note that this is a different edge for different pairs $P',P_i$, as the subtrees $T(P'^{\downarrow})$ are disjoint, and also the different fragments are disjoint.

For the highways $P_i$ left after the above discussion, we use Lemma \ref{lemma_paths_nh_h} to compare $E'_i$ to $P_i$. For this, all vertices in $T(P')$ should learn the values $\{\cov(e'),\extcov(e',P_i)\}_{e' \in E'_i}$ for all $1 \leq i \leq k$. To do so, we first let all edges $e' \in P'$ learn the values $\extcov(e',P)$ for all highways $P$, this takes $O(\dfrag + \nfrag)$ time using Claim \ref{claim_cov_eP_nh}, and can be done in all non-highways in the same layer simultaneously. Then, the information $\{\cov(e'),\extcov(e',P_i)\}_{e' \in E'_i}$ is known to the edge $e'$. To let all vertices in $T(P')$ learn it we use Claim \ref{claim_broadcast}. As we have $\sum_{i=1}^k |E'_i| = O(\dfrag + k)$, this takes $O(\dfrag + k)$ time. From Lemma \ref{lemma_partitioning}, we also have that all vertices in $T(P')$ know the identity of all edges in the sets $E'_i$. We next discuss the information known in $P_i$. First, using upcast and broadcast in the fragment $F_{P_i}$ of $P_i$, we can make sure that all vertices in the fragment know all the values $\{\cov(e)\}_{e \in P_i}$, they can also learn the identity of the edge $f$ between $T(P'^{\downarrow})$ and $F_{P_i}$, as follows. As vertices in $T(P'^{\downarrow})$ know the identity of $f$, then $f$ has an endpoint that knows about it, and can inform the second endpoint in $F_{P_i}$. Then, the information can be broadcast in $F_{P_i}$. This is only done if $P_i$ is potentially interested in the fragment of $P'$, hence only for $O(\log n)$ different fragments
of non-highways $P'$ in layer $j$. This shows that vertices in $T(P')$ and $F_{P_i}$ have all the information needed for applying Lemma \ref{lemma_paths_nh_h}.

From Lemma \ref{lemma_paths_nh_h}, using $O(|E'_i|)$ aggregate and broadcast computations in $F_{P_i}$, each edge $e \in F_{P_i}$ would know the values $\cut(e',e)$ for all edges $e' \in E'_i$. As $P_i$ is potentially interested in non-highways in $O(\log n)$ 
different fragments, participating in all computations takes at most $\tilde{O}(\sfrag)$ time. 
Moreover, the computation was inside $F_{P_i}$, hence we can work in parallel in different fragments. This allows comparing $E'_i$ to $P_i$ for all $1 \leq i \leq k$ in parallel. This concludes the description of comparing $P'$ to $P_H$. To compare $P'$ to all $O(\log^2 n)$ long highway paths $P_H$ that $P'$ is potentially interested in, we repeat this computation $O(\log^2 n)$ times (the partitioning requires $O(\dfrag + k)$ time for each of these computations separately, the parts within different highways can be done in parallel). Also, as discussed throughout, this computation can be done for all non-highways in layer $j$ in parallel (this results in $\tilde{O}(\sfrag)$ time inside each highway as explained above). To take care of non-highways of all layers, we have $O(\log n)$ such iterations. The overall complexity is $\tilde{O}(D + \sfrag + \nfrag)$.

At the end of the computation, for each pair of a non-highway $P'$ and a highway $P$ that have an edge between them, and are also potentially interested in each other, we have a vertex that knows the values $e',e,\cut(e',e)$ for $e' \in P', e \in P$ that minimize $\cut(e',e).$ To learn the minimum value over all such pairs, we use broadcast and convergecast in a BFS tree which takes $O(D)$ time.
\end{proof}
\end{document}